\documentclass[11pt]{article}
\usepackage{fullpage}
\usepackage[letterpaper, margin=0.97in]{geometry}
\usepackage[utf8]{inputenc}
\usepackage{bibentry}
\usepackage{amsmath}
\usepackage{amsfonts}
\usepackage{amsthm}
\usepackage{thmtools}
\usepackage{appendix}
\usepackage{color}
\usepackage{algorithm}
\usepackage{algorithmic}
\usepackage{amssymb}
\usepackage[pagebackref]{hyperref}
\usepackage{setspace}
\usepackage{xspace}
\usepackage{enumerate}
\usepackage{graphicx}
\usepackage{enumitem}
\usepackage{bbm}
\usepackage{mathtools}
\usepackage{tikz}
\usetikzlibrary{positioning, arrows.meta}

\newcommand{\mbbone}{{\mathbbm{1}}}

\newcommand*{\rank}{{\textnormal{rank}}}
\newcommand*{\mbbr}{{\mathbb R}}
\newcommand*{\mbbc}{{\mathbb C}}
\newcommand*{\mbbe}{{\mathbb E}}
\newcommand*{\mbbf}{{\mathbb F}}
\newcommand*{\msfx}{{\mathsf X}}
\newcommand*{\msfm}{{\mathsf M}}
\newcommand*{\msfa}{{\mathsf A}}

\newcommand*{\hmin}{{H_{\textnormal{min}}}}
\newcommand*{\mbold}{{\mathbf M}}
\newcommand*{\jbold}{{\mathbf J}}
\newcommand*{\abold}{{\mathbf A}}
\newcommand*{\bbold}{{\mathbf B}}
\newcommand*{\ubold}{{\mathbf U}}
\newcommand*{\vbold}{{\mathbf V}}
\newcommand*{\sigmabold}{{\boldsymbol{\Sigma}}}

\newcommand*{\ybold}{{\mathbf y}}
\newcommand*{\xbold}{{\mathbf x}}
\newcommand*{\ibold}{{\mathbf I}}
\newcommand*{\vlbold}{{\mathbf v}}
\newcommand*{\clbold}{{\mathbf c}}
\newcommand*{\ulbold}{{\mathbf u}}

\newcommand*{\im}{{\textnormal{Im}}}

\newcommand*{\tout}{{\textnormal{out}}}

\newcommand*{\mboldu}{{\mathbf M_{\mathbf U}}}
\newcommand*{\muxn}{{\mu_{\mathcal X^n}}}
\newcommand*{\mcale}{{\mathcal E}}
\newcommand*{\mcala}{{\mathcal A}}
\newcommand*{\mcalc}{{\mathcal C}}
\newcommand*{\mcalq}{{\mathcal Q}}
\newcommand*{\mcalm}{{\mathcal M}}

\newcommand*{\tcore}{{\textnormal{core}}}

\newcommand*{\wlbold}{{\mathbf w}}
\newcommand*{\albold}{{\mathbf a}}
\newcommand*{\blbold}{{\mathbf b}}

\newcommand*{\phiklip}{{\|\phi_k\|_{\textnormal{Lip}}}}

\newcommand*{\rphiklip}{{\|\phi_k^{-1}\|_{\textnormal{Lip}}}}
\newcommand*{\rphiklipr}{{\|\phi_k^{-1}\|^{-1}_{\textnormal{Lip}}}}

\newcommand{\QwS}{query-with-sketch\xspace}

\newcommand\algorithmicprocedure{\textbf{procedure}}
\newcommand{\algorithmicendprocedure}{\algorithmicend\ \algorithmicprocedure}
\makeatletter
\newcommand\PROCEDURE[3][default]{%
  \ALC@it
  \algorithmicprocedure\ \textsc{#2}(#3)%
  \ALC@com{#1}%
  \begin{ALC@prc}%
}
\newcommand\ENDPROCEDURE{%
  \end{ALC@prc}%
  \ifthenelse{\boolean{ALC@noend}}{}{%
    \ALC@it\algorithmicendprocedure
  }%
}
\newenvironment{ALC@prc}{\begin{ALC@g}}{\end{ALC@g}}
\makeatother
\newcommand{\mcalx}{\mathcal{X}}
\newcommand{\mcaly}{\mathcal{Y}}

\newtheorem{problem}{Problem}[section]
\newtheorem{theorem}{Theorem}[section]
\newtheorem{lemma}[theorem]{Lemma}
\newtheorem{definition}[theorem]{Definition}
\newtheorem{corollary}[theorem]{Corollary} 
\newtheorem{fact}[theorem]{Fact}

\newtheorem{proposition}[theorem]{Proposition}
\newtheorem*{theorem*}{Theorem}
\newtheorem*{lemma*}{Lemma}

\title{Systematic Data Structure Lower Bounds via the Query-with-Sketch Model}

\date{}

\author{Sumegha Garg\\Rutgers Univerisity\\sumegha.garg@rutgers.edu \and Songhua He\\Rutgers University\\sh1511@scarletmail.rutgers.edu \and Yuanzhi Li\\Carnegie Mellon University\\yuanzhil@andrew.cmu.edu \and Periklis A. Papakonstantinou\\Rutgers University\\periklis.research@gmail.com \and Xin Yang\\Meta\\yx1992@cs.washington.edu}

\begin{document}

\maketitle

\begin{abstract}
We study data structure lower bounds for the \emph{Approximate Matrix Powering} (AMP) problem. Given a substochastic, symmetric matrix $\mathbf{M}\in\mathbb{R}^{n\times n}$ and parameters $k$ and $\alpha$, the goal is to preprocess $\mathbf{M}$ so as to answer entry queries $(u,v)\mapsto \mathbf{M}^{k}[u,v]$ up to additive error $1/n^{\alpha}$. We focus on AMP in the \emph{succinct and systematic} regime, in which the data structure stores $\mathbf{M}$ verbatim, uses an additional $r$ bits of redundancy, and must answer queries by probing only a small number of entries of $\mathbf{M}$.

Our main conceptual contribution is a general framework for proving probe--redundancy trade-offs for systematic data structures. We introduce the \emph{query-with-sketch} model and develop a min-entropy-based approach that lifts conditional min-entropy bounds in the absence of redundancy to probe lower bounds in the presence of redundancy. We then establish these min-entropy bounds using problem-specific analytic and algebraic tools, for the downstream applications to AMP and its variants. As a consequence, our results provide new unconditional evidence toward a conjecture of P\u{a}tra\c{s}cu and Roditty on the space required for constant-time set-disjointness queries~\cite{patrascu2010distance}.

\end{abstract}

\section{Introduction}

Unconditional lower bounds are among the clearest windows we have into the limits of computation.  In the data structure world this promise is especially tangible: computation is free, and the only scarce resources are \emph{information} -- how many bits one stores, and how many locations one probes. A rich line of work has explored the tradeoff between these two resources across a variety of models, including the standard cell-probe model~\cite{cell_yao}, as well as more restrictive frameworks such as succinct data structures~\cite{jacobson1988succinct} -- under which the space usage is close to the information-theoretic minimum. While the strongest known lower bounds on the number of probes per query are logarithmic in the cell probe model~\cite{larsen2012higher}, linear lower bounds are known for succinct data structures (e.g., \cite{gal2007cell}).
In this paper, we develop a new framework for proving such tradeoffs in the succinct data structure regime, under an additional restriction: the data structure stores the input in read-only memory, together with an auxiliary data structure of at most $r$ bits. Data structures of this form are known as \emph{systematic} data structures\footnote{It is well known that data structure lower bounds have deep connections to other areas of theoretical computer science, including rigidity bounds~\cite{dvir2019static, natarajan2020equivalence}, circuit lower bounds~\cite{viola2019lower, dvovrak2021data}, cryptography~\cite{golovnev2020data}, and fine-grained complexity~\cite{henzinger2015unifying, corrigan2019function}. Importantly, \cite{corrigan2019function} show that improving the state-of-the-art data structure lower bounds for the function inversion problem -- under systematic and non-adaptive data structures -- would also imply new circuit lower bounds in Valiant’s common-bits model~\cite{valiant1977graph, valiant1992boolean}.}. Lower bounds in this model have been studied extensively in prior work~\cite{miltersen1995data, gal2007cell, golynski2007optimal, golynski2007size, golynski2008redundancy, bringmann2013succinct, chakraborty2018tight}.

A guiding application of our lower-bound framework is the \emph{Approximate Matrix Powering} (AMP) problem. Given a substochastic, symmetric matrix $\mbold \in \mathbb{R}^{n \times n}$ and parameters $\alpha$ and $k$, the goal is to preprocess $\mbold$ so as to answer entry queries $(u,v) \mapsto \mbold^{k}[u,v]$ up to additive error $1/n^{\alpha}$. AMP is a fundamental computational task that arises in the analysis of the long-term behavior of Markov chains~\cite{norris1998markov}, in iterative methods from numerical linear algebra~\cite{golub2013matrix}, and in PageRank and related random-walk-based algorithms~\cite{brin1998anatomy}.
AMP subsumes, as special cases, the counting set intersection problem as well as its decision version, set disjointness. These set problems are central throughout algorithms and complexity theory (e.g.,~\cite{cohen2010fast, patrascu2010distance, goldstein2017conditional, goldstein2019hardness, kopelowitz2020towards}). In this work, we obtain (to our knowledge) the first lower bounds for set intersection and set disjointness in the systematic and static data structure model.

We model a static data structure problem as a function
$
f:\mathcal{X}^N \times \mathcal{Q} \to \mathcal{Y},
$
where \(\mathcal{X}^N\) denotes the space of possible inputs (data),
\(\mathcal{Q}\) denotes the set of queries, and \(f(X,\phi)\) is the answer
to query \(\phi \in \mathcal{Q}\) on input \(X \in \mathcal{X}^N\).
The computation proceeds in two phases. In the preprocessing phase, a
systematic data structure with unbounded computational power,
given access to the input \(X\), stores \(X\) verbatim together with
an additional \(r=o(N)\) bits of auxiliary information, called the
\emph{redundancy}. Subsequently, upon receiving a query
\(\phi \in \mathcal{Q}\), the data structure is given the redundancy bits
for free and must output \(f(X,\phi)\) while making only a bounded number
of probes to the input \(X\). After answering the query, the data
structure is allowed to update its redundancy bits\footnote{Previous works~\cite{gal2007cell,chakraborty2018tight} on succinct and
systematic data structures assume probe access to the redundancy bits,
and do not explicitly allow the redundancy to be updated after each
query. Consider the Prefix Sum problem over a bitstring $X \in \{0,1\}^N$, where the data structure should support querying the parity of a prefix $\bigoplus_{i=1}^k X_i$, given $k$. \cite{gal2007cell} proved that any succinct and systematic data structure with $r$ bits of redundancy requires $q=\Omega(N/(r+1))$ probes to answer parity queries. However, if we allow the data structure to update its redundancy between queries, this trade-off disappears. For the sequential query order $k=1, \dots, N$, a data structure can simply maintain the running parity, achieving constant costs even in the worst-case. This simple example demonstrates the separation between read-only and read-write access to redundancy.
}.

To prove general lower bounds for systematic data structures, we formalize and analyze a query-with-sketch model using a novel min-entropy lemma (Lemma \ref{lem:min-entr}). The key power of this lemma is that it shifts the quantifier in the analysis: instead of reasoning about the information gained from a sequence of $t$ probes conditioned on $r$ bits of redundancy, we reason about the entropy loss in output after the algorithm sees $t$ arbitrary input elements (without any redundancy). This reduction is the core of our framework, as it allows us to set aside the data structure algorithm itself and focus purely on the problem's inherent properties, and it can be summarized as follows.

\begin{theorem}[Lifting min-entropy to data structure lower bounds]
Fix a data structure problem $f:\mcalx^N\times \mcalq\rightarrow \mcaly$. Fix the redundancy parameter $r$. If there exists a sequence of $m$ queries $Q \in \mcalq^m$, and an input distribution $\mu$ over $\mcalx^N$ such that:

\begin{quote}
    For \emph{every} partial assignment $\sigma$ to at most $o(N)$ elements of the input, the min-entropy of the answers to these $m$ queries remains high. That is, the probability of \emph{any} single answer vector $\mathbf{y}$ is exponentially small:
    $$
    \max_{\mathbf{y} \in \mcaly^m} \Pr_{X \leftarrow \mu}\big[ (f(X,Q_1), \dots, f(X,Q_{m})) = \mathbf{y} \mid \sigma \big] \le 2^{-2r}.
    $$
\end{quote}
Then, any (possibly randomized) succinct and systematic data structure (with read-write access to $r$ bits of redundancy) that answers all queries in $Q$ correctly with $\ge 99/100$ probability,  must make $\Omega(N/m)$ probes per query on average. Here, a probe returns an element of $\mcalx$.
\end{theorem}

This theorem provides a unifying template for all of our results. For a given data structure problem, we partition the query set $\mathcal{Q}$ into a sequence of blocks, each containing $m$ queries. Let $t$ denote the amortized probe complexity of a systematic data structure $D$ for answering all the queries. To establish a tradeoff of the form $r \cdot t = \tilde{\Omega}(N)$, we will show that the problem satisfies the min-entropy condition with block size $m = \tilde{\Theta}(r)$. This approach yields an amortized lower bound even when $D$ knows the entire query sequence in advance and is allowed to adaptively update its redundancy after each query.
Note that the min-entropy condition must hold for \emph{every} partial assignment: even if a benevolent adversary reveals the $o(N)$ most informative input elements, knowing the $m$ queries, there must still remain sufficient entropy in the answers. This is a strong guarantee and is technically challenging to establish for our downstream applications to AMP. See Figure~\ref{fig:proof_strategy} for outline of our results.

\paragraph{Comparison with previous works.}
Previous lower bounds for succinct and systematic data structures have been
obtained for problems such as rank and fully indexable dictionaries
\cite{miltersen1995data, golynski2007optimal, golynski2007size, golynski2008redundancy}, succinct sampling from discrete
distributions \cite{bringmann2013succinct}, online matrix--vector multiplication \cite{chakraborty2018tight},
and substring search, prefix sum, polynomial evaluation, and related problems
\cite{gal2007cell}. In particular, Gal and Miltersen~\cite{gal2007cell} established a
general lower bound of the form \((r+1)t \ge \Omega(N)\) for problems
satisfying a list-decoding property: namely, that the entire list of answers
can be recovered from any sufficiently large subset. While this property holds
for problems with strong global algebraic structure---such as polynomial
evaluation, which serves as the main example in~\cite{gal2007cell}---it fails for
many natural problems in which answers are more local and largely independent,
including the approximate matrix powering problem studied in this work.

Our approach departs from these reconstruction-based arguments. Many previous
bounds, including those of~\cite{gal2007cell,chakraborty2018tight}, proceed by encoding information
about the input from the behavior of a low-probe data structure on a carefully
chosen sequence of queries, and then compare the resulting encoding length with
an entropy lower bound for the input distribution. Such arguments naturally
lead to inequalities on the total transcript length, and hence to
probe--redundancy tradeoffs after the query sequence is optimized. By contrast,
our query-with-sketch reduction does not attempt to reconstruct the input.
Instead, it reduces lower bounds to showing that the answer vector to a block
of queries retains large conditional min-entropy even after an arbitrary small
set of input locations has been revealed. Thus the information-theoretic object
in our framework is the residual uncertainty of the answers under partial
information, rather than the entropy of the input itself.


\subsection{Applications to Approximate Matrix Powering}\label{sec:introamp}
In this section, we detail our results for the fundamental problem of Approximate Matrix Powering. 
\begin{figure}[ht]
\centering
\begin{tikzpicture}[
    scale=0.8, transform shape,
    node distance=2cm and 0.5cm,
    box/.style={rectangle, draw, thick, text width=11em, align=center, minimum height=3em},
    tech/.style={font=\footnotesize, align=center, above, midway},
    reduct/.style={font=\footnotesize, align=center, midway, right, xshift=2mm},
    arrow/.style={-Latex, thick}
]

\node (ds) [box] {Target data structure problem (AMP on $n\times n$-sized matrix)};
\node (qws) [box, below=of ds, yshift=0.5cm] {Query-with-sketch lower bounds};
\node (me) [box, below=of qws, yshift=0.5cm] {High conditional min-entropy};

\draw [arrow] (ds) -- (qws) node [right, midway, xshift=1cm] {Lemma~\ref{lem:reduction_ds_qws} (reduction)};
\draw [arrow] (qws) -- (me) node [right, midway, xshift=1cm] {Lemma~\ref{lem:min-entr} (min-entropy lemma)};

\node (amp1) [box, below=of me, xshift=-5.5cm] {AMP ($r=O(n\log n)$) (Thm~\ref{thm:amp_lb_small_s})};
\node (amp2) [box, below=of me] {AMP ($r=\Omega(n\log n)$) (Thm~\ref{thm:amp_lb_large_s})};
\node (amp3) [box, below=of me, xshift=5.5cm] {Set disjointness (Thm~\ref{thm:sd_lb})};

\draw [arrow] (me.south) -- (amp1.north) node [tech, midway, left, xshift=-2mm] {singular value bounds \\ on Jacobian};
\draw [arrow] (me.south) -- (amp2.north) node [tech, midway, above, yshift=-4mm] {rank-based min-entropy};
\draw [arrow] (me.south) -- (amp3.north) node [tech, midway, right, xshift=8mm] {combinatorial \\ ``forcing'' argument};

\node (setint) [box, below=of amp2, yshift=0.5cm] {Counting set intersection (Cor~\ref{thm:csi_lb})};
\node (apsp) [box, below=of amp3, yshift=0.5cm] {(2-$\varepsilon$)-approx. APSP (Cor~\ref{thm:apsp_lb})};

\draw [arrow] (amp2) -- (setint) node [reduct] {reduction};
\draw [arrow] (amp3) -- (apsp) node [reduct] {reduction};

\end{tikzpicture}
\caption{Our proof strategy. We reduce data structure bounds to a query-with-sketch bound (Lemma~\ref{lem:reduction_ds_qws}), which we then reduce to an information-theoretic min-entropy question (Lemma~\ref{lem:min-entr}). We solve the latter by applying a different analytical technique for each problem.}
\label{fig:proof_strategy}
\end{figure}
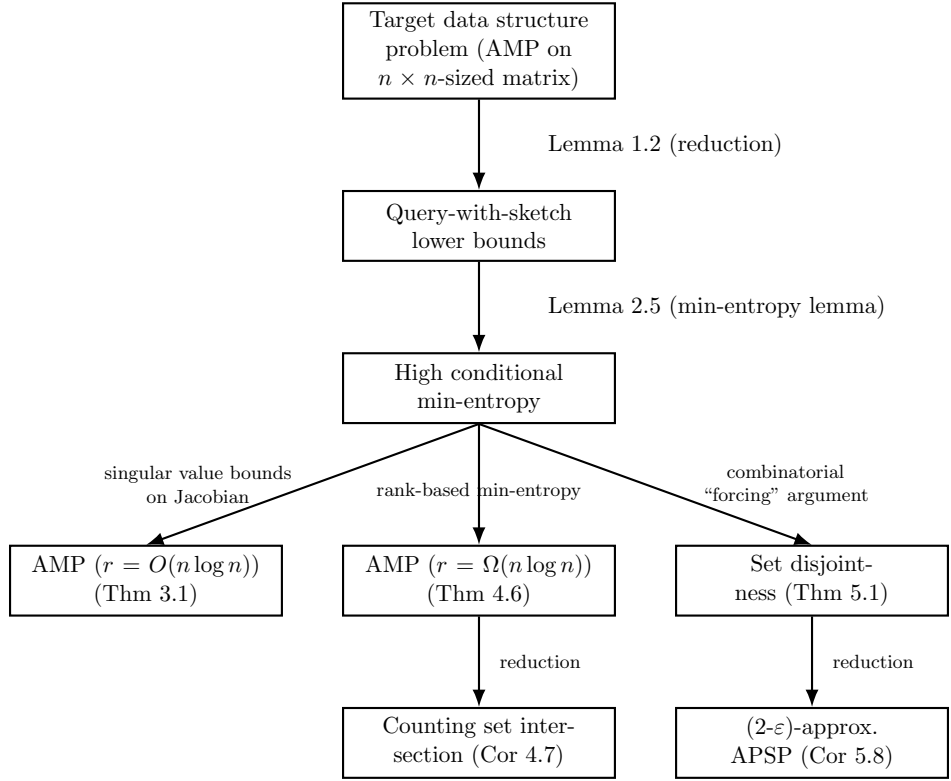
\begin{problem}[Approximate Matrix Powering (AMP)]
\label{prob:amp}
Fix parameters $\alpha > 0$ (a large enough constant) and an integer $k$ ($2 \le k \le n$). The data structure must solve the following task:
\begin{itemize}
    \item \textbf{Input:} A substochastic and symmetric matrix $\mbold \in \mbbr^{n \times n}$.
    \item \textbf{Query:} A pair of indices $(u,v) \in [n] \times [n]$.
    \item \textbf{Output:} An estimation $\hat{y} \in \mbbr$ that approximates the corresponding entry of the $k$-th matrix power within a specified additive error:
    $$
    \big|\hat{y} - \mbold^k[u,v]\big| \le \frac{1}{n^\alpha}
    $$
\end{itemize}
\end{problem}

This problem is closely related to the online matrix-vector multiplication (OMV) problem studied by~\cite{chakraborty2018tight}. That work establishes a tradeoff of $r \cdot t = \Omega(n^{3})$ for any succinct or systematic data structure with $r$ bits of redundancy and amortized probe complexity $t$, for answering Boolean matrix-vector multiplication queries.
We emphasize that OMV and AMP are incomparable problems. On the one hand, AMP can be viewed as potentially harder than OMV, since it requires computing entries of higher powers of the matrix. On the other hand, AMP is easier in the following aspect: in OMV, the vector is part of the query and is not stored in the data structure, leading to $2^{n}$ possible queries, whereas in AMP the queries range over matrix entries and there are only $n^{2}$ possible queries.

Our data structure lower bound for the AMP problem is as follows. Since AMP involves real-valued matrices, to avoid discussion on word size, we assume that each probe returns an entry of the matrix.
\begin{theorem*}[Informal statement for Theorem~\ref{thm:amp_lb_small_s} and Theorem~\ref{thm:amp_lb_large_s}]
    Fix parameters $\alpha,k$. Any (possibly randomized) succinct and systematic data structure of $r$ bits of redundancy that supports answering approximated elements of the $k$-th power of a substochastic and symmetric matrix $\mbold$ should make $\Omega(n^2/r)$ probes per query on average. This holds for $r\in[ \Omega(n\log n),O(n^2)]$ for $k=2$, and $r\in[\Omega(\log ^2n),O(n^2)]$ for every $k\in[3,n]$.
\end{theorem*}

Since computing elements of a matrix square is a special case of matrix powering, can always be performed using $O(n)$ probes, the above lower bound holds for a full range of probe complexity. Interestingly, the cases of sublinear-in-$n$ (Theorem~\ref{thm:amp_lb_small_s}) and superlinear-in-$n$ (Theorem~\ref{thm:amp_lb_large_s}) redundancy use different hard  input distributions and proof techniques, where either does not seem to extend to the full range. We give a proof overview in the next subsection.

For general \(k\), the naive algorithms give only
straightforward upper bounds obtained by reading enough entries of \(M\) to
compute the requested entry.  We are not aware of nontrivial succinct
systematic data structures for AMP that match the above probe--redundancy
tradeoff across the full parameter range.  Thus our result should be viewed as
a lower bound establishing a new barrier for this problem, rather than as a
tight characterization.

We next state our results for data structure problems on sets and graphs.

\paragraph{Set Intersection.} First, our proof of Theorem~\ref{thm:amp_lb_large_s} (AMP with superlinear-in-$n$ redundancy) also yields an identical lower bound for the set intersection problem and its counting variant. Set intersection is a fundamental data structure problem, which has received significant attention in fine-grained complexity~\cite{cohen2010fast, patrascu2010distance, goldstein2017conditional, goldstein2019hardness, kopelowitz2020towards}. In this problem, the input consists of $n$ subsets of the universe $[n]$, and the goal is to answer queries asking for the size of the intersection between the $u$-th and $v$-th sets. We show that, under the same parameter regime, answering each query requires $\Omega(n^{2}/r)$ probes on average. See Corollary~\ref{thm:csi_lb} for the full statement.

To the best of our knowledge, data structure lower bounds for the set intersection problem were previously known only for restricted data structures that are in comparable to the succinct and systematic setting. For instance, in the pointer machine model \cite{tarjan1979class}, which disallows random access and forces data to be accessed by traversing pointers stored in cells, \cite{afshani2016data} showed tradeoff of $S(N)Q(N)=\Omega(N^{2-o(1)})$, where $N$ is the total size of the input sets, $S(N)$ is the total space of the data structure, and $Q(N)$ is the probe cost in pointers followed.


\paragraph{Set disjointness.}  In 2010, P\u{a}tra\c{s}cu and Roditty~\cite{patrascu2010distance} conjectured that any cell-probe data structure which preprocesses $n$ sets $A_1,\dots,A_n \subseteq [\mathrm{polylog}(n)]$ and answers set disjointness queries, that is, whether the $u$-th and $v$-th sets intersect, using constant probe complexity must use $\tilde{\Omega}(n^2)$ space.
We study a related version of the set disjointness problem in which the sets $A_1,\dots,A_n$ are subsets of $[n]$, and make the first progress toward this conjecture in the systematic data structure setting. In particular, we prove a time–space tradeoff of the form
$
r \cdot t = \Omega\!\left(\frac{n^2}{\log^3 n}\right),
$
for any succinct and systematic data structure solving set disjointness.


\begin{theorem*}[Informal statement for Theorem~\ref{thm:sd_lb}]
Any (possibly randomized) succinct and systematic data structure using $r$ bits of redundancy, that answers set disjointness queries, should make $\Omega(n^2/(r\log^3n))$ probes on average per query. Here, $r\in[\Omega(n\log n),O(n\sqrt n/\log ^3 n)]$.
\end{theorem*}

We note that the range of $r$ is optimal for our hard input distribution, in which the total size of all sets is $O(n\sqrt{n})$. Moreover, the probe complexity is always at most $O(n)$, since a set disjointness query can be answered by probing only the two sets involved in the query. 

\paragraph{All Pair Shortest Paths (APSP).} By a reduction from set disjointness, we also obtain a data structure lower bound for $(2-\varepsilon)$-approximate all-pairs shortest paths (APSP). Specifically, any succinct and systematic data structure that answers a sequence of distance queries within a multiplicative factor strictly less than $2$, must either use $r$ bits of redundancy or incur an average probe complexity of $\Omega(n^{2}/(r \log^{3} n))$. This tradeoff holds for $r \in [\Omega(n \log n),\, O(n \sqrt{n} / \log^{3} n)]$. See Corollary~\ref{thm:apsp_lb} for the full statement.

Data structures for answering approximate distance queries are commonly referred to as \emph{distance oracles}. Thorup and Zwick~\cite{thorup2005approximate} showed that there exists a distance oracle with $O(n^{3/2})$ space that answers distance queries within a factor of $3$ in constant time. Subsequently, P\u{a}tra\c{s}cu and Roditty~\cite{patrascu2010distance} constructed an $O(n^{5/3})$-space data structure that returns an estimate of at most $2d+1$ in constant time, where $d$ is the true distance. On the lower-bound side,~\cite{patrascu2010distance} established a conditional result showing that any cell-probe data structure achieving a $2$-approximation in constant time must use $\tilde{\Omega}(n^{2})$ space.
Our result provides the first \emph{unconditional} time-space tradeoff for approximate APSP. Moreover, the tradeoff is tight: when $r = \Omega(n^{3/2})$, the distance oracles from prior work can be stored entirely within the redundancy.

\paragraph{An LLM-generated proof for the full redundancy range.} After the CCC version of this paper, we prompted GPT-5.5 pro to extend our set disjointness lower bound to a full range of redundancy. We include the resulting LLM-generated proof in Appendix~\ref{app:sd_full_range}. It establishes a probe lower bound of $\Omega(n^2/r)$ for every $r\in[n,n^2/1024]$, removing the $\log^3 n$ loss and extending the redundancy range to a constant fraction of $n^2$ (Theorem~\ref{thm:sd_lb_full_range}). The same reduction as in Section~\ref{sec:set disjointness} also gives the corresponding lower bound for $(2-\varepsilon)$-approximate APSP (Corollary~\ref{cor:apsp_lb_full_range}).

In particular, even when $r=n^2/\omega(1)$ bits of redundancy can be read and updated for free, answering set disjointness queries still requires $\omega(1)$ input probes on average. This confirms the linear-universe analogue of the P\u{a}tra\c{s}cu--Roditty conjecture~\cite{patrascu2010distance}: the universe here has size $n$, while their original conjecture concerns a polylogarithmic universe. Indeed, a cell-probe data structure with $o(n^2)$ bits of total memory could be simulated by storing its entire memory in the redundancy. All its memory accesses would then be free, giving a systematic data structure that uses no input probes, contradicting our lower bound. However, this argument does not resolve the original conjecture. With a polylogarithmic universe, the entire input fits in $\tilde O(n)$ bits of redundancy, allowing every query to be answered without any input probes once this much redundancy is available. The original conjecture concerns the cost of accessing this stored representation, which our model treats as free.

\subsection{Proof Overview and Technical Contributions}
In this section, we outline the proofs of our main results. We introduce the \emph{query-with-sketch} model, which captures data structures that store a short sketch of the input and answer a fixed batch (or sequence) of queries using limited probes to the raw input. We then show a general \emph{min-entropy lemma} that reduces lower bounds in this model to a conditional min-entropy bound. Finally, we instantiate this template for AMP using different analytic tools for different parameter regimes, obtaining strong probe--redundancy trade-offs.

\subsubsection{The query-with-sketch model and the min-entropy lemma}

First, we introduce the query-with-sketch model. The final data structure lower bounds consist of a reduction to this model, and a min-entropy lemma that transfers min-entropy upper bounds to query-with-sketch lower bounds. These definitions and reductions are formalized at Section~\ref{sec:qws}.

The query-with-sketch model can be thought of as a succinct and systematic data structure with its queries fixed beforehand\footnote{We adopt this name from the query complexity literature, where a ``query'' typically refers to a low-level read from the input. However, to maintain clarity and align with data structure conventions, this paper will consistently use ``probe'' for a low-level read from the input $X$ and ``query'' for the high-level request. Thus, what the query complexity literature calls ``query complexity'' corresponds directly to what we call ``probe complexity'' in this paper.}.
An algorithm in the query-with-sketch model aims to compute a multi-bit function $f:\{0,1\}^N\to\{0,1\}^M$ (or more generally, $f:\mcalx^N\to\mcaly$ ) on $x$ with the help of a short piece of information that we call ``sketch'', and limited number of coordinate probes to $x$. The sketch size should be smaller than $N$ and $M$ for the model to make sense. A more detailed definition and some notation are deferred to Section~\ref{sec:prel}.

\paragraph{Capabilities and limits of algorithms in the query-with-sketch.}
The query-with-sketch model, as a natural information-theoretic model, is of independent interest. We use the Hamming weight problem to illustrate that sketch can help reduce the query (probe) complexity in a non-trivial way.
In the Hamming weight problem, the input is a string from $\{0,1\}^n$ and the output is the number of $1$s. In the standard query model, $n$ probes are needed to compute, with high probability, the exact answer to this problem. But, in the \QwS,  given a small sketch (shorter than $\log n$ bits) of the lower-order digits in the hamming weight, it suffices to sample $n^\varepsilon$-many bits, $\varepsilon<1$, to compute the higher-order digits. The intuition is that the least significant digits are harder to compute than the more significant ones.

\paragraph{Min-entropy lemma.} All of our main results are built upon the following lemma, which provides a powerful connection between probe complexity and information. This min-entropy lemma formalizes a simple yet effective intuition: an $r$-bit sketch is useless if the answer would have remained more than $r$ bits uncertain, for any set of probes the algorithm could have made. 

It translates the task of proving a time-space tradeoff into a more manageable, information-theoretic question: how much uncertainty (min-entropy) about the answer $f(X)$ remains, even after we are given the values of any $q$ locations in the input $X$?

\begin{lemma*}[Informal statement of Lemma~\ref{lem:min-entr}]
    Let $f:\mathcal X^N\rightarrow \mathcal Y$ be a function over an adversarial distribution $\mu_{\mathcal X^N}$.
    Let $\Pi^S$ be a deterministic algorithm for $f$ using an $r$-bit sketch $S$.
    If for some integer $q$, the min-entropy of $f(X)$ remains $>r$ conditioned on any partial assignment on $X$ of length $q$, then $\Pi^S$ must make $>\Omega(q)$ probes on average, even given sketch $S$.
\end{lemma*}

Here is the high-level intuition of the min-entropy lemma. The proof (see Section~\ref{sec:min-entropy-theorem}) views any query-with-sketch algorithm as a collection of $2^r$ decision trees. The $r$-bit sketch $S(X)$ simply tells the algorithm which one of these trees to execute. A leaf at depth $q$ in any tree corresponds to the answers from $q$ specific probes.

Our lemma's condition says that for any such leaf, the final answer $f(X)$ is still highly ambiguous—there are more than $2^r$ possible correct answers consistent with that probe path. Since the $r$-bit sketch (which was fixed from the start) cannot possibly resolve this ambiguity, the algorithm cannot halt and must make more probes.

This lemma is powerful, but our main results require a technical strengthening. We show that the same lower bound holds even if the high min-entropy condition is only guaranteed to apply to a large, high-probability set of ``good inputs''. See Lemma~\ref{lem:min-entr} for the full, formal statement.

\paragraph{Connection to succinct and systematic data structures}

The query-with-sketch complexity was studied for the direct sum of parity problems \cite{nisan1998products, beigel1998one}, and was applied to obtain data structure lower bounds for the prefix sum problem \cite{gal2007cell}. However, the reduction in \cite{gal2007cell} was restricted to the prefix sum problem. We defer the discussion on this reduction to Section~\ref{sec:qws_reduction}. Below, we give a reduction that is generally applicable to many data structure problems.

\begin{lemma}[Reduction from data structures to query-with-sketch algorithms]
    \label{lem:reduction_ds_qws}
    Fix a data structure problem $f:\mcalx^N\times\mcalq\rightarrow \mcaly$. Suppose there is a (possibly randomized) succinct and systematic data structure $D$ with the following properties:
    \begin{itemize}
        \item It has a redundancy of $r$ bits that one can update for free.
        \item For a fixed sequence of queries $Q=(\kappa_1,\dots,\kappa_l)\in \mcalq^l$, it answers the entire sequence correctly with probability at least $99/100$.
        \item The total number of probes made to answer the entire sequence $Q$ is at most $tl$ in the worst case.
    \end{itemize}
    Then for any integer $m\in[l]$ that divides $l$, there exists a starting index $i\in\{1,m+1,2m+1,\dots,l-m+1\}$ and a (randomized) query-with-sketch algorithm $\Pi^S$ for answering the subsequence $(\kappa_i,\dots,\kappa_{i+m-1})$ that has:
    \begin{itemize}
        \item An $r$-bit sketch $S$.
        \item An average-case probe complexity of at most $tm$.
        \item A success probability of at least $99/100$.
    \end{itemize}
\end{lemma}

While our reduction applies to randomized data structures, our min-entropy lemma applies to deterministic query-with-sketch algorithms. This gap can be fixed by applying Yao's minimax principle, while we are studying the distributional complexity of data structures and query-with-sketch algorithms.

We note that our reduction applies to common variants of data structure problems in the literature, including those where the goal is to answer a single query (e.g., \cite{gal2007cell}) and those that answer, or are motivated by, a sequence of queries (e.g., \cite{chakraborty2018tight}).

Furthermore, the lemma's requirement for success on all queries simultaneously is not a strong restriction. An algorithm that only guarantees to answer each of the $l$ queries correctly with high probability (e.g., $\ge 1-o(1)$) can be converted to one that satisfies our stricter requirement. This is achieved via standard probability amplification, which increases the total probe complexity by a multiplicative factor of $O(\log l)$.


We also note that, as a general lower bounding tool, our approach inherently cannot prove time-space trade-offs stronger than $r\cdot t=\Omega(N\log |\mcaly|)$. A nontrivial query-with-sketch algorithm always has $r\le m\log|\mcaly|$, since $m\log|\mcaly|$ bits are sufficient to encode all the answers of the $m$ consecutive queries. In addition, a nontrivial query-with-sketch algorithm never makes more than $N$ probes to the input, which implies that $t\le N/m$ in our reduction. The two facts together imply that there always exists query-with-sketch algorithms achieving $r\cdot t\le N\log |\mcaly|$, for every $r\ge 1$.
To the best of our knowledge, among previous works \cite{miltersen2005lower,golynski2007size,golynski2007optimal,gal2007cell,golynski2008redundancy,bringmann2013succinct,chakraborty2018tight}, only a recent breakthrough \cite{chakraborty2018tight} overcame this barrier to achieve a stronger trade-off for the OuMv (online vector-matrix-vector multiplication) problem. They achieved it by showing that any efficient algorithm for OuMv would implicitly function as a compression scheme that violates Shannon's source coding theorem. While powerful, this argument is highly tailored to the algebraic properties of matrix multiplication, and the massive $2^n$ probe space. It does not provide a general template for breaking the $r\cdot t$ trade-off barrier.

\subsubsection{Approximate matrix powering with sublinear-in-\texorpdfstring{$n$}{n} redundancy}

\paragraph{Warm up: a lower bound to approximating the matrix square.}

Let us begin our proof outline for a restricted form of the full lower bound. Our lower bound holds for symmetric matrices and for arbitrary matrix powers $k>2$. In the proof outline we will only discuss the lower bound for squaring, $\mbold^2$, not-necessarily symmetric matrices $\mbold$. These are the only simplifying assumptions made for this high-level exposition. Specifically, in this proof overview we will show that every succinct and systematic data structure for approximating the matrix square $\mbold^2$ either requires $r=0.5n$ bits of redundancy or $t=\Omega(n)$ number of probes for every query. After that, we show how we extend it to an $r\cdot t=  \Omega(n^2)$ lower bound for general sublinear-in-$n$ $r$, for computing $\mbold^k$.

While we accept every possible output within an $1/n^\alpha$ additive error, we restrict the AMP problem to a version where the output is fixed, which enables us to apply the query-with-sketch reduction. Specifically, fix a large enough constant $\alpha'$, we let the output to be $[\mbold^k]_{\alpha'}$, a rounded matrix power where each element is rounded to its nearest multiple of $1/n^{\alpha'}$. We will show in Lemma~\ref{lem:round_negl} that the reduction to this rounded version will only incur negligible error probability, when $\alpha'$ is large enough compared to $\alpha$.

Our two-step reduction (Lemma~\ref{lem:reduction_ds_qws} and Lemma~\ref{lem:min-entr}) implies that, one only needs (i) a nemesis input distribution $\mcalm$; (ii) a sequence $Q$ of $n^2$ queries; and (iii) to show that for every $m=r$ consecutive queries from $Q$, even when conditioned on knowing any $0.001n^2$ elements from $\mbold$, the answers remain an $\Omega(r)$ min-entropy.

We construct a nemesis input distribution $\mcalm_{\beta,\gamma}$ of the form $\mbold=\frac{n^2-1}{n^2}\ibold+\frac1{n^\beta}\mboldu$, where $\mboldu$ is a random matrix where each element is i.i.d. (independent and identically distributed) uniformly random from the range $[\frac1n-\frac1{n^\gamma},\frac1n]$. For the reader's convenience, we may consider an instance where $\beta =2$ and $\gamma=4$. In the actual proof, $\mboldu$ is a uniformly random symmetric matrix, and a similar proof applies. In addition, we construct the sequence of queries $Q$ in a way that every $r$ consecutive queries span $r$ different rows and columns. To illustrate our technique in this warm-up, focus on an illustrative block of $r$ diagonal queries of approximating $\mbold^2[1,1],\mbold^2[2,2],\dots,\mbold^2[r,r]$.

Notice that $\mbold^2=(\frac{n^2-1}{n^2})^2\ibold+\frac2{n^2}\cdot \frac{n^2-1}{n^2}\mboldu+\frac1{n^4}\mboldu^2$. We shall assume that the elements $\mbold[1,1],\mbold[2,2],\dots,\mbold[r,r]$ are revealed to the algorithm for free, since this will only decrease the probe complexity. We will show that, however, the second-order term $\mboldu^2$ remains highly uncertain.

More formally, we use $\sigma$ to denote an arbitrary \emph{partial assignment} to $0.001n^2$ elements of $\mbold$.
Our reductions (the reduction from data structure lower bounds to query-with-sketch lower bounds, and the reduction to the min-entropy of the problem) help us get rid of analyzing the algorithms, and reduce it to a task of showing that the output has a high conditional min-entropy:

\begin{equation}
\label{eq:ams_probability}
\max_{\sigma,(y_i)_{i\in[r]}}\Pr_{\mbold\leftarrow\mcalm}[[\mbold^2[1,1]]_{\alpha'}=y_1,\dots,[\mbold^2[r,r]]_{\alpha'}=y_r|\sigma]<2^{-2r}.
\end{equation}

Our strategy is to select $\Omega(r)$ many input elements that are not fixed by $\sigma$, so that their \emph{Jacobian matrix} with the output elements is a diagonal matrix.
That is, (if given all the other $n-\Omega(r)$ elements fixed), each output element is an increasing function of its corresponding input element, but does not depend on other elements. We call these selected elements \emph{core} elements. We can achieve it because each element of the matrix square only depends on input elements of the same row or the same column.

More precisely, for each output element $\mbold^2[u,u]$, we try to find an input element $\mbold[u,i]$ such that $i>r$, and $\mbold[u,i]$ is not fixed by $\sigma$. That is, the element shares the same row or column with $\mbold^2[u,u]$, but does not share the same row or column with any other output element. We use $\jbold$ to denote its Jacobian matrix with the output elements, where each row corresponds to a core element $\mbold[u_i,i]$, each column corresponds to an output element $\mbold^2[u,u]$, and each $\jbold[i,u]:=\frac{\partial \mbold^2[u,u]}{\partial \mbold[u_i,i]}$ quantifies how the rate $\mbold^2[u,u]$ will increase if we perturb $\mbold[u_i,i]$ by a little bit. Intuitively, the higher the rate is, the less likely the output will evaluate to a fixed value, which yields a min-entropy bound.

$$\jbold=\begin{pmatrix}

    \Omega(1/n^3)&0&\dots&0\\

    0&\Omega(1/n^3)&&\vdots\\

    \vdots&&\ddots&\vdots\\

    0&\dots&\dots&\Omega(1/n^3)

    \end{pmatrix}$$

We can always select $r'=\Omega(r)$ core elements since a partial assignment $\sigma$ of length $0.001n^2$ can only fix a small proportion of elements at the first $r=0.5n$ rows and the first $r$ columns.

Given these core elements $\mbold[i_1,j_1],\dots, \mbold[i_{r'},j_{r'}]$, we may discard all the other elements of $\mbold$. Specifically, upper bounding the conditional probability in inequality (\ref{eq:ams_probability}) can be reduced to upper bound

\begin{equation}
\label{eq:ams_probability_sigmap}
\max_{\sigma',(y_i)_{i\in[r]}}\Pr_{\mbold\leftarrow\mcalm}[[\mbold^2[1,1]]_{\alpha'}=y_1,\dots,[\mbold^2[r,r]]_{\alpha'}=y_r|\sigma']
\end{equation}
where $\sigma'$ is any partial assignment for all the other elements of $\mbold$ except core elements, such that $\sigma'$ is consistent with $\sigma$. This is because the previous probability is a weighted sum of the latter one.\footnote{For a discrete probability space and two events $A,B$, $\Pr[A|B]=\sum_{B'}\Pr[A|B'\cap B]\cdot \Pr[B'|B]$, where the summation runs over a partition of $B$ space; and thus, $\Pr[A|B]\le \Pr[A|B']$, for some $B'\subseteq B$.}

The Jacobian matrix tells us that given $\sigma'$, the output elements are mutually independent. For each core element $\mbold[i,j]$, its corresponding output element $\mbold^2[u,u]$, and each possible output $y_u$, given that $\frac{\partial \mbold^2[u,u]}{\partial \mbold[i,j]}=\Omega(1/n^3)$, we have

$$\Pr_{\mbold\leftarrow \mcalm}[[\mbold^2[u,u]]_{\alpha'}=y_u|\sigma']\le O(1/n^{\alpha'+9})$$
which is obtained by dividing the output length ($1/n^{\alpha'}$) by the partial derivative ($\Omega(1/n^3)$) and the range length that the core element is uniformly selected from ($1/n^{\gamma+\beta}=1/n^6$).

Therefore, the probability in (\ref{eq:ams_probability_sigmap}) is upper bounded by $O((1/n^{\alpha'+9})^{r'})=O(1/2^{\Omega(r\log n)})$, which implies a desired $\Omega(r\log n)$ conditional min-entropy.

\paragraph{From $k=2$ to $k\ge 3$: singular value bounds to the Jacobian matrix.} While each element of the matrix square only depends on $2n-1$ elements of the input matrix, one cannot get a stronger probe lower bound for a smaller redundancy $r$. Therefore, we generalize the above arguments to $3\le k\le n$, where an $r\cdot t= \Omega(n^2)$ lower bound still holds.

We use the same input distribution $\mcalm$ and the sequence of queries $Q$ as above. When $k\ge 3$, the expansion $\mbold^k=\sum_{l=0}^k\binom kl(\frac{n^2-1}{n^2})^{k-l}(\frac1{n^\beta})^l\mboldu^l$ is still dominated by the terms where $t$ are small. Again, we fix and focus on a sequence of output elements $\mbold^k[1,1],\dots, \mbold^k[r,r]$ for some $r=o(n)$.

Given a partial assignment of length $\le 0.001n^2$, unlike the previous case, all the elements of the first $r$ rows and columns may be fixed by $\sigma$. However, we are still going to show that the output remains a high min-entropy. Our proof aligns with the proof outline above, but with the following edits:

First, we partition the set of input elements of $\mbold$ from the last $n-r$ columns and $n-r$ rows into disjoint subsets of size $r$, such that each pair of elements from the same subset does not share the same row or the same column. In this way, given a partial assignment $\sigma$ of length $\le 0.001n^2$, there will be at least one subset where $r'=\Omega(r)$ elements in it are not fixed by $\sigma$. We choose these elements as our core elements.

Then, denote by $(i_1,j_1),\dots, (i_{r'},j_{r'})$ the core elements. We look into the Jacobian matrix between the core elements and the output elements again. Since the $r$ output elements and  $r'$ core (input) elements span different rows and columns, their Jacobian matrix is of the following form for $K:=\binom k3\frac{(n^2-1)^{k-3}}{n^{2k-6+2\beta}}$.

    $$\jbold=\begin{pmatrix}

    K\mboldu[1,i_1]\mboldu[j_1,1]+O(k^5/n^{2+3\beta})&\dots&\dots&K\mboldu[r,i_1]\mboldu[j_1,r]+O(k^5/n^{2+3\beta})\\

    \vdots&\ddots&&\vdots\\

    \vdots&&\ddots&\vdots\\

    K\mboldu[1,i_{r'}]\mboldu[j_{r'},1]+O(k^5/n^{2+3\beta})&\dots&\dots&K\mboldu[r,i_{r'}]\mboldu[j_{r'},r]+O(k^5/n^{2+3\beta})

    \end{pmatrix}$$
This matrix $\jbold$ is a perturbed random matrix. The dominant term for each entry is a product of two random $\mboldu$ elements, while the $O(k^5/n^{2+3\beta})$ term is a small, deterministic-like perturbation. Critically, because all $r+r'$ input and output indices are distinct, each $\mboldu$ element appears at most once in the dominant terms of $\jbold$. This makes the dominant terms i.i.d. random variables.

We apply a concentration bound on its singular values \cite{wei2017upper} to show that the largest $0.9r'$ singular values of $\jbold$ are high with high probability. Intuitively, the singular values of a matrix measure the ``stretch'' it applies to its input vectors. A matrix with $\Omega(r')$ large singular values is robustly high-rank.

This high-rank property is the key to our min-entropy bound. We use it to bound the pre-image of a fixed output. If the map ``stretches'' the input space in $\Omega(r')$ dimensions, then to map to a tiny output region (our rounded answer box of side length $1/n^{\alpha'}$), the input must originate from an exponentially small-volume region. This small volume, under our uniform distribution $\mcalm$, translates directly to an exponentially small probability, which is the high conditional min-entropy we need.

By a union bound over the (polynomially many) possible choices for the core element set, we can define our ``good input set'' $\mcalm$ for which the conditional min-entropy can be bounded as the set where this high-rank property holds for all relevant Jacobians.

More specifically, given $\sigma'$, we show that if two core element assignments $\xbold,\xbold'$ map to the same output $\ybold$ (within precision $1/n^{\alpha'}$), they must be geometrically close, if a constant proportion of singular values of $\jbold$ are large.





In fact, for any fixed output $\ybold$, the set of all consistent core inputs $\xbold$ is \emph{concentrated} into an exponentially small-volume region. Under our uniform input distribution, this small volume implies an exponentially small probability, which confirms the high conditional min-entropy.

We omitted some details that will be clarified in the formal proof.

\subsubsection{Approximate matrix powering with superlinear-in-\texorpdfstring{$n$}{n} redundancy}

The above arguments do not easily extend to the case $r=\omega(n\log n)$, since the Jacobian matrix is no longer controlled when there are more than $n$ core elements and more than $n$ output elements. The singular value concentration bound \cite{wei2017upper} only works for random matrices whose elements are independent. To that end, we construct and study the following combinatorial distribution of matrices.

To better illustrate the distribution, we describe its transition graph $G=(V,E,W)$ instead, which is an undirected graph with edge weight to be its transition probability from both sides. Consider a vertex set $[n]$ that is partitioned to three parts $V_1,V_2,V_3$, each of size exactly $n/3$. For every pair of vertices from $V_1\times V_2$ or $V_2\times V_3$, we connect them by an edge of weight $1/n^3$ with independent probability $1/2$. In addition, we connect each vertex to itself by an edge of weight $1-1/n^2$. Note that as a substochastic matrix, we do not require each row sum of the matrix to be $1$, but less than or equal to $1$.

We construct our sequence of queries as the sequence of all pairs of vertices between $V_1$ and $V_3$. The sequence is constructed in a way that among every $r$ consecutive queries, each vertex is mentioned by $\Theta(r/n)$ times.

Intuitively, the graph is a $3$-layered graph with heavy lazy transitions. We will show in Fact~\ref{fact:symmetric_matrices} that approximating the matrix power of the transition graph reduces to counting the number of length-$2$ walks from each $u\in V_1$ to $v\in V_3$. Due to the heavy lazy transitions, the matrix power is dominated by the length-$2$ walks. In fact, we will show that even computing the \emph{parity} of the number of walks is hard.

Although the nemesis input distribution is combinatorial, the analysis remains linear-algebraic in spirit.
We will try to find a basis of the input to the output that preserves the min-entropy.

Given a partial assignment of length $\le 0.001n^2$, we remove those vertices whose $\ge 0.01n$ neighbors are revealed. For each of the remaining vertices $v$ from $V_3$, there are $O(r/n)$ queries on $v$. We denote by $(u_1,v),\dots,(u_m,v)\in V_1\times V_3$ the queries on $v$, for $m=O(r/n)$. We select $m$ corresponding vertices $w_1,\dots,w_m\in V_2$ such that none of the edges $\forall i\in[m],(w_i,v)$ are revealed by $\sigma$.

As above, we fix a partial assignment $\sigma'$ that reveals all the input edges other than the selected edges above, and is consistent with $\sigma$. Lower-bounding the min-entropy conditioned on $\sigma$ reduces to bounding the min-entropy conditioned on every such $\sigma'$. Given that the whole subgraph between $V_1$ and $V_2$ is revealed by $\sigma'$, queries on different vertices $v\in V_3$ are independent, and the min-entropy is the sum of min-entropy on each individual $v$.

Fix a vertex $v\in V_3$, denote by $\hat \mbold\in\mbbf_2^{m\times m}$ the indicator matrix of whether each $u_i\in V_1$ is connected to $w_j$. Denote by $\xbold\in\mbbf_2^m$ the indicator vector of whether each $(w_i,v)\in E$ or not, which is a uniformly random vector conditioned on $\sigma'$. The parity of the output (the number of length-$2$ paths between each $u_i$ and $v$) is exactly $\hat\mbold\xbold$. Using a classic result that a random binary matrix has a high rank with high probability (\cite{belsley1993rates}), we show that with high probability every $\hat \mbold$ with $m\ge 10\log n$ has a high rank, which implies a high min-entropy of the output.

\subsubsection{Lower bounds to \texorpdfstring{$(2-\varepsilon)$}{(2-eps)}-approximate APSP}

In APSP, we study the \emph{existence} instead of the parity of the number of length-$2$ paths between pairs of vertices. The graph distribution defined above is too dense, forcing the distance between vertices from $V_1$ and $V_3$ to be $2$ almost surely. To that end, we work on its sparse variant instead. We still consider a graph $G=(V,E)$ of three layers $V_1,V_2,V_3$ with $n/3$ vertices each. However, for each pair of vertices $(u,w)\in V_1\times V_2$ or $(w,v)\in V_2\times V_3$, we connect them with independent probability $1/\sqrt n\log n$. For each pair of vertices $(u,v)\in V_1\times V_3$, they have a length-$2$ path if and only if there exists a vertex $w\in V_2$ such that $(u,w),(w,v)\in E$, i.e., whether the neighbor sets of $u$ and $v$ intersect. Since the entropy of the graph is $ \Theta(n\sqrt n/\log n)$, we prove an $ \Omega(n^2/\log^3n)$ time-space trade-off for every $\Omega(n\log^2n)\le r\le O(n\sqrt n/\log^3n)$, where the range of $r$ is nearly tight.

A similar proof to the above can no longer apply here. While the output can still be formulated as a matrix-vector product of the form $\hat \mbold \xbold$, the arithmetic is over semirings (i.e., ANDs and ORs). There does not exist a natural notion of a ``basis'' in this setting.

Instead, we introduce a combinatorial \emph{forcing argument}. We show that for a sparse graph, a single ``disjoint'' (0) answer for a query $(u,v)$ \emph{forces} a large number of potential edges to be non-existent (0). The high min-entropy comes from two facts: (i) with high probability a large proportion ($1-1/\textnormal{polylog}(n)$) of answers are $0$ (Lemma~\ref{lem:sd_good_input_sigma}); and (ii) only with exponentially small probability that \emph{many} such forcing events (from a block of queries) can be simultaneously satisfied by a random graph.

We omitted details on how to apply the forcing argument to the input graph conditioned on a partial assignment $\sigma$, and how we deal with partial assignments that reveal too much information about the answer. We deal with these problems formally in Section~\ref{sec:set disjointness}.

\subsubsection{An LLM-generated proof for the full redundancy range}
\label{sec:overview_sd_full_range}

The LLM-generated proof in Appendix~\ref{app:sd_full_range} follows the same forcing argument as Section~\ref{sec:set disjointness}, but uses a different input distribution. Surprisingly, the proof is also shorter.
We still work on a graph $G=(V,E)$ of three layers $V_1,V_2,V_3$ with $n/3$ vertices each. For simplicity, we omit most vertices from $V_3$ but only focus on $b=\Theta(r/n)$ of its vertices $v$, each connects to a disjoint set $W_v$ of $n/3b$ vertices from $V_2$. For each pair of vertices $(u,w)\in V_1\times V_2$, we connect them independently with probability $1-2^{-b/n}$. This distribution is constructed in a way that for each pair of vertices $(u,v)\in V_1\times V_3$, they still have a length-2 path with $\Theta(1)$ probability. Besides, these $\Theta(r)$ answers are independent because they depend on disjoint input coordinates.

A similar forcing argument applies here. By concentration, with high probability a constant fraction of the independent answers are $0$. Fix any such answer vector. Each zero answer forces all edges between the corresponding vertex $u$ and $W_v$ to be absent. The key point is that different zero answers force disjoint sets of edges. Thus, even conditioned on a partial assignment $\sigma$ that fixes a sufficiently small part of the input, these constraints hold simultaneously with exponentially small probability. The probe lower bound then follows from our min-entropy lemma.

\section{Query-with-sketch model and the min-entropy lemma}
\label{sec:qws}

In this section we give the two abstract ingredients used throughout the
paper: the query-with-sketch model and a min-entropy lemma for proving lower
bounds in this model.  We then show how lower bounds in the query-with-sketch
model imply lower bounds for succinct and systematic data structures.

Because the query-with-sketch model will serve as the target of our lower bounds, we make its execution model explicit: the sketch is fixed from the input, subsequent probes may be chosen adaptively from the transcript, and the algorithm eventually halts with an output.

\subsection{Query-with-sketch model}
\label{sec:prel}

We define the query-with-sketch (QwS) model for a fixed function
 \(f:\mcalx^N\to \mcaly\); in contrast, a data structure is typically specified by a query problem \(f:\mcalx^N\times \mcalq\to \mcaly\), where the query \(Q\in \mcalq\) determines which value \(f(X,Q)\) must be returned.

\begin{definition}[Query-with-sketch algorithms]
\label{def:QwS}
Let \(f:\mcalx^N\to\mcaly\) be a function.  A deterministic
query-with-sketch algorithm \(\Pi^S\) with sketch length \(r\) consists of two
objects.

\begin{itemize}
    \item A sketching map \(S:\mcalx^N\to\{0,1\}^r\).  On input \(X\), the
    algorithm is given the sketch \(S(X)\).

    \item A deterministic transition rule \(\Pi\).  For every $l\ge 0$, given a transcript
    $$
        C_\ell =
        \bigl(S(X),(i_1,a_1),\ldots,(i_\ell,a_\ell)\bigr),
    $$
    where \(i_j\in[N]\) is a probed coordinate and \(a_j\in\mcalx\) is the
    value returned by the probe, the rule \(\Pi\) either outputs a new probe
    \(i_{\ell+1}\in[N]\) or halts with an answer \(y\in\mcaly\).
\end{itemize}

The computation on input \(X\) starts from \(C_0=(S(X))\).  If
\(\Pi(C_\ell)=i_{\ell+1}\in[N]\), the algorithm probes coordinate
\(i_{\ell+1}\), receives \(a_{\ell+1}=X_{i_{\ell+1}}\), and continues with
\(C_{\ell+1}\).  If \(\Pi(C_\ell)=y\in\mcaly\), the algorithm halts and outputs
\(y\).

We allow repeated probes, and count them with multiplicity.  A transcript of
length \(\ell\) induces a partial assignment to at most \(\ell\) distinct input
coordinates.  Let \(\mathsf{T}_{\Pi^S}(X)\) denote the number of probes made before
halting, with \(\mathsf{T}_{\Pi^S}(X)=\infty\) if the computation does not halt.  The
worst-case and average-case probe complexities are
$$
    \text{\rm Cost}(\Pi^S) := \sup_{X\in\mcalx^N} \mathsf{T}_{\Pi^S}(X),
    \qquad
    \text{\rm Avg-Cost}_{\mu}(\Pi^S) :=
        \mathbb{E}_{X\sim\mu}\bigl[\mathsf{T}_{\Pi^S}(X)\bigr].
$$
\end{definition}

\begin{definition}[Randomized query-with-sketch algorithms]
A randomized query-with-sketch algorithm is a distribution over deterministic
query-with-sketch algorithms.  Equivalently, a random tape \(\rho\) is sampled
before the sketch is formed, and the algorithm then uses the deterministic pair
\((S_\rho,\Pi_\rho)\) throughout the computation.

The worst-case sketch length and worst-case probe complexity are the maximum
over all random tapes and all inputs.  The average-case probe complexity with
respect to an input distribution \(\mu\) is
$$
    \mathbb{E}_{\rho}\mathbb{E}_{X\sim\mu}
    \bigl[T_{\Pi_\rho^{S_\rho}}(X)\bigr].
$$
The success probability is taken over both \(X\sim\mu\) and the random tape
\(\rho\).
\end{definition}

\paragraph{Query-with-sketch and related models}

As a natural information-theoretic model of computation, variants of the query-with-sketch model have been introduced and studied in different scenarios. The query-with-sketch model is a direct generalization of the decision tree model with \emph{help bits} \cite{nisan1998products, beigel1998one}. In that model, algorithms compute the direct sum of $k$ instances of a boolean function using $k-1$ help bits, which can be viewed as a sketch of length $k-1$. In addition, each output bit is computed using an individual decision tree in their case, whereas in our setting, all the output bits are decided in a single decision tree. The lower bound to succinct but unrestricted data structures for the partial sums problem \cite{puatracscu2010cell} also reduces to lower bounds to data structures with \emph{published bits}. However, the model used in this work is substantially different from ours. The queries are made to a data structure instead of the input elements.

\subsection{The min-entropy lemma}
\label{sec:min-entropy-theorem}

We show that lower bounds to query-with-sketch algorithms can be reduced to lower bounding the conditional joint min-entropy of the output.

\begin{definition}[joint min-entropy]
    Let $X$ be a random variable and $Y$ an event. The \emph{joint min-entropy} of $X$ and $Y$ is defined as $\hmin(X,Y):=\min_x-\log \Pr[X=x\wedge Y]$. When $\Pr[Y]=0$, $\hmin(X,Y)$ is defined as $+\infty$.
\end{definition}

We relate the complexity of query-with-sketch algorithms with a special form of conditional joint min-entropy of the output. To formally state our result, let us give necessary definitions.

\begin{definition}
    A \emph{partial assignment} to a function $f:\mcalx^N\rightarrow \mcaly$
    of length $q$ is an event of the form $(X_{i_1}=x_{i_1},\dots,X_{i_q}=x_{i_q})$,
    where $i_1,\dots,i_q\in [N]$ are distinct coordinates and
    $x_{i_j}\in \mcalx$ for every $j\in[q]$. We use $\mcale_{\le q}$ to
    denote the set of all partial assignments of length at most $q$.
\end{definition}


It is important to note that a partial assignment is an event about the function itself and is irrelevant to the algorithm.

\begin{lemma}
    \label{lem:min-entr}
    Let $f:\mathcal X^N\rightarrow \mathcal Y$ be an arbitrary function, and $\mu_{\mathcal X^N}$ an distribution over $\mathcal X^N$.
    Fix $\Pi^S$ to be a deterministic query-with-sketch algorithm for $f$ with average-case probe complexity $\le 0.1q$.
    
    In addition, suppose there exists a set $\mathsf X\subseteq\mcalx^N$ of ``good'' inputs such that
    $$\Pr_{X\leftarrow \mu_{\mcalx^N}}[X\in \mathsf X]\ge 99/100.$$
    Furthermore, assume there is a positive integer $r\ge 2$ such that for every partial assignment $\sigma\in \mcale_{\le q}$ of length $\le q$, there exists a subset of inputs $\msfx_\sigma\subseteq \mcalx^N$ such that
    $$\Pr_{X\leftarrow \mu_{\mcalx^N}}[ X\in \msfx\backslash\msfx_\sigma|\sigma]\le 2^{-2r}$$
    and the output of $f$ remains highly uncertain:
        $$H_{\min}(f(X),X\in \mathsf X\cap \msfx_\sigma|\sigma)>2r$$
    Then, for $\Pi^S$ to achieve $\ge 99/100$ success probability given a random input from $\mu_{\mcalx^N}$, its sketch length must be $|S|> r$.
\end{lemma}

For our purpose of proving a query-with-sketch lower bound for randomized algorithms given a fixed sequence of queries, we could prove distributional lower bounds for deterministic query-with-sketch algorithms instead, by the standard Yao's minimax principle.

The sets $\mathsf X,\msfx_\sigma$ describe good inputs where the min-entropy can be well-bounded. We note that $\mathsf X$ and $\msfx_\sigma$ are optional. If we set $\msfx=\msfx_\sigma=\mcalx^N$ for every $\sigma$, there is no constraint on the inputs.

\begin{proof}[Proof of Lemma \ref{lem:min-entr}]
    Assume for the sake of contradiction that the sketch length is $\le r$. Without loss of generality, we let the sketch length be $r$. We will show that the success probability of $\Pi^S$ cannot achieve $\ge 99/100$, which is a contradiction.

    We characterize $\Pi^S$ as a set of $2^r$ decision trees. Each input $X$ is assigned to a decision tree $T_S$ indexed by its sketch $S$. $f(X)$ is computed correctly if $T_S$ given input $X$ produces the correct answer $f(X)$. Each leaf of $T_S$ can be characterized by a pair $(y,\sigma)$, corresponding to the output at the leaf and the probe history. Let $E_{S}$ denote the event that ``$X$ is assigned with sketch $S$'', and $E_{q}$ the event that ``$\Pi^S$ given input $X$ makes no more than $q$ probes''. Then $\Pr[E_q]\ge 0.9$, or the average probe complexity would exceed $0.1q$.
    We have

    \begin{align*}
        \Pr_{X\leftarrow\mu_{\mcalx^N}}[\Pi^S(X)=f(X)]&=\sum_{S\in\{0,1\}^r}\sum_{(y,\sigma)\textnormal{ is a leaf of }T_S}\Pr[E_S\textnormal{ and }\sigma\textnormal{ and }f(X)=y]\\
        &\le 0.11+\sum_{S\in\{0,1\}^r}\sum_{\substack{(y,\sigma)\textnormal{ is a leaf of }T_S \\ \textnormal{of depth }\le q}}\Pr[E_S\textnormal{ and }\sigma\textnormal{ and }f(X)=y\textnormal{ and }X\in \mathsf X
        \textnormal{ and }E_q]\\
        &\le 0.11+\sum_{S\in\{0,1\}^r}\sum_{\substack{(y,\sigma)\textnormal{ is a leaf of }T_S \\ \textnormal{of depth }\le q}}\Pr[\sigma]\cdot \Pr[f(X)=y\textnormal{ and }X\in \mathsf X\textnormal{ and }E_q|\sigma]\\
        &\le 0.11+\sum_{S\in\{0,1\}^r}\sum_{\substack{(y,\sigma)\textnormal{ is a leaf of }T_S \\ \textnormal{of depth }\le q}}\Pr[\sigma]\cdot (\Pr[f(X)=y\textnormal{ and }X\in \mathsf X\cap \msfx_\sigma\textnormal{ and }E_q|\sigma]+2^{-2r})\\
        &< 0.11+\sum_{S\in\{0,1\}^r}\sum_{\substack{(y,\sigma)\textnormal{ is a leaf of }T_S \\ \textnormal{of depth }\le q}}\Pr[\sigma]\cdot 2^{-2r+1}\\
        &\le 0.11+\sum_{S\in\{0,1\}^r}2^{-2r+1}\\
        &\le 2^{-r+1}+0.11<0.99
    \end{align*}
    The second line follows our assumption that $X\in\msfx$ with $\ge 99/100$ probability. The fourth and the fifth line follow our assumption to the conditional joint min-entropy and $\Pr[X\in\msfx_\sigma|\sigma]$. The sixth line follows from the fact that the partial assignments $\sigma$ of different leaves of the same decision tree are disjoint events of the input.

\end{proof}

\subsection{Reduction from data structures}

\label{sec:qws_reduction}

We prove Lemma~\ref{lem:reduction_ds_qws}, which reduces data structure lower bounds to query-with-sketch lower bounds.

\begin{proof}[Proof to Lemma~\ref{lem:reduction_ds_qws}]
    The data structure $D$ can be characterized as a sequence of $l/m$ query-with-sketch algorithms for computing each block of $m$ queries. The sketch to each algorithm is the (updated) redundancy after the end of its previous algorithm. The probe made by all the query-with-sketch algorithms are exactly the probe made by $D$ to answer $Q$. In addition, the error probability of each query-with-sketch algorithm is at most the error probability of the data structure $D$. Besides, in both algorithms, the succinct space (sketch) exclusively depends on the input $X\in \mcalx^n$. Among these algorithms, at least one of them has average-case probe complexity below average $\le tm$.
\end{proof}

Through the above reduction and min-entropy lemma, to obtain a data structure lower bound one only need to (i) construct an input distribution $\muxn$, a subset $\mathsf X$ of $\mcalx^N$, and subsets $\msfx_\sigma$ for every $\sigma$; (ii) show that for every partial assignment $\sigma\in \mcale_{\le q}$, $\hmin(f(X), X\in \msfx\cap \msfx_\sigma|\sigma)>2r$; and (iii) show that with high probability a random input falls in $\msfx$ (resp., falls in $\msfx_\sigma$ conditioned on $\sigma$, for every $\sigma$).

\paragraph{Comparison to a previous reduction.} The query-with-sketch complexity of the \emph{direct sum} of boolean problems was introduced and studied in \cite{nisan1998products,beigel1998one}.
They showed that any $l$ decision trees for computing the direct sum of the parity of $m$ bits each, with $l-1$ help bits (i.e., the sketch), at least one of the decision trees must have a depth $\ge m$. Later, \cite{gal2007cell} got a data structure lower bound for prefix sum by a reduction from the help-bits lower bound for direct-sum parity. However, their reduction is specific to the prefix sum problem. Our reduction goes in a different way and is generally applicable to succinct and systematic data structures whose min-entropy can be bounded. In addition, our reduction is applicable to every problem $f$ whose conditional min-entropy is bounded. The remainder of this paper will be devoted to showing the min-entropy bounds on the aforementioned problems. To compare our approach with the previous work, we explain the lower bound in \cite{gal2007cell} below.

In the prefix sum problem, the input is a length-$n$ bit string $X$. The data structure, given an arbitrary index $k\in[n]$, should output the parity of $\sum_{i=1}^k X_i$. In \cite{gal2007cell}, they showed that every succinct and systematic data structure for the prefix sum with $r$ bits of redundancy requires $\Omega(n/(r+1))$ probes to the input bits. Their proof idea is as follows. Assume for contradiction that there exists a data structure with $r$ bits of redundancy that makes $< \frac{n}{2(r+1)}$ bit probes for every given index $k\in[n]$. Split the input string $X$ into $r+1$ blocks of about equal length, the parity of each block can be determined using $<n/(r+1)$ probes. However, \cite{nisan1998products} shows that even with $r$ help bits (i.e., the sketch), to compute $r+1$ instances of parity of $n/(r+1)$ bits, at least one of the $r+1$ decision trees has depth $\ge n/(r+1)$.

\section{Lower bounds to approximate matrix powering with sublinear-in-\texorpdfstring{$n$}{n} redundancy}

In this section, we study the approximate matrix powering problem. For this problem, we establish an $r\cdot t=\Omega(n^2)$ trade-off for data structures with redundancy $r< 0.25n\ln n$ and amortized probe complexity $t$. The subsequent section addresses the case where $r\in [10n\log n, n^2/1000]$, proving the same trade-off with a different technique. Let us begin by defining our rounding notation and then formally state the problem.

\paragraph{Approximate matrix powering (AMP) problem.} Fix a constant $\alpha>0$ and a number $k\ge 2$. The input is a substochastic symmetric matrix $\mbold\in\mathbb R^{n\times n}$. The succinct and systematic data structure, given a pair of indices $(u,v)\in [n]\times [n]$, must produce a value $ans$ such that $|\mbold^k[u,v]-ans|\le 1/n^\alpha$.

\begin{theorem}[Lower bound to AMP with sublinear-in-$n$ redundancy]
    \label{thm:amp_lb_small_s}
    Fix $3\le k\le n$ and a sufficiently large constant $\alpha$. For every (possibly randomized) succinct and systematic data structure with $r\in(c\log^2 n,0.25n\ln n)$ bits of redundancy that answers the whole approximated matrix power $\mbold^k$ with additive error $1/n^\alpha$ with $\ge 9/10$ probability, its total number of probes to $\mbold$ is $\Omega(n^4/r)$ at the worst case, where $c$ is a large enough constant.
\end{theorem}

We note that we prove a strong statement where in average, answering each query requires $\Omega(n^2/r)$ probes.

We give a distributional lower bound to the following nemesis input distribution.

\paragraph{Nemesis input distribution.} Fix a large enough integer $n\ge 1$, and $\beta,\gamma$ large enough constant parameters that depends on $\alpha$. Let $\mathcal M_{\beta,\gamma}$ be a distribution of $n\times n$ matrices $\mathbf M=\frac{n^2-1}{n^2}\mathbf I+\frac{1}{n^\beta}\mathbf M_{\mathbf U}$, where $\mathbf M_{\mathbf U}$ is a random symmetric matrix whose upper triangle elements $\mbold[i,j]$ are independent and identically distributed (i.i.d.) uniformly from $[\frac 1n-\frac{1}{n^{\gamma}},\frac{1}{n}]$, for every $i\le j$. We write $\mathbf M\leftarrow \mathcal M_{\beta,\gamma}$ (or simply $\mbold\leftarrow\mcalm$) for $\mathbf M$ sampled from distribution $\mathcal M_{\beta,\gamma}$.

To apply Lemma~\ref{lem:reduction_ds_qws}, we also need to fix a sequence of queries (i.e., a permutation to the set of possible queries to the data structure $[n]\times [n]$). We will show that to answer each of the consecutive $3m:=3r/\ln n$ queries, a query-with-sketch algorithm either requires $r$ bits of redundancy or $\Omega(n^2)$ probes to the input matrix.

Specifically, we fix the following sequence of queries, which is a concatenation of $(n-1)$ shifted main diagonals of $\mbold^k$ (excluding the main diagonal itself):
$$Q=(\underbrace{(1,2),(2,3),\dots,(n,1)}_{\textnormal{Shift 1: }v=u+1},\underbrace{(1,3),(2,4),\dots,(n,2)}_{\textnormal{Shift 2: }v=u+2},\dots,\underbrace{(1,n),(2,1),\dots,(n,n-1)}_{\textnormal{Shift }n-1:~v=u-1})$$

Formally, the sequence $Q$ consists of $n-1$ blocks, where for every $i\in[n-1]$, the $i$-th block contains $n$ queries $(u,v)$ defined by:
\begin{itemize}
    \item $u\in[n]$;
    \item $v=((u+i-1)\bmod n)+1$.
\end{itemize}

All we need from the construction is the following property.

\begin{fact}
    \label{fact:amp_q_non_conflict}
    For any window of $3j$ consecutive queries from $Q$ (where $3j<n$), one can select a subset of $j$ queries from this window such that all the $2j$ indices (that is, both the row $u$ and column $v$ from each selected query) are pairwise distinct.
\end{fact}

\begin{proof}
    First, we show that any $3j$ consecutive queries span $3j$ different rows and $3j$ different columns.
    
    Since $3j<n$, the $3j$ consecutive queries must belong to at most 2 different shifted main diagonals. For queries belong to the same diagonal, they clearly belong to different rows and columns.

    Given two queries $(u,v)$, $(u',v')$ that respectively belong to the $i$-th and the $(i+1)$-th diagonals. If they belong to a window of length $3j<n$, we have $u'\le u-2$. They do not share the same row. In addition,
    $$(v-v')\bmod n=(u-u'-1)\bmod n\ne 0$$
    when $u'\le u-2$. They do not share the same column as well.

    The existence of the $j$ queries follows from a greedy construction: we greedily pick unselected queries that have distinct rows and columns, until we cannot pick any more queries. Since the $3j$ elements belong to different rows and columns, picking one query $(u,v)$ from it will disable at most $2$ other queries from the window: at most one from column $u$, and at most one from row $v$. Therefore, one can always pick $j$ many queries.
\end{proof}

Our query-with-sketch framework is built for exact problems. To apply it to AMP, which is an approximate problem (allowing additive error), we first prove our lower bound on an exact variant we call \emph{rounded AMP}. In this problem, the algorithm must output the true value rounded to a precision of $1/n^{\alpha'}$.

A lower bound for this rounded problem implies a lower bound for the original AMP problem because the two problems are almost identical. They only differ on a small set of ``boundary'' inputs--inputs where the true value is too close to a rounding point to be distinguished. We formalize this by first showing that this property holds in each dimension. Then we apply a union bound over all the dimensions.

\begin{definition}[rounding notation]
    Fix a rounding parameter $\alpha'>0$. For a real number $x\in\mathbb R$, we denote its rounded value to its nearest multiple of $1/n^{\alpha'}$ as $[x]_{\alpha'}$. For a real matrix $\mbold \in \mathbb{R}^{n \times n}$, we use $[\mbold]_{\alpha'}$ to denote the matrix where this rounding is applied element-wise.
\end{definition}

\paragraph{Rounded AMP problem.} Fix a rounding constant $\alpha'>0$ and a number $k\ge 2$. The input is a substochastic symmetric matrix $\mbold\in\mathbb R^{n\times n}$. The succinct and systematic data structure, given a pair of indices $(u,v)\in [n]\times [n]$, must output $[\mbold^k[u,v]]_{\alpha'}$.

Given Lemma~\ref{lem:reduction_ds_qws}, we only need to show that for every $3m=3r/\ln n$ consecutive queries, every query-with-sketch algorithm with $r$ bits of sketch that computes the $3m$ queries with probability $\ge 99/100$, the algorithm makes $\Omega(n^2)$ probes in the worst case.

By Lemma~\ref{lem:min-entr}, proving such a lower bound reduces to lower bounding the conditional min-entropy of the $m$ query answers by $2r$, conditioned on assignments to the input matrix $\mbold$ of length $q:=10^{-4}n^2$. For this input distribution, we will construct a nontrivial good input set $\mathsf{M}\subseteq\mbbr^{n\times n}$. And we will trivially set $\msfm_\sigma=\mbbr^{n\times n}$ for every $\sigma$. Formally, the remainder of this section will be devoted to proving the following technical lemma.

\begin{lemma}
    \label{lem:tech_lem_amp}
    Fix $k\ge 3$ and constants $\alpha',\beta,\gamma$ such that $\gamma\ge 2$, $\beta\ge \gamma+2$ and $\alpha'\ge3\beta+2\gamma+3.6$. For every large enough integer $n$, there exists a set of ``good'' matrices $\msfm\subseteq \mathbb R^{n\times n}$ such that
    $$\Pr_{\mbold\leftarrow \muxn}[\mbold\in \msfm]\ge 1-o(1).$$
    In addition, fix $r\in(c\log^2 n,n\ln n/2)$, for every $3m=3r/\ln n$ consecutive queries $(u_1,v_1),\dots,(u_{3m},v_{3m})$ from $Q$, and every partial assignment $\sigma$ of length $q=10^{-4}n^2$ to $\mbold$, we have 
    $$\hmin(([\mbold^k[u_1,v_1]]_{\alpha'},\dots,[\mbold^k[u_{3m},v_{3m}]]_{\alpha'}),\mbold\in \msfm|\sigma)>2r$$
    where $c$ is a constant that depends only on $\alpha',\beta,\gamma$.
\end{lemma}

In addition, the following lemma shows that the outputs close to the boundary has a negligible measure. Its proof is deferred to Appendix~\ref{app:lipschitz}

\begin{lemma}
    \label{lem:round_negl}
    Fix $k\ge 3$ and constants $\alpha,\alpha',\beta,\gamma$ such that $\gamma\ge  2$, $\beta\ge \gamma+2$, $\alpha'\ge 3\beta+2\gamma+3.6$ and $\alpha> \alpha'+2$. For every large enough integer $n$, we have
    $$\Pr_{\mbold\leftarrow \mcalm}\left[\exists u,v\in[n],a\in \mathbb N, \left|\mbold^k[u,v]-a/2n^{\alpha'}\right|\le 1/n^\alpha \right]< o(1)$$
\end{lemma}

Given Lemma~\ref{lem:tech_lem_amp} and Lemma~\ref{lem:round_negl}, our lower bound to approximate matrix powering follows.

\begin{proof}[Proof to Theorem~\ref{thm:amp_lb_small_s}]
    Let us first show that any data structure for rounded AMP with success probability $\ge 99/100$ either requires $\ge r$ redundancy or $\ge \Omega(n^4/r)$ worst-case probe complexity to answer queries in $Q$.
    
    By Lemma~\ref{lem:reduction_ds_qws}, given the constructed sequence $Q$ of queries, we only need to show that for every $3m$ consecutive queries from $Q$, every randomized query-with-sketch algorithm that computes the $3m$ queries with probability $\ge 99/100$ must pay $\ge r$ bits of sketch or $\ge 0.1q=0.3tm$ average-case probe complexity.

    By Yao's minimax principle, given our input distribution $\mcalm$, we instead show that for every $m$ queries and every deterministic query-with-sketch algorithm with success probability $\ge 99/100$ given a random input from $\mcalm$, either $r$ bits of sketch or $0.1q=0.3tm$ average-case probe complexity is required. The lower bound follows from Lemma~\ref{lem:min-entr} and Lemma~\ref{lem:tech_lem_amp}.

    Now, we show that the same lower bound applies to the original AMP. We achieve it by a reduction from the rounded AMP. That is, we design a data structure $D'$ for the rounded AMP given a data structure $D$ for the AMP, with the same redundancy, probe complexity, and only an $o(1)$ increase in fail probability.

    $D'$ simply outputs the output of $D$ rounded to the nearest multiple of $1/n^{\alpha'}$. By the correctness of $D$, for every query $(u,v)$, $$|\mbold^k[u,v]-D(u,v)|\le 1/n^{\alpha}.$$
    $\mbold^k[u,v]$ rounds to another value only if there exists a number $a\in \mathbb N$ such that one of the following two cases happen:
    $$D(u,v)\le a/2n^{\alpha'}\le \mbold^k[u,v] $$
    or
    $$D(u,v)\ge a/2n^{\alpha'}\ge \mbold^k[u,v],$$
    which happens only if
    $$|\mbold^k[u,v]-a/2n^{\alpha'}|\le |\mbold^k[u,v]-D(u,v)|\le 1/n^\alpha.$$

    However, by Lemma~\ref{lem:round_negl}, this happens with $<o(1)$ probability. Therefore, the success probability of $D'$ is at most $o(1)$ worse than the probability of $D$. The above lower bound also applies to AMP.
\end{proof}

\subsection{Partial derivatives between the inputs and outputs}

We are going to show that, for every $3m=3r/\ln n$ consecutive queries $((u_1,v_1),\dots, (u_m,v_m))$ that span $m$ different rows and columns, and every partial assignment $\sigma$ of length $q=0.001n^2$, the query answers conditioned on the partial assignment remain a high min-entropy for good inputs from $\msfm$. At a high level, every partial assignment fixes a proportion of input elements. We will show that, given every partial assignment, one can always pick $\Theta(r/\log n)$ elements from the unfixed input elements that form a ``basis'' to the output elements. That is, each fixed and possible output only corresponds to a small subset of inputs. We capture this by studying the Jacobian matrix between the output elements and the input elements. Note that such a basis does not always exist, where we construct a large enough good input set $\msfm$ that guarantees the existence of the basis. 

To start with, we regard each output element $\mbold^k[u,v]$ as a function of $n^2$ input elements $\mbold[i,j]$, and study their partial derivatives $\frac{\partial\mbold^k[u,v]}{\partial\mbold[i,j]}$. Recall that $\mbold=\frac{n^2-1}{n^2}\ibold+\frac1{n^\beta}\mboldu$ where each element $\mboldu[i,j]=\mboldu[j,i]$ of $\mboldu$ is uniformly random from $[\frac1n-\frac1{n^\gamma},\frac1n]$. We ignore for now the fact that $\mbold$ is a symmetric matrix. The partial derivative under the symmetric constraint is simply $\frac{\partial\mbold^k[u,v]}{\partial\mbold[i,j]}+\frac{\partial\mbold^k[u,v]}{\partial\mbold[j,i]}$ when $i\ne j$.

We will use the following fact multiple times in the proof.

\begin{fact}
    \label{fact:e_inequality}
    For every integer $n>1$, $$\left(1-\frac1n\right)^n\le e^{-1}\le \left(1-\frac1n\right)^{n-1}$$
\end{fact}

\begin{proof}
    For simplicity, we let $f(n)=(1-\frac1n)^n$, $g(n)=(1-\frac1n)^{n-1}$ for now.
    It is known that $$\lim_{n\rightarrow \infty}f(n)=\lim_{n\rightarrow\infty}g(n)=e^{-1}.$$ 
    Given this fact, to show the inequalities hold whenever $n>1$, we show that (i) the inequalities hold true when $n=2$, and (ii) $f(n)$ is an increasing function and $g(n)$ is a decreasing function.
    
    One can verify that the inequalities hold for the case $n=2$. To show the monotonicity of $f(n)$ and $g(n)$, 
    $$\begin{aligned}
    f'(n)&=\left(e^{n\ln(1-1/n)}\right)'=\left(1-\frac1n\right)^n\left(\ln(1-1/n)+\frac{1}{n-1}\right)\\
    &=\left(1-\frac1n\right)^n\left(\left(-\frac1n-\frac1{2n^2}-\frac1{3n^3}-...\right)+\frac{1}{n-1}\right)>0
    \end{aligned}$$
    when $n\ge 2$.

    $$\begin{aligned}
    g'(n)&=\left(e^{(n-1)\ln(1-1/n)}\right)'=\left(1-\frac1n\right)^{n-1}\left(\ln(1-1/n)+\frac{1}{n}\right)\\
    &=\left(1-\frac1n\right)^n\left(\left(-\frac1n-\frac1{2n^2}-\frac1{3n^3}-...\right)+\frac{1}{n}\right)<0
    \end{aligned}$$
    when $n\ge 2$.
\end{proof}

To bound the partial derivatives $\frac{\partial\mbold^k[u,v]}{\partial\mbold[i,j]}$, we bound $\frac{\partial\mboldu^k[u,v]}{\partial\mboldu[i,j]}$ first.

\begin{lemma}
    \label{lem:partmboldu}
    Let $\gamma\ge 2$, and the values of elements of $\mboldu$ range from $[\frac 1n-\frac{1}{n^\gamma},\frac 1n]$. For every positive integer $3\le k\le n$, for every element $[i,j]$ of $\mboldu$ and every element $[u,v]$ of $\mboldu^k$ such that $u\ne i$ and $j\ne v$. When $k\ge 3$, we have:
        $$\frac{\partial \mboldu^k[u,v]}{\partial \mboldu[i,j]}=\Theta\left(\frac{k-2}{n^2}\right)$$
    Specifically, when $k=3$,
        $$\frac{\partial \mboldu^k[u,v]}{\partial \mboldu[i,j]}=\mboldu[u,i]\mboldu[j,v]$$
\end{lemma}

\begin{proof}[Proof of Lemma \ref{lem:partmboldu}]
    Note that
        $$\mboldu^k[u,v]=\sum_{u=h_0,h_1,...,h_{k-1},h_k=v}\prod_{t=0}^{k-1}\mboldu[h_t,h_{t+1}]$$
When $u\ne i$ and $j\ne v$, the partial derivative is
    \begin{align*}
    \frac{\partial \mboldu^k[u,v]}{\partial\mboldu[i,j]}=\sum_{g=1}^{k-2}\left(\sum_{u=h_0,h_1,...,h_g=i,h_{g+1}=j,...,h_{k-1},h_k=v}\left(\prod_{t=0,t\ne g}^{k-1}\mboldu[h_t,h_{t+1}]\right)\right)
    \end{align*}
In this equation the appearance of $\mboldu[i,j]$ is canceled at position $g$ in the summation of the product. This directly gives the expression for $k=3$. Since for all $i,j\in[n]$, $\mboldu[i,j]=\mboldu[j,i]\in[\frac1n-\frac{1}{n^\gamma},\frac1n]$, which is a quite small range, $\prod_{t=0,t\ne g}^{k-1}\mboldu[h_t,h_{t+1}]$ can be lower-bounded by setting all elements in $\mboldu$ equal to $\frac 1n-\frac{1}{n^{\gamma}}$ and upper bounded by setting all elements in $\mboldu$ equal to $\frac1n$. We only need to count the number of sequences $u=h_0,h_1,...,h_{k-1},h_k=v$ in which two adjacent indices are $i$ and $j$, while in both the lower and upper bounds all elements are equivalent. Since $g$ can only be picked from $\{1,2,...,k-2\}$, there are $(k-2)n^{k-3}$ distinct sequences $h_0,\dots,h_k$. We get
$$(k-2)n^{-2}=(k-2)n^{k-3}\left(\frac 1n\right)^{k-1}\ge \frac{\partial \mboldu^k[u,v]}{\partial\mboldu[i,j]}\ge (k-2)n^{k-3}\left(\frac{1}{n}-\frac{1}{n^\gamma}\right)^{k-1}\ge e^{-\frac{k-1}{n^{\gamma-1}-1}}(k-2)n^{-2}.$$
Recall that $\gamma\ge 2$, so $\frac{\partial \mboldu^k[u,v]}{\partial\mboldu[i,j]}=\Theta(\frac{k-2}{n^2})$.
\end{proof}

\begin{lemma}
    \label{lem:partmkm}
    Let $\gamma,\beta\ge 2$, and the values of elements of $\mboldu$ range from $[\frac 1n-\frac{1}{n^\gamma},\frac 1n]$. For every positive integer $3\le k\le n$, every element $(i,j)$ of $\mathbf M$ and every element $(u,v)$ of $\mbold^k$ such that $u,v,i,j$ are pairwise distinct, we have
        $$\frac{\partial \mbold^k[u,v]}{\partial \mbold[i,j]}=\binom k3\frac{(n^2-1)^{k-3}}{n^{2k-6+2\beta}}\left(\mboldu[u,i]\mboldu[j,v]+\mboldu[u,j]\mboldu[i,v]\right)+O(k^5/n^{2+3\beta}).$$
\end{lemma}

\begin{proof}
    Notice that    
    $$\mbold^k=\sum_{t=0}^k\binom kt\left(\frac{n^2-1}{n^2}\right)^{k-t}\left(\frac{1}{n^\beta}\mboldu\right)^t$$
    and
    \begin{align*}
        \frac{\partial\mbold^k[u,v]}{\partial\mbold[i,j]}&=\sum_{t=0}^k\binom kt\left(\frac{n^2-1}{n^2}\right)^{k-t}\left(\frac{1}{n^\beta}\right)^{t}\frac{\partial\mboldu^t[u,v]}{\partial\mbold[i,j]}\\
        &=\sum_{t=0}^k\binom kt\left(\frac{n^2-1}{n^2}\right)^{k-t}\left(\frac{1}{n^\beta}\right)^{t-1}\frac{\partial\mboldu^t[u,v]}{\partial\mboldu[i,j]}
    \end{align*}

    While $u,v,i,j$ are pairwise distinct, $\frac{\partial\mboldu^t[u,v]}{\partial\mboldu[i,j]}$ is $0$ when $t\le 2$. Therefore,
    $$\frac{\partial\mbold^k[u,v]}{\partial\mbold[i,j]}=\binom k3\frac{(n^2-1)^{k-3}}{n^{2k-6+2\beta}}\mboldu[u,i]\mboldu[j,v]+O(k^5/n^{2+3\beta}).$$
\end{proof}

\subsection{Jacobian matrix and singular value concentration bounds}

While the previous section studies the partial derivatives between individual input elements and output elements of different rows and columns. We extend it to the interplay between multiple input elements and multiple output elements, which formalizes as the Jacobian matrix. We will show that, fix a number $m\le n/4$, for every $m$ input elements and $m$ output elements whose column indices and row indices are pairwise distinct, the largest $0.9m$ singular values of their Jacobian matrix are controlled with high probability. We use a recent work \cite{wei2017upper} on singular values concentration bounds to achieve it.

\begin{definition}
    Fix $m\le n/4$, $m$ input indices $(i_1,j_1),\dots,(i_{m},j_m)$, and $m$ output indices $(u_1,v_1),\dots,(u_m,v_m)$. When $i_1\ne j_1,\dots,i_m\ne j_m$, their Jacobian matrix is defined as
    $$\jbold=\begin{pmatrix}
    \frac{\partial\mbold^{k}[u_1,v_1]}{\partial\mbold[i_1,j_1]}+ \frac{\partial\mbold^{k}[u_1,v_1]}{\partial\mbold[j_1,i_1]}&\dots&\dots&\frac{\partial\mbold^{k}[u_m,v_m]}{\partial\mbold[i_1,j_1]}+\frac{\partial\mbold^{k}[u_m,v_m]}{\partial\mbold[j_1,i_1]}\\
    \vdots&\ddots&&\vdots\\
    \vdots&&\ddots&\vdots\\
    \frac{\partial\mbold^{k}[u_1,v_1]}{\partial\mbold[i_m,j_m]}+\frac{\partial\mbold^{k}[u_1,v_1]}{\partial\mbold[j_m,i_m]}&\dots&\dots&\frac{\partial\mbold^{k}[u_m,v_m]}{\partial\mbold[i_m,j_m]}+\frac{\partial\mbold^{k}[u_m,v_m]}{\partial\mbold[j_m,i_m]}
    \end{pmatrix}$$

    Given an $m\times m$ Hermitian matrix $\abold\in \mbbr^{m\times m}$, we denote by $\lambda_1(\abold)\ge \dots\ge\lambda_m(\abold)$ eigenvalues of $\abold$.

    The singular values $\sigma_1(\abold),\dots,\sigma_m(\abold)$ of a (possibly non-Hermitian) matrix $\abold\in \mbbr^{m\times m}$, are the eigenvalues of $\sqrt{\abold^T\abold}$ arranged in non-increasing order.
\end{definition}

We will use the following standard tools from matrix analysis.

\begin{proposition}[Weyl's inequality for singular values \cite{horn2012matrix}]
    \label{prop:weyl's inequality}
    Let $\abold,\bbold\in \mbbr^{m\times m}$ be general matrices. For every $k\in[m]$,
    $$|\sigma_k(\abold+\bbold)-\sigma_k(\abold)|\le \|\bbold\|_{op}$$
    where $\|\bbold\|_{op}=\sigma_1(\bbold)$ is the operator norm.
\end{proposition}

\begin{proposition}[\cite{thompson1976behavior}]
    \label{prop:thompson's interlacing} Let $\abold,\bbold\in\mbbr^{m\times m}$ where $\rank(\bbold)= 1$. For every $k\in[m]$,
    $$\sigma_{k+1}(\abold)\le \sigma_k(\abold+\bbold)\le \sigma_{k-1}(\abold)$$
    where we denote $\sigma_0(\abold)=+\infty$ and $\sigma_{m+1}(\abold)=0$.
\end{proposition}

\begin{proposition}[Ger\v{s}gorin circle theorem \cite{horn2012matrix}]
    \label{prop:gergorin_circle_theorem}
    Let $\abold\in\mbbc^{m\times m}$. Let $R_i(\abold)=\sum_{j\ne i}|\abold[i,j]|$ denote the absolute row sums of $\abold$ other than $\abold[i,i]$. And consider the $m$ Ger\v{s}gorin discs $$\{z\in \mbbc:|z-\abold[i,i]|\le R_i(\abold)\}$$ for $i\in[m]$. The eigenvalues of $\abold$ are in the union of Ger\v{s}gorin discs
        $$G(A)=\bigcup_{i=1}^m\{z\in \mbbc:|z-\abold[i,i]|\le R_i(\abold)|\}.$$
\end{proposition}

We will use the following concentration bound on the intermediate singular values of a random matrix \cite{szarek1990spaces, wei2017upper}. We adapt their result in our setting and show the following in Appendix~\ref{app:singular_value}.

\begin{lemma}
    \label{lem:singular_value_concentration}
    For $i\in\{1,2\}$, let $X_i,Y_i,Z_i,W_i$ be i.i.d.\ uniform on some symmetric interval (the interval may depend on $i$).
Define
$$
U_1 := X_1+Y_1+Z_1+W_1,
\qquad
U_2 := X_2Y_2+Z_2W_2.
$$
Let $\mu_i$ be the distribution of $U_i$.
    
    Fix $\mu$ to be either $\mu_1$ or $\mu_2$.
    Let $m$ a large enough number, and $\abold\in \mbbr^{m\times m}$ be a matrix whose elements are i.i.d. drawn from $\mu$. There exists constants $0<C_1<C_2$ and $C_3>0$ such that for all $\ell$ between $1$ and $m$,
    $$\Pr\left[\frac{C_1\ell}{\sqrt m}\le \sigma_{m+1-\ell}(\abold)\le \frac{C_2\ell}{\sqrt m}\right]\ge 1-\exp(-C_3\ell)$$
\end{lemma}

Basically, a random matrix of the above form has high singular values to a large proportion of its dimensions with high probability. Our Jacobian matrix is a combination of scaled, shifted and perturbed random matrices of the above form. Formally, we give the following concentration bound on our Jacobian matrix.

\begin{lemma}
    \label{lem:singular_bound_jacobian}
    Let $\beta,\gamma\ge 2$ be constants such that $\beta\ge \gamma+2$. For every large enough $n$ and every $m\le n/4$. Fix a sequence of indices $(i_1,j_1),\dots,(i_m,j_m)$ of $\mbold$ and a sequence of indices $(u_1,v_1),\dots(u_m,v_m)\in [n]\times [n]$ of $\mbold^k$ such that $i_1,j_1,\dots, i_m,j_m,u_1,v_1,\dots,u_m,v_m$ are $4m$ pairwise distinct numbers. With probability $\ge 1-\exp(-\Omega(m))$, we have $\sigma_{0.9m}(\jbold)=\Omega(k^3\sqrt m/n^{2\beta+\gamma+1})$.

    In addition, the above holds true even if each element of $\jbold$ is perturbed additively by at most $O(k^5/n^{2+3\beta})$.
\end{lemma}

\begin{proof}
    Let $K:=\binom k3\frac{(n^2-1)^{k-3}}{n^{2k-6+2\beta}}=\Theta(k^3/n^{2\beta})$. We define four auxiliary matrices whose singular values are bounded, and their combination is exactly $\jbold$.
    Notice that the expectation of each element of $\mboldu$ is exactly $E:=\frac1n-\frac{1}{2n^\gamma}$.
    Let $\abold_1,\abold_2,\bbold\in\mbbr^{m\times m}$ such that for every $a,b\in[m]$, 
    $$\abold_1[a,b]:=6\sqrt 2n^{2\gamma}\left(\left(\mboldu[u_b,i_a]-E\right)\left(\mboldu[j_a,v_b]-E\right)+\left(\mboldu[u_b,j_a]-E\right)\left(\mboldu[i_a,v_b]-E\right)\right),$$
    $$\abold_2[a,b]:=\sqrt 3n^{\gamma}\left(\mboldu[u_b,i_a]+\mboldu[j_a,v_b]+\mboldu[u_b,j_a]+\mboldu[i_a,v_b]-4E\right),$$\
    $$\bbold[a,b]:=\frac{\partial\mbold^k[u_b,v_b]}{\partial\mbold[i_a,j_a]}+\frac{\partial\mbold^k[u_b,v_b]}{\partial\mbold[j_a,i_a]}-K(\mboldu[u_b,i_a]\mboldu[j_a,v_b]+\mboldu[u_b,j_a]\mboldu[i_a,v_b]).$$

    By Lemma~\ref{lem:partmkm}, we can rewrite the Jacobian matrix as
    $$\jbold=\frac{K}{6\sqrt 2n^{2\gamma}}\abold_1+\frac{K}{\sqrt 3n^\gamma}E\abold_2+2KE^2\mbbone\mbbone^T+\bbold$$
    where $\mbbone$ denotes the all-$1$ vector of length $m$.
    One can verify that both $\abold_1,\abold_2$ are random matrices satisfying the condition in Lemma~\ref{lem:singular_value_concentration}. Both $\abold_1,\abold_2$ have i.i.d. random elements because $\mboldu$ is a uniformly random symmetric matrix, and the indices $i,j,u,v$ do not repeat.

    By Ger\v{s}gorin circle theorem (Proposition~\ref{prop:gergorin_circle_theorem}), for every $i\in[m]$, we have
    \begin{equation}
    \label{eq:singular values B}
    \sigma_i(\bbold)=\sqrt{\lambda_i(\bbold^T\bbold)}\le O(mk^5/n^{2+3\beta}).
    \end{equation}

    We know from Lemma~\ref{lem:singular_value_concentration} that for every $\abold=\abold_1,\abold_2$.
    $$\Pr\left[C_1\sqrt m\le \sigma_{1}(\abold)\le C_2\sqrt m\right]\ge 1-\exp(-C_3m)$$
    and for every $l\in[m]$,
    $$\Pr\left[C_1l/\sqrt m\le \sigma_{m+1-l}(\abold)\le C_2l/\sqrt m\right]\ge 1-\exp(-C_3l).$$

    By Proposition~\ref{prop:thompson's interlacing},
    $$\sigma_{0.9m}\left(\frac{K}{\sqrt 3n^\gamma}E\abold_2+2KE^2\mbbone\mbbone^T\right)\ge \frac{K}{\sqrt 3n^\gamma}E\sigma_{0.9m+1}(\abold_2)$$

    By Weyl's inequality (Proposition~\ref{prop:weyl's inequality}) and triangle inequality, we have
    $$\sigma_{0.9m}(\jbold)\ge \frac{K}{\sqrt 3n^\gamma}E\cdot\sigma_{0.9m+1}(\abold_2)-\frac{K}{6\sqrt 2n^{2\gamma}}\cdot \sigma_1(\abold_1)-\sigma_1(\bbold).$$
    By inequality~(\ref{eq:singular values B}) and the concentration bounds on the singular values of $\abold_1,\abold_2$, the following holds with high probability
    \begin{align*}
        \sigma_{0.9m}(\jbold)\ge&\frac{K}{\sqrt 3n^\gamma}E\cdot 0.1C_1\sqrt m-\frac{K}{6\sqrt 2n^{2\gamma}}\cdot C_2\sqrt m-O(mk^5/n^{2+3\beta})
    \end{align*}
    which implies that $\sigma_{0.9m}(\jbold)=\Omega(K\sqrt m/n^{\gamma+1})=\Omega(k^3\sqrt m/n^{2\beta+\gamma+1})$
    when $\beta\ge \gamma+2$.

    The above holds true even if each element of $\jbold$ is perturbed additively by at most $O(k^5/n^{2+3\beta})$, since the absolute values of elements of $\bbold$ remains $O(k^5/n^{2+3\beta})$.
\end{proof}

\subsection{From singular values to the min-entropy: proof to Lemma~\ref{lem:tech_lem_amp}}

We have prepared to prove our technical lemma, which is to show that there exists a good input set $\msfm$ such that the joint min-entropy conditioned on any partial assignment of length $\le q=10^{-4}n^2$ is high. Our good input set $\msfm$ will be defined as inputs where the Jacobian matrices have well-controlled singular values. Before that, we partition the set of input elements into groups of \emph{core} elements of size $m$, so that for every partial assignment $\sigma$, we can always pick one such group of core elements that form a basis of the output elements.

\begin{definition}
\label{def:candidate_core_set}
For every window of \(3m\) consecutive queries from \(Q\), fix once and for all
one subset of \(m\) queries whose \(2m\) row and column indices are pairwise
distinct.  Such a subset exists by Observation~\ref{fact:amp_q_non_conflict}.
Let $(u_1,v_1),\ldots,(u_m,v_m)$ denote this fixed subset for the window.  We define \(\mathcal C\) to be an
arbitrary collection of disjoint sequences from
$$
    ([n]\setminus\{u_1,\ldots,u_m,v_1,\ldots,v_m\})
    \times
    ([n]\setminus\{u_1,\ldots,u_m,v_1,\ldots,v_m\})
$$
such that, for every $((i_1,j_1),\ldots,(i_m,j_m))\in\mathcal C$, the indices \(i_1,j_1,\ldots,i_m,j_m\) are pairwise distinct.  Moreover, we
choose \(\mathcal C\) so that
$$
    |\mathcal C|
    \ge
    \left\lfloor \frac{(n-2m)(n-2m-1)}{3m}\right\rfloor .
$$
We call \(\mathcal C\) the \emph{candidate core set} associated with this window.
\end{definition}

    

The existence of such a partition is guaranteed by the same construction as the query sequence $Q$, where by Fact~\ref{fact:amp_q_non_conflict}, the permutation of indices we defined as $Q$ guarantees that no $m$ consecutive elements share the same column or the same row. Given $\mcalc$, we can define our good input set $\msfm$.

\begin{definition}
We define \(\msfm\) as the set of inputs satisfying the following
condition.  For every window of \(3m\) consecutive queries from \(Q\), let $(u_1,v_1),\ldots,(u_m,v_m)$
be the fixed subset of queries chosen in Definition~\ref{def:candidate_core_set}, and let
\(\mathcal C\) be its candidate core set.  Then for every sequence $((i_1,j_1),\ldots,(i_m,j_m))\in\mathcal C$,
the Jacobian matrix between $\mbold[i_1,j_1],\ldots,\mbold[i_m,j_m]$ and $\mbold^k[u_1,v_1],\ldots,\mbold^k[u_m,v_m]$ satisfies
$$
    \sigma_{0.9m}(J)=\Omega\!\left(\frac{k^3\sqrt m}{n^{2\beta+\gamma+1}}\right).
$$
\end{definition}



\begin{lemma}
    \label{lem:amp_largeness_input_set}
    There exists a constant $c$ such that for every $m\in[c\log n,n/4]$, $\Pr_{\mbold\leftarrow\mcalm}[\mbold\in \msfm]>1-o(1)$.
\end{lemma}

\begin{proof}
Fix a window of \(3m\) consecutive queries from \(Q\), and consider the fixed
subset of \(m\) queries and the candidate core set \(\mathcal C\) from
Definition~\ref{def:candidate_core_set}.  By Lemma~\ref{lem:singular_bound_jacobian}, for every fixed sequence in \(\mathcal C\),
the corresponding Jacobian matrix has the desired singular-value bound with
probability at least \(1-\exp(-\Omega(m))\).

It remains to union bound over all events that appear in the definition of
\(\mathcal M\).  There are \(O(n^2)\) windows of \(3m\) consecutive queries in
\(Q\).  For each such window, the candidate core set \(\mathcal C\) contains
at most \(O(n^2/m)\) disjoint sequences.  Hence the total number of
Jacobian matrices over which we union bound is at most \(O(n^4/m)\), which is
polynomial in \(n\).  Therefore,
$$
    \Pr[\mbold\notin\mathcal M]
    \le
    O(n^4/m)\cdot \exp(-\Omega(m))
    =
    o(1),
$$
provided \(m\ge c\log n\) for a sufficiently large constant \(c\).
\end{proof}


Given guarantees that the Jacobian matrices have high singular values, given an arbitrary partial assignment $\sigma$ of length $\le q=10^{-4}n^2$, there always exists a sequence of input indices from $\mcalc$ where a large proportion of these input elements are unfixed by $\sigma$. Those elements, due to their nice properties in their Jacobian matrix with the output, form a basis to the output. In this way, the entropy of the input preserves to the output. Before proving Lemma~\ref{lem:tech_lem_amp}, we need the following definition and basic tools.

\begin{definition}
    Given an arbitrary partial assignment $\sigma$, the \emph{core indices} of $\mbold$ is a length-$(m':=0.9m)$ subsequence $((i_1,j_1),\dots,(i_{m'},j_{m'}))$ of an arbitrary sequence from $\mcalc$, such that $\mbold[i_1,j_1],\dots, \mbold[i_{m'},j_{m'}]$ and $\mbold[j_1,i_1],\dots, \mbold[j_{m'},i_{m'}]$ are not fixed by $\sigma$.
    We use $(\mbold)_\tcore:=(\mbold[i_1,j_1],\dots, \mbold[i_{m'},j_{m'}])$ to denote the vector representation of core elements. In addition, we use $(\mbold^k)_\tout:=(\mbold^k[u_1,v_1],\dots, \mbold^k[u_{m},v_m])$ to denote the vector representation of the selected output elements. Also, we always use $[\vlbold]_{\alpha'}$ to denote the vector where each element is rounded to its nearest multiple of $1/n^{\alpha'}$.
\end{definition}

We will show that, for every fixed (and rounded) output, the large singular values of the Jacobian matrix between core elements and output elements forces $\Omega(m)$ dimensions of the input to be concentrated. The min-entropy will be reduced to measuring the volume of a hyperrectangle in the uniform input space.

\begin{proposition}[Cauchy interlacing theorem for singular values \cite{horn2012matrix}]
    Let $\jbold\in \mbbr^{m\times n}$ a matrix where $m\le n$, and $\jbold'$ a matrix obtained by deleting one row from it. Then
    $$\sigma_1(\jbold)\ge \sigma_1(\jbold')\ge \sigma_2(\jbold)\ge \sigma_2(\jbold')\ge\dots\ge\sigma_{m-1}(\jbold)\ge \sigma_{m-1}(\jbold').$$
\end{proposition}

\begin{proposition}[Lagrange's mean value theorem]
    Let $f:\mbbr^{n}\rightarrow \mbbr$ be a continuous and differentiable function. For every two vectors $\albold,\blbold\in\mbbr^{n}$, there exists $\clbold=(1-t_0)\albold+t_0\blbold$ for some $t_0\in(0,1)$ such that
    $$f(\blbold)-f(\albold)=\nabla f(\clbold)\cdot (\blbold-\albold)$$
    where $\nabla f(\clbold)$ is the gradient of $f$ evaluated at the point $\clbold$.
\end{proposition}

\begin{proof}[Proof to Lemma~\ref{lem:tech_lem_amp}]
    We have shown in Lemma~\ref{lem:amp_largeness_input_set} that a random input drawn from $\mcalm$ is a good input with high probability. What remains is to show that the conditional joint min-entropy of the output is high for every partial assignment $\sigma$ of bounded length.

    Given a partial assignment $\sigma$ of length $\le q=10^{-4}n^2$, there exists a sequence $((i_1,j_1),\dots,(i_m,j_m))\in\mcalc$ where at most $0.1m$ elements of $\mbold$ or their symmetric elements are revealed by $\sigma$. Otherwise there would be $>\frac12\cdot\lfloor\frac{(n-2m)(n-2m-1)}{3m}\rfloor\cdot 0.1m>q$ elements revealed by $\sigma$. We fix an arbitrary such sequence, and the core indices $(i_1,j_1),\dots,(i_{m'},j_{m'})$ be its subsequence obtained by removing all the elements and the symmetric elements fixed by $\sigma'$. For every possible output vector $\ybold$, we have
    \begin{equation}
    \label{eq:hyperrectangle_1}
    \Pr_{\mbold\leftarrow \mcalm}\left[\left.[(\mbold^k)_\tout]_{\alpha'}=\ybold\textnormal{ and }\mbold\in\msfm\right|\sigma\right]\le \max_{\sigma'}\Pr_{\mbold\leftarrow \mcalm}\left[\left.[(\mbold^k)_\tout]_{\alpha'}=\ybold\textnormal{ and }\mbold\in\msfm\right|\sigma'\right]
    \end{equation}
    where the maximum is taken over all partial assignments $\sigma'$ on all the input elements other than core elements and their symmetric elements, such that $\sigma'$ is consistent with $\sigma$. We denote $f:\mbbr^{m'}\rightarrow \mbbr^m$ as the function of $(\mbold^k)_\tout$ given core elements, and all the other input elements fixed by $\sigma'$. And we denote by $f_i(\cdot)$ the $i$-th element of the output vector.

    Let $\|\cdot\|_\infty$ denote max norm, i.e., the maximum absolute value of all elements. Fix an arbitrary instance $\mbold'$ that is consistent $\sigma'$ such that $[(\mbold'^k)_\tout]_{\alpha'}=\ybold$. Then, for matrices $\mbold$ consistent with $\sigma'$ and with the same output, we have
    $$\|f((\mbold)_\tcore)-f((\mbold')_\tcore)\|_\infty\le 1/n^{\alpha'}.$$
    
    By the mean value theorem, there exists $m$ vectors $\clbold_1,\dots,\clbold_m$, such that for every $\ell\in[m]$,
    $$f_\ell((\mbold)_\tcore)-f_\ell((\mbold')_\tcore)=\nabla f_\ell(\clbold_\ell)\cdot ((\mbold)_\tcore-(\mbold')_\tcore)$$
    where $\nabla f_\ell(\clbold_\ell)$ is exactly the $\ell$-th column of the Jacobian matrix evaluated at $\clbold_\ell$. We let $\jbold\in\mbbr^{m'\times m}$ be the matrix where the $\ell$-th column is exactly $\nabla f_\ell(\clbold_\ell)$ for each $\ell\in[m]$. Then we have
    \begin{equation}
        \label{eq:hyperrectangle_2}
        1/n^{\alpha'}\ge \|f((\mbold)_\tcore)-f((\mbold')_\tcore)\|_\infty=\|\jbold^T((\mbold)_\tcore-(\mbold')_\tcore)\|_\infty.
    \end{equation}

    Because by Lemma~\ref{lem:partmkm}, when given $\sigma'$ fixed, the Jacobian matrix is almost fixed but the $O(k^5/n^{2+3\beta})$ term. We have
    $$\|\jbold-\jbold(\mbold)\|_\infty,\|\jbold-\jbold(\mbold')\|_\infty\le O(k^5/n^{2+3\beta})$$
    for $\jbold(\mbold),\jbold(\mbold')$ respectively the Jacobian matrix evaluated at $\mbold$ and $\mbold'$.
    The singular value bounds given in Lemma~\ref{lem:singular_value_concentration} also apply to $\jbold$, since, the as stated in Lemma~\ref{lem:singular_value_concentration}, the bound applies even when $\jbold$ is a matrix obtained by perturbing every element of a Jacobian matrix by at most $O(k^5/n^{2+3\beta})$.

    Let $m'':=0.8m$, by Lemma~\ref{lem:singular_value_concentration} and Cauchy interlacing theorem, $\jbold\in\mbbr^{m'\times m}$ between the core elements in $\mbold$ and $(\mbold^k)_\tout$ satisfies: for every $\ell\in[0.8m]$, 
    \begin{equation}
    \label{eq:hyperrectangle_2.5}
    \sigma_\ell(\jbold)=\Omega(k^3\sqrt m/n^{2\beta+\gamma+1}).
    \end{equation}

    Consider the singular value decomposition $\jbold=\ubold\sigmabold\vbold^T$ for orthogonal matrices $\ubold\in\mbbr^{m'\times m'}, \vbold\in\mbbr^{m\times m}$ and a diagonal matrix $\sigmabold$ where $\sigmabold[\ell,\ell]=\sigma_\ell(\jbold)$ for every $\ell\in[m']$. Denote by $\ulbold_1,\dots, \ulbold_{m'}$ the columns of $\ubold$, which are orthonormal vectors. Let $((\mbold)_\tcore-(\mbold')_\tcore)=\sum_{\ell=1}^{m'}\alpha_\ell\ulbold_\ell$ for some $\alpha_1,\dots,\alpha_{m'}\in\mbbr$. Then we have
    \begin{align*}
        \|\jbold^T((\mbold)_\tcore-(\mbold')_\tcore)\|_\infty&\ge \frac1{\sqrt m}\|\jbold^T((\mbold)_\tcore-(\mbold')_\tcore)\|_2\\
        &=\frac1{\sqrt m}\|\vbold\sigmabold^T \ubold^T(\sum_{\ell=1}^{m'}\alpha_\ell\ulbold_\ell)\|_2\\
        &=\frac1{\sqrt m}\|\sigmabold^T \ubold^T(\sum_{\ell=1}^{m'}\alpha_\ell\ulbold_\ell)\|_2\\
        &=\sqrt{\frac{{\sum_{\ell=1}^{m'}\sigma_\ell(\jbold)^2\alpha_\ell^2}}{m}}
    \end{align*}
    where the first inequality is by Cauchy-Schwarz, the third line is because the $L_2$ norm remains the same under orthogonal transformations.
    
    Notice that for every $\ell\in[m']$, $\alpha_\ell=\ulbold_\ell\cdot ((\mbold)_\tcore-(\mbold')_\tcore)$.
    Combining the above with inequality~(\ref{eq:hyperrectangle_2}) and equation~(\ref{eq:hyperrectangle_2.5}), for every $\ell\in[m'']$, we have
    $$((\mbold)_\tcore-(\mbold')_\tcore)^T\ulbold_\ell\le \frac{\sqrt m}{n^{\alpha'}\sigma_\ell(\jbold)}\le O(n^{2\beta+\gamma-\alpha'+1}/k^3).$$

    Therefore, given $\ybold$ and $\sigma'$, the possible input set is contained in the solution set of the above $m''$ inequalities, whose volume is proportional to the volume of the hyperrectangle with $m''$ dimensions concentrated in the uniform input space.

    Notice that the total volume of possible inputs $\mbold$ given $\sigma'$ is $1/n^{(\beta+\gamma)m'}$. For every pair of different possible matrices consistent to $\sigma'$, their maximum $L_2$ distance is at most $\sqrt{m'}/n^{\beta+\gamma}$. Therefore, the maximum volume of the hyperrectangle is upper-bounded by
    $$(O(n^{2\beta+\gamma-\alpha'+1}/k^3))^{m''}\cdot (\sqrt {m'}/n^{\beta+\gamma})^{m'-m''}.$$
    Recall that $m''=0.8m$, and $m'=0.9m$. By inequality (\ref{eq:hyperrectangle_1}),
    \begin{align*}
    &\Pr_{\mbold\leftarrow \mcalm}\left[\left.[(\mbold^k)_\tout]_{\alpha'}=\ybold\textnormal{ and }\mbold\in\msfm\right|\sigma\right]\\
    \le &(O(n^{2\beta+\gamma-\alpha'+1}/k^3))^{m''}\cdot (\sqrt m/n^{\beta+\gamma})^{m'-m''}\cdot n^{(\beta+\gamma)m'}\\
    \le&\exp(m''(2\beta+\gamma-\alpha'+1)\ln n)\cdot \exp((m'-m'') (0.5-\beta-\gamma)\ln n)\cdot \exp(m'(\beta+\gamma)\ln n)\\
    \le&\exp((2.4\beta+1.6\gamma-0.8\alpha'+0.85)m\ln n)\\
    \le& \exp(-2r)
    \end{align*}
    while $\alpha'\ge3\beta+2\gamma+3.6$ and $m=\frac{r}{\ln n}<n$.

    Since the conditional min-entropy of $m$ selected output elements is at least $2r$, the conditional min-entropy of all the $3m$ output elements is also at least $2r$.
\end{proof}

\section{Lower bounds to approximate matrix powering with superlinear-in-\texorpdfstring{$n$}{n} redundancy}

The above approach does not easily extend to proofs to $r\cdot t=\Omega(n^2)$ lower bounds for $r=\omega(n\ln n)$ because the Jacobian matrix of $\omega(n)$ core elements and output elements can not be easily controlled. Our analysis above relies on the fact that the core elements, output elements, and their symmetric elements do not share rows and columns. To that end, we use a different input distribution which leads to the same trade-off for $r\in[10n\log n, n^2/1000]$. Notably, the lower bound derived using this distribution also implies a lower bound for the fundamental \emph{counting set intersection} problem in data structures.

\paragraph{Nemesis input distribution.} Fix $n\ge 1$ a large enough integer. Fix a partition $V_1:=\{1,2,\dots,n/3\},V_2:=\{n/3+1,\dots, 2n/3\},V_3:=\{2n/3+1,\dots, n\}$ to $[n]$. Let $\mcalm$ be a distribution of $n\times n$ symmetric matrices $\mbold$ with the following conditions.
\begin{enumerate}
    \item For every $(u,v)\in (V_1\times V_2)\cup (V_2\times V_3)$,
    $$\mbold[u,v]=\mbold[v,u]=\begin{cases}
        1/n^3&\textnormal{with independent probability }1/2\\
        0&\textnormal{with independent probability }1/2.
    \end{cases}$$
    \item For every $u\in [n]$, $\mbold[u,u]=1-1/n^2$.
    \item For every other $(u,v)\in[n]\times [n]$, $\mbold[u,v]=0$.
\end{enumerate}

In other words, its underlying graph is a three-layered undirected graph with heavy lazy transitions. For each pair of vertices of consecutive layers, we connect them with independent probability $1/2$. We have the following observation.

\begin{fact}
    \label{fact:symmetric_matrices}
    For every pair of vertices $(u,v)\in V_1\times V_3$ and every $2\le k\le n$, if there are exactly $a$ length-$2$ simple paths from $u$ to $v$,
    $$\mbold^k[u,v]=\binom k2\left(\frac{n^2-1}{n^2}\right)^{k-2}\frac{a}{n^6}+O(k^4/n^9).$$
\end{fact}

\begin{proof}
    Let $a_2,\dots, a_k$ respectively denote the number of length-$2,\dots, k$ walks from $u$ to $v$ without self loops. Then each $a_i=O(n^{i-1})$. Specifically, for every odd $i$, $a_i=0$. Because the underlying graph of $M$ without self-loops is a bipartite graph where $u,v$ belong to the same part. We get
    $$\mbold^k[u,v]=\binom k2\left(\frac{n^2-1}{n^2}\right)^{k-2}\frac {a_2}{n^6}+\sum_{i=4}^k\binom ki\left(\frac{n^2-1}{n^2}\right)^{k-i}\frac {a_i}{n^{3i}}$$
    where the summation is upper bounded by the sum of a geometric series, and is dominated by the case $i=4$.
\end{proof}

Since $k^4/n^9=o(k^2/n^6)$, approximating the power of $\mbold$ can be used to answer the number of length-$2$ paths between vertices from $V_1$ and $V_3$ exactly. In this way, we get rid of the matrix, and are able to focus on a clean combinatorial task. In fact, we will show that even computing the parity of $\mbold^k[u,v]$ requires high data structure complexity.

This problem can also be rephrased in terms of counting set intersection: for every pair of vertices $(u,v)\in V_1\times V_3$, the number of length-$2$ paths from $u$ to $v$ is exactly the size of the intersection of their neighbor sets.

The sequence of queries that we construct is similar to the above. We let it to be a concatenation of $n/3$ disjoint bijections between $V_1$ and $V_3$:
$$Q=(\underbrace{(1,2n/3+1),(2,2n/3+2),\dots,(n/3,n)}_{\textnormal{identity bijection}},\dots,\underbrace{(1,n),(2,2n/3+1),\dots,(n/3,n-1)}_{\textnormal{shifted bijection}})$$

Formally, for every $i\in[n/3]$, the $i$-th bijection is exactly a sequence of pairs $(u,v)$ where
\begin{itemize}
    \item $u\in V_1$;
    \item $v=(u+i-2)\bmod (n/3)+2n/3+1.$
\end{itemize}

The following fact is immediate from our construction.

\begin{fact}
    \label{fact:degree_bound}
    For every consecutive $m\ge n$ queries in $Q$, there does not exist an index $v\in[n]$ that is covered more than $\lceil\frac{m}{n/3}\rceil$ times.
\end{fact}

Even though the input distribution is a clean uniform distribution, the conditional min-entropy remains non-trivial to analyze. Our strategy, analogous to the previous section, is to find a ``basis'' from the input elements to the outputs. We show that we can always find such a basis for a large set $\msfm$ of ``good inputs''.

\begin{definition}
    Let $\hat \mbold\in\mbbf_2^{n\times n}$ denote the binary matrix where each element $\hat \mbold[i,j]=1$ if and only if $\mbold[i,j]\ne 0$. 

    Let $\msfm\subseteq \mbbr^{n\times n}$ be a set of ``good'' input matrices $\mbold$ such that for every $d\in[10\log n, n/3]$, and every $V_1'\subseteq V_1,V_2'\subseteq V_2$, $|V_1'|=|V_2'|=d$, the submatrix has a high rank:
    $$\rank(\hat \mbold[V_1',V_2'])\ge d/2.$$
\end{definition}

Notice that the submatrix is a uniformly random square matrix, where there are classic results bounding its rank.

\begin{proposition}[page 38 of \cite{belsley1993rates}; see also page 1 of \cite{fulman2015stein}]
    \label{prop:random_f2_matrices}
    Let $\hat\mbold$ be a uniformly random $n\times n$ matrix over $\mbbf_2$. For every $k\in\{0,\dots, n\}$,
    $$\Pr[\rank(\hat\mbold)=n-k]=\frac{1}{2^{k^2}}\frac{\prod_{i=1}^n(1-1/2^i)\prod_{i=k+1}^n(1-1/2^i)}{\prod_{i=1}^{n-k}(1-1/2^i)\prod_{i=1}^k(1-1/2^i)}.$$
\end{proposition}

Formally, the following lemma shows that the input matrix is good with high probability.

\begin{lemma}
    \label{lem:good_matrices}
    With probability $\ge 1-o(1)$, a random input matrix from $\mcalm$ is good.
\end{lemma}

\begin{proof}
    For every $d\in[10 \log n, n/3]$, there are $\binom {n/3}d^2\le (\frac{en/3}{d})^{2d}$ submatrices of dimension $d\times d$. For each submatrix $\hat\mbold[V_1',V_2']$, and every $k\in[d/2+1,\dots, d]$, by Proposition~\ref{prop:random_f2_matrices} we have
    $$\Pr[\rank(\hat\mbold[V_1',V_2'])=d-k]=\frac{1}{2^{k^2}}\frac{\prod_{i=1}^d(1-1/2^i)\prod_{i=k+1}^d(1-1/2^i)}{\prod_{i=1}^{d-k}(1-1/2^i)\prod_{i=1}^k(1-1/2^i)}\le 2^{d-d^2/4}.$$
    Therefore,
    $$\Pr[\rank(\hat\mbold[V_1',V_2'])\le d/2]\le (d/2)2^{d-d^2/4}.$$
    By union bound, the probability that at least one of the submatrices have a low rank is
    \begin{align*}
        \Pr[\mbold\textnormal{ is bad}]\le \sum_{d=10\log n}^{n/3}(\frac{en/3}{d})^{2d}(d/2)2^{d-d^2/4}\le \sum_{d=10\log n}^{n/3}2^{-5\log^2n}=o(1).
    \end{align*}
\end{proof}

Equipped with the above tools, we are ready to prove the lower bound to AMP with high redundancy.

\begin{theorem}[Lower bound to AMP with superlinear-in-$n$ redundancy]
    \label{thm:amp_lb_large_s}
    Fix $2\le k\le n$ and a constant $\alpha>6$. For every (possibly randomized) succinct and systematic data structure with $r\in[10n\log n, n^2/1000]$ bits of redundancy that answers the whole approximated matrix power $[\mbold^k]_\alpha$ correctly with $\ge 99/100$ probability, its total number of probes is $\Omega(n^4/r)$ in the worst case.
\end{theorem}

\begin{proof}
    For each $(u,v)\in V_1\times V_3,$ to approximate $[\mbold^k]_\alpha$ for $\alpha>6$, by Fact~\ref{fact:symmetric_matrices}, the algorithm should also be able to answer exactly the number of length-$2$ paths from $u$ to $v$.

    By Yao's minimax principle, there is a randomized data structure that answer the fixed sequence of queries $Q$ with high probability efficiently for every possible input only if there exists a deterministic data structure that with the same success probability and the same complexity given a random input drawn from $\mcalm$. In below, we focus on deterministic data structures and query-with-sketch algorithms.

    By Lemma~\ref{lem:reduction_ds_qws} and Lemma~\ref{lem:min-entr}, the lower bound reduces to showing that for every block of $m=20r$ consecutive queries from $Q$, and for every partial assignment $\sigma$ of length at most $q:=0.001n^2$ to $\mbold$, the joint min-entropy of the output conditioned on $\sigma$ is $\ge 2r$. We will focus on the min-entropy of the unrounded elements of $\mbold^k$ in below, as it is equivalent to the min-entropy of the rounded elements of $[\mbold^k]_\alpha$, given our input distribution $\mcalm$. Besides, in applying Lemma~\ref{lem:min-entr}, we only pick a good input set $\msfm$ constructed above, and set $\msfm_\sigma=\mbbr^{n\times n}$ for every $\sigma$.

    By our construction, the input matrix $\mbold$ is fixed other than two submatrices, $\mbold[V_1\times V_2]$ and $\mbold[V_2\times V_3]$.
    Given a partial assignment of length at most $q$, there are at most $0.1n$ vertices from $V_1\cup V_3$, where each of their corresponding rows (if in $V_1$) or columns (if in $V_3$) has $\ge 0.01n$ elements revealed from $\sigma$. We call these vertices ``bad vertices'', and other vertices ``good vertices''.
    Let $Q'=((u_1,v_1),\dots,(u_{m'},v_{m'}))$ be the subsequence of the $m$ consecutive queries containing only pairs of good vertices. By Fact~\ref{fact:degree_bound}, we have 
    $$m'\ge m-0.1n\cdot (\frac{m}{n/3}+1 )\ge 0.69m$$

    There exists a sequence $(w_1,\dots, w_{m'})$ such that the sequence of triples $(u_1,v_1,w_1),\dots(u_{m'},v_{m'},w_{m'})$ are \emph{edge-disjoint} triples where none of the edges $(u_i,w_i),(w_i,v_i)$ are revealed by $\sigma$, which means that none of pairs of indices $(u,v),(v,w)$, or $(u,w)$ belong to multiple triples. Its existence directly follows from a construction: we decide each $w_i$ one by one. For each $(u_i,v_i)$, given $w_{<i}$ fixed, among $n/3$ indices from $V_2$, at least $n/3-2\lceil\frac{m}{n/3}\rceil-0.02n\ge 0.19n$ elements of $\mbold[u_i,\cdot]$ and $\mbold[\cdot, v_i]$ are unfixed by $\sigma$ and $w_{<i}$, which means there always exists an $w_i$ that satisfies our condition. 

    Observe that for every partial assignment $\sigma$ of length $\le q$,

    \begin{equation*}
    \hmin(\mbold^k[u_1,v_1],\dots,\mbold^k[u_{m'},v_{m'}],\mbold\in\msfm|\sigma)\ge \min_{\sigma'}\hmin(\mbold^k[u_1,v_1],\dots,\mbold^k[u_{m'},v_{m'}],\mbold\in\msfm|\sigma'])
    \end{equation*}
    where $\sigma'$ is selected from all possible partial assignments to $\mbold$ that is consistent to $\sigma$ and reveals all the elements other than the $m'$ elements $\mbold[w_1,v_1],\dots,\mbold[w_{m'},v_{m'}]$.

    Fix a vertex $v\in V_3$ such that there are $\ge 10\log n$ queries on $v$ in $Q'$. Since $m'\ge 0.69m\ge 13.8r\ge 138n\log n$, at least $\ge 202m'/207$ queries from $Q'$ are made on such vertices $v\in V_3$. Or otherwise there will be less than
    $$\frac{n}3\cdot 10\log n+\frac{202m'}{207}\le m'$$
    queries in $Q'$.

    Let $d\ge 10\log n$ be the number of queries on $v$ in $Q'$, and $(u_1',v,w_1'),\dots,(u_d',v,w_d')$ the triples. By our construction of $w$ and $Q$, none of the vertices $u_i'\in V_1$ and $w_i'\in V_2$ repeat in the $d$ triples. In addition, conditioned on $\sigma'$, different columns of the matrix power $\mbold^k[\cdot, v]$ are independent. The outputs $\mbold^k[u_1',v],\dots,\mbold^k[u_d',v]$ only depend on variables $\mbold[w_1',v],\dots, \mbold[w_d',v]$. More precisely, each $\mbold^k[u_i',v]$ is a linear combination of the variables: for every $i\in[d]$,
    $$\mbold^k[u_i',v]=\underbrace{\sum_{w:\forall j,w\ne w_j'}\mbold[u_i',w]\mbold[w,v]}_{\textnormal{fixed given }\sigma'}+\sum_{j=1}^d\mbold[u_i',w_j']\mbold[w_j',v]$$
    
    We only need to show that the set of vectors $\{(\hat\mbold[u_i',w_j'])_{j\in[d]}\}_{i\in[d]}$ form a basis of rank at least $d/2$. If it is true, the parity of the output elements, which can always be presented as a linear combination of these vectors, will have $\ge d/2$ min-entropy.

    Formally, let $\ybold\in \mbbf^{d}_2$ be the vector where $\ybold(i)=\hat\mbold^k[u_i',v]$. Let $\xbold\in \mbbf^d_2$ be the vector where $\xbold(i)=\hat \mbold[w_i,v]$. Let $V_1'=\{u_1',\dots,u_d'\}$ and $V_2'=\{w_1',\dots,w_d'\}$ Then for every $\sigma'$, there exists a vector $\wlbold\in\mbbf^d_2$ such that
    $$\ybold=\wlbold+\hat\mbold[V_1',V_2']\xbold$$
    where $\xbold$ is a uniformly random vector over $\mbbf^d_2$, given $\sigma'$. Therefore,
    \begin{align*}
    &\hmin(\mbold^k[u_1',v],\dots,\mbold^k[u_{d}',v],\mbold\in\msfm|\sigma')\\
    \ge&\min_{\sigma'}\hmin(\ybold|\sigma')\\
    =&\min_{\sigma'}\hmin(\hat\mbold[V_1',V_2']\xbold|\sigma')\\
    =&\min_{\sigma'}\rank(\hat \mbold[V_1',V_2'])\\
    \ge&d/2.
    \end{align*}
    Here, the min-entropy of the output is greater than or equal to the min-entropy of $\ybold$ because by Fact~\ref{fact:symmetric_matrices}, $\mbold^k[u,v]$ determines the number of length-$2$ walks from $u$ to $v$, and hence its parity.
    
    Therefore, each vertex $v$ of $d$ queries contribute $\ge d/2$ min-entropy. In addition, different columns $\mbold^k[\cdot,v]$ of the output are independent given $\sigma'$. The total min-entropy of the output is at least $\frac12\cdot \frac{202m'}{207}$, where $\frac{202m'}{207}$ is the least number of queries from $Q'$ that are made on these vertices $v$. We conclude that 
    \begin{equation*}
    \hmin(\mbold^k[u_1,v_1],\dots,\mbold^k[u_{m'},v_{m'}],\mbold\in\msfm|\sigma)\ge \frac{101m'}{207}\ge 6.7r>2r
    \end{equation*}

    By Lemma~\ref{lem:reduction_ds_qws} and Lemma~\ref{lem:min-entr}, the lower bound follows.

\end{proof}

The lower bound to the counting set intersection follows through a reduction from AMP.

\begin{corollary}[Lower bound to counting set intersection]
    \label{thm:csi_lb}
    Given a (possibly randomized) succinct and systematic data structure with $r\in[10n\log n, n^2/1000]$ bits of redundancy that given as input $2n/3$ subsets of $[n/3]$ answers their pairwise set intersection size correctly with $\ge 99/100$ probability, its total number of probes is $\Omega(n^4/r)$ at the worst case.
\end{corollary}

We note that no super-linear-in-$n$ amortized probe lower bound for counting set intersection can be obtained. Because the intersection size of every pair of sets can be answered in $O(n)$ probes.

\begin{proof}
    Let the input sets follow i.i.d. uniform distribution from $\{0,1\}^{n/3}$. Counting set intersection is exactly counting the number of length-$2$ paths given an input from $\mcalm$, which we have established a lower bound above.
\end{proof}

\section{Lower bounds to set disjointness and \texorpdfstring{$(2-\varepsilon)$}{(2-eps)}-approximate APSP}
\label{sec:set disjointness}

In this section, we study the lower bound to the set disjointness, which is the decision version of counting set intersection. The same nemesis distribution and analysis also implies the same probe-redundancy trade-off for the $(2-\varepsilon)$-approximation of APSP. We state our proof in the context of set disjointness.

\paragraph{Set disjointness problem.} Given as input $n$ sets $A_1,\dots,A_n\subseteq[n]$. The succinct and systematic data structure, given a pair of indices $(u,v)\in[n]\times [n]$, must answer whether $|A_u\cap A_v|=0$. For simplicity, we denote the answer as $Ans_{u,v}$, which is a binary bit that is $1$ iff the answer is NO.

For the convenience of our analysis, each $A_u$ also denotes a binary vector of length $n$, where the $w$-th bit denotes whether $w\in A_u$.

The uniform input distribution is no longer hard for this problem: when each vector is i.i.d. uniformly from $\{0,1\}^n$, with high probability for every $u,v$, $|A_u\cap A_v|\ne 0$. To that end, we construct the following nemesis distribution where $1$s are sparse in the input vectors.

\paragraph{Nemesis input distribution.} Let $\mcala$ denote the input distribution that, for every $u,i\in[n]$, $A_u(i)=1$ with independent probability $1/\sqrt n\log n$. 

One can see that $\Pr[Ans_{u,v}=1]=\Theta(1/\log^2n)$ for every $u\ne v$.

We use a similar sequence $Q$ of queries as above. Let $V_1:=\{1,\dots, n/2\}$ and $V_2:=\{n/2+1,\dots, n\}$.
We let $Q$ to be a concatenation of $n/2$ disjoint bijections between $V_1$ and $V_2$:
$$Q=(\underbrace{(1,n/2+1),(2,n/2+2),\dots,(n/2,n)}_{\textnormal{identity bijection}},\dots,\underbrace{(1,n),(2,n/2+1),\dots,(n/2,n-1)}_{\textnormal{shifted bijection}})$$

Formally, for every $i\in[n/2]$, the $i$-th bijection is exactly a sequence of pairs $(u,v)$ where
\begin{itemize}
    \item $u\in V_1$;
    \item $v=(u+i-2)\bmod(n/2)+n/2+1$.
\end{itemize}

The main result in this section is the following lower bound.

\begin{theorem}[Lower bound to set disjointness]
    \label{thm:sd_lb}
    Given a (possibly randomized) succinct and systematic data structure with $r\in[10n\log^2 n, n\sqrt n/1000\log^3 n]$ bits of redundancy that given as input $n$ subsets of $[n]$ answers the sequence of queries $Q$ correctly with $\ge 99/100$ probability, its total number of probes is $\Omega(n^4/(r\log^3n))$ in the worst case.
\end{theorem}

Different from the previous lower bounds to AMP, the answers cannot be formulated as (nearly-)linear combinations of selected input elements, where analytic or linear algebraic tools may apply. However, the sparsity of the sets still enables us to lower bounding the min-entropy. Formally, we define the following set of ``good inputs''.

\begin{definition}
    \label{def:sd_good_input}
    Let $\msfa$ be the set of good inputs $\{A=(A_1,\dots, A_n)\}$ satisfying the following conditions
    \begin{enumerate}[label=(\arabic*)]
        \item For every $u\in[n]$,
        $$0.99\frac{\sqrt n}{\log n}\le |A_u|\le 1.01\frac{\sqrt n}{\log n}.$$
        \item For every $k\in[\log^2 n,\sqrt n]$, $u\in V_1$ and $v\in V_2$,
        $$\sum_{i=0}^{k-1}Ans_{u+i,v}\le \frac{2k}{\log n}$$
        where $Ans_{u+i,v}$ denotes $Ans_{u+i-n/2,v}$ if $u+i>n/2$.
        \item For every $k\in[\log^2 n,\sqrt n]$, every $u\in V_1$ and every $w\in [n]$,
        $$\sum_{i=0}^{k-1}A_{u+i}(w)\le 2\log n$$
        where $A_{u+i}(w)$ denotes $A_{u+i-n/2}(w)$ if $u+i>n/2$.
    \end{enumerate}
\end{definition}

We will show that, for most of indices $v\in V_2$, for a fixed assignment to all the $A_1,\dots, A_{n/2}$, and for a fixed sequence of answers $Ans_{\cdot,v}$, a large proportion of bits in $A_v$ are forced to be $0$ if the input is good, which results in a high conditional min-entropy.

We use the following form of Chernoff bound to show that a random input is good with high probability.

\begin{proposition}[simplified Chernoff bound]
    \label{prop:chernoff}
    Let $X_1,X_2,\dots, X_n$ be independent random variables such that $X_i\in [0,1]$. Let $X=\sum_{i=1}^nX_i$. Then, for every $\delta>0$,
    $$\Pr[|X-\mbbe[X]|\ge \delta\cdot \mbbe[X]]\le 2\cdot \exp\left(-\frac{\delta^2}{2+\delta}\cdot \mbbe[X]\right).$$
\end{proposition}

\begin{lemma}
    \label{lem:sd_good_input}
    For every large enough $n$, a random input from $\mcala$ is good with probability at least $99/100$.
\end{lemma}

\begin{proof}
    We will show that each condition holds with low probability. The desired probability bound follows from the union bound.

    For condition (1), we know from our input construction that $\mbbe[|A_u|]=\frac{\sqrt n}{\log n}$. By Chernoff bound,
    $$\Pr\left[\left|A_u-\frac{\sqrt n}{\log n}\right|\ge0.01\frac{\sqrt n}{\log n}\right]\le \exp(-\Omega(\sqrt n/\log n)).$$
    By union bound over all $u$, with $\ge 1-o(1)$ probability condition (1) holds.

    Then, we show that condition (2) holds with high probability conditioned on (1) is true. Fix an index $v\in V_2$ and an assignment $\alpha$ to $A_v$ where $|\alpha|\in [0.99\frac{\sqrt n}{\log n},1.01\frac{\sqrt n}{\log n}]$. Conditioned on a fixed $A_v=\alpha$, for every $u$ and $i$, $Ans_{u+i,v}$ are independent random variables where
    $$\Pr[Ans_{u+i,v}=0|A_v=\alpha]=\left(1-\frac{1}{\sqrt n\log n}\right)^{|\alpha|}.$$
    By Fact~\ref{fact:e_inequality} and the Taylor expansion $\exp(-x)=1-x+\frac{x^2}{2}-\dots$, we have
    $$\Pr[Ans_{u,v}=1|A_v=\alpha]\ge 1-\exp\left(-\frac{|\alpha|}{\sqrt n\log n}\right)\ge \frac{|\alpha|}{\sqrt n\log n}-\frac{|\alpha|^2}{2n\log^2n}\ge \frac{0.98}{\log^2n}$$
    and
    $$\Pr[Ans_{u,v}=1|A_v=\alpha]\le1-\exp\left(-\frac{|\alpha|}{\sqrt n\log n-1}\right)\le \frac{|\alpha|}{\sqrt n\log n-1}\le \frac{1.02}{\log^2n}.$$
    
    By Chernoff bound,
    $$\Pr\left[\left.\sum_{i=0}^{k-1}Ans_{u+i,v}> \frac{2k}{\log n}\right|A_v=\alpha\right]\le \exp(-1.9k/\log n)\le O(n^{-2.7})$$
    By union bound over $u,v$, and $k$, with $\ge 1-o(1)$ probability condition (2) holds.

    For condition (3), notice that $\mbbe[\sum_{i=0}^{k-1}A_{u+i}(w)]=\frac{k}{\sqrt n\log n}$. By Chernoff bound, for every $k,u,w$,
    $$\Pr\left[\sum_{i=0}^{k-1}A_{u+i}(w)> 5\log n\right]\le O(n^{-2.7}).$$
    Again, by union bound, condition (3) holds with $\ge 1-o(1)$ probability.
    
\end{proof}

However, the above set of good inputs still does not guarantee a high conditional min-entropy. Imagine a magical algorithm that finds all the $1$s in all the strings in $ O(n\sqrt n/\log n)$ probes. No entropy would remain in the input. To that end, we construct good input sets $\msfa_\sigma$ for every partial assignment $\sigma$ of length $\le q$. Intuitively, the good input sets $\msfa_\sigma$ ensures that high uncertainty remains in the unprobed part of the input.

\begin{definition}
    \label{def:sd_good_input_sigma}
    For every partial assignment $\sigma$ of length $\le 10^{-4}n^2$, let $\msfa_\sigma$ be the set of good inputs $\{A=(A_1,\dots, A_n)\}$ such that for every shifted bijection $(1,v_1),\dots,(n/2,v_{n/2})$ in the sequence of queries $Q$, no more than $0.15n$ indices $v_i$ in $V_2$ satisfy the following condition: let $W_{i,v_i}\subseteq[n]$ denote the set of indices $w$ such that $A_i(w)=1$ and $A_{v_i}(w)$ is revealed by $\sigma$; then $|W_{i,v_i}|> 0.5\sqrt n/\log n$.
\end{definition}

To show that a random input falls in $A_\sigma$ with high probability conditioned on $\sigma$, we will use the following standard Chernoff bound.

\begin{proposition}[standard Chernoff bound~\cite{motwani1996randomized}]
    \label{prop:standard_chernoff}
    Let $X_1,\dots, X_n$ be independent random variables with $X_i\in\{0,1\}$ and $\Pr[X_i=1]\in(0,1)$ for $i\in[n]$. Then, for $X=\sum_{i=1}^n X_i, \mu=\mbbe[X]$, and any $\delta>0$,
    $$\Pr[X>(1+\delta)\mu]<\left(\frac{e^\delta}{(1+\delta)^{(1+\delta)}}\right)^\mu.$$
\end{proposition}

\begin{lemma}
    \label{lem:sd_good_input_sigma}
    For every partial assignment $\sigma$ of length $\le 10^{-4}n^2$, we have
    $$\Pr_{A\leftarrow \mcala}[A\in \msfa-\msfa_\sigma|\sigma]\le \exp(-0.0048n\sqrt n/\log n).$$
\end{lemma}

\begin{proof}
    We show that for every fixed shifted bijection $(1,v_1),\dots, (n/2,v_{n/2})$, the condition is satisfied with high probability. And we conclude with a union bound over all the $n$ shifted bijections.

    Given $\sigma$ of length at most $q$, there are at most $0.01n$ vectors in which $\ge 0.01n$ bits are revealed from $\sigma$. Without loss of generality, we assume that the first $0.49n$ pairs $(1,v_1),\dots, (0.49n,v_{0.49n})$ do not contain such vectors.

    For each of these pairs $(i,v_i)$ for $i\in[0.49n]$, there are $\ge 0.98n$ indices $w\in[n]$ such that neither of $A_i(w)$ nor $A_{v_i}(w)$ are revealed by $\sigma$. Without loss of generality, we assume these indices are $w=1,\dots, 0.98n$.
    When $A\in\msfa$ and $|W_{i,v_i}|>0.5\sqrt n/\log n$, we should have $\sum_{w=1}^{0.98n} A_i(w)< 0.51\sqrt n/\log n$, by $|A_i|\le 1.01{\sqrt n}/{\log n}$. We apply the simplified Chernoff bound (Proposition~\ref{prop:chernoff}) with $\mbbe[\sum_{w=1}^{0.98n} A_i(w)|\sigma]=0.98\sqrt n/\log n$, and obtain
    $$\Pr[\sum_{w=1}^{0.98n} A_i(w)< 0.51\sqrt n/\log n|\sigma]\le 2\cdot \exp(-\frac{47}{243}\cdot 0.51\sqrt n/\log n)\le 2\cdot \exp(-0.098\sqrt n/\log n).$$

    Let $X_1,\dots, X_{0.49n}\in\{0,1\}$ be the indicators of whether $|W_{i,v_i}|>0.5\sqrt n/\log n$. Then, conditioned on $\sigma$, these indicators are independent random bits, each evaluating to $1$ with probability $\le 2\cdot \exp(-0.098\sqrt n/\log n)$. $A\in \msfa-\msfa_\sigma$ only if $X:=\sum_{i=1}^{0.49n}X_i>0.05n$. By applying the standard Chernoff bound (Proposition~\ref{prop:standard_chernoff}) with $(1+\delta)\mu=0.05n$ and $\mu\le 0.98n\exp(-0.098\sqrt n/\log n)$, we get
    $$\Pr[X>0.05n]<\exp(-0.0048n\sqrt n/\log n).$$
    The desired inequality follows by a union bound over all $n/2$ shifted bijections.
\end{proof}

Given an input with the above conditions, we are able to prove the lower bound.

\begin{proof}[Proof to Theorem~\ref{thm:sd_lb}]
    Again, by Yao's minimax principle, we may focus on proving lower bounds for data structures given a random input from $\mcala$.
    By Lemma~\ref{lem:reduction_ds_qws} and Lemma~\ref{lem:min-entr}, the lower bound reduces to showing that for every $m:=120r\log^3n$ consecutive queries $(u_1,v_1),\dots,(u_m,v_m)$ from $Q$, and for every partial assignment $\sigma$ of length at most $q:=10^{-4}n^2$ to $\mbold$, the joint min-entropy of the output conditioned on $\sigma$ is $\ge 2r$. We may always assume that $m$ is a multiple of $n/2$, as this will only increase $m$ by at most $n/2=o(r\log^3n)$. In addition, this ensures that the consecutive $m$ queries is a concatenation of shifted bijections. As shown in Lemma~\ref{lem:sd_good_input} and Lemma~\ref{lem:sd_good_input_sigma}, the input falls in $\msfa$ with high probability, and falls in $\msfa_\sigma$ with high probability for every $\sigma$.

    Observe that for every partial assignment $\sigma$ of length $\le q$,
    $$\hmin(Ans_{u_1,v_1},\dots,Ans_{u_m,v_m},A\in\msfa\cap \msfa_\sigma|\sigma)\ge \min_{\sigma'}\hmin(Ans_{u_1,v_1},\dots,Ans_{u_m,v_m},A\in\msfa\cap \msfa_\sigma|\sigma')$$
    where $\sigma'$ is selected from all possible partial assignments to $A=(A_1,\dots,A_n)$ that is consistent to $\sigma$ and reveals all the $A_1,\dots,A_{n/2}$. We will show that the min-entropy remains high even conditioned on $\sigma'$.

    Before bounding the min-entropy, let us further pick an analyzable part from the remaining input.
    Given $\sigma$ of length at most $q$, there are at most $0.01n$ vectors in which $\ge 0.01n$ bits are revealed from $\sigma$. We call these vectors ``bad vectors'', and other vectors ``good vectors''.

    By our construction to $Q$, for every $m$ consecutive queries, each $A_u$ for $u\in[n]$ is involved in exactly $m/(n/2)$ queries. Therefore, at least $0.98m$ queries are among good vectors. In addition, at most $0.3m$ queries $(u,v)$ has $|W_{u,v}|> 0.5\sqrt n\log n$. Denote by $m'\ge 0.68m$ the number of remaining queries, and the sequence of remaining queries $Q'$.

    Fix an index $v\in V_2$ such that it is contained in at least $m'/3n$ queries in $Q'$. There are at least $\ge 0.31n$ many such vertices, otherwise the total number of queries in $Q'$ is
    $$<0.31n\cdot \frac{2m}{n}+\left(\frac n2-0.31n\right)\cdot \frac{m'}{3n}\le m'$$
    where $2m/n$ is the maximum number of queries that involve a vertex $v\in V_2$ in $m$ consecutive queries.
    
    We denote by $d\ge m'/3n>\log^2 n$ the number of queries that contain $v$, and $(u'_1,v),\dots,(u_d',v)\in Q'$ these queries. Given all the vectors in $V_1$ fixed by $\sigma'$, these queries are independent of other queries in $Q'$.

    Assuming a good input, by the condition (2) of Definition~\ref{def:sd_good_input}, at most $\frac{2\cdot 2m/n}{\log n}=\frac{4m}{n\log n}$ out of $d$ queries have answers of $1$. Recall that condition (2) gives an upper bound to the sum of answers of a consecutive sequence of queries to every fixed $v$. Note from the construction to $Q$ that the $m$ consecutive queries contain exactly $2m/n\le \sqrt n$ consecutive queries on $v$.
    While $Q'$ is obtained by removing elements from the $m$ queries, the sum of answers of queries in $Q'$ is smaller than or equal to the consecutive sum in $Q$, and is also bounded by condition (2).
    
    Without loss of generality, we assume that they are the first $d'\le \frac{4m}{n\log n}$ queries:
    $$Ans_{u_1',v}=\dots =Ans_{u_{d'}',v}=1.$$
    For each $u\in[n]$, we use $A'_u$ to denote the vector $A_u$ removed all the indices $w$ where $A_v(w)$ is revealed by $\sigma$.
    We will show that the Hamming weight of $A_{u_{d'+1}'}'\vee\dots\vee A_{u_d'}'$ is high. That is, given $\sigma'$, for every fixed and possible answer to good inputs, a large enough proportion of bits of $A_v'$ are forced to be $0$, to avoid intersecting with any $A_{u_{d'+1}'}',\dots, A_{u_d'}'$. Therefore, the conditional probability this happens is always small.

    Notice that by condition (1) of Definition~\ref{def:sd_good_input} and Definition~\ref{def:sd_good_input_sigma},
    $$\sum_{i=d'+1}^d |A_{u_i'}'|\ge (d-d')\cdot 0.49\frac{\sqrt n}{\log n}> \frac{9.7r\log^2n}{\sqrt n}$$
    where for each index $w\in[n]$, by condition (3) of Definition~\ref{def:sd_good_input},
    $$\sum_{i=d'+1}^d|A'_{u_i'}(w)|\le 2\log n.$$
    Therefore, for the Hamming weight of the OR of these binary strings,
    $$\left|A'_{u_{d'+1}}\vee\dots\vee A'_{u_d}\right|\ge \frac{9.7r\log^2n}{\sqrt n}\cdot \frac{1}{2\log n}=\frac{4.8r\log n}{\sqrt n}.$$
    
    Since given $\sigma$, $A'_v$ is a random vector where each bit is $1$ with independent probability $1/\sqrt n\log n$, we have
    $$\Pr[A'_v\wedge (A'_{u_{d'+1}}\vee\dots\vee A'_{u_d})=0\textnormal{ and }A\in \msfa\cap\msfa_{\sigma}|\sigma']\le (1-1/\sqrt n\log n)^{\frac{4.8r\log n}{\sqrt n}}\le \exp(-4.8r/n).$$
    That means that the min-entropy of queries on $v$ is at least $$4.8\log(e)r/n\ge 6.92r/n.$$
    Since the min-entropy on different $v\in V_2$ conditioned on $\sigma'$ are independent, the total min-entropy is lower bounded by their sum:
    $$\hmin(Ans_{u_1,v_1},\dots,Ans_{u_l,v_l},A\in\msfa\cap \msfa_\sigma|\sigma)\ge 0.31n\cdot \frac{6.92r}{n}>2r$$
    Therefore, the lower bound holds for set disjointness.
    
\end{proof}

The above lower bound also applies for the $(2-\varepsilon)$-approximate APSP, for every $\varepsilon>0$. We obtain it via a reduction to set disjointness under the nemesis distribution $\mcala$.

\paragraph{$(2-\varepsilon)$-approximate all-pairs shortest paths (APSP).} Fix a parameter $\varepsilon>0$, Given as input an undirected graph $G=(V,E)$. The succinct and systematic data structure, given a sequence of all pairs of vertices $(u,v)\in \binom V2$, must answer a value $d\in [d_G(u,v),(2-\varepsilon)d_G(u,v)]$ for each $(u,v)$, where $d_G(u,v)$ denotes the shortest distance between $u$ and $v$ in $G$.

\begin{theorem}[Lower bound to $(2-\varepsilon)$-approximate APSP]
\label{thm:apsp_lb}
Fix $\varepsilon>0$. Given a (possibly randomized) succinct and systematic data structure with $r\in[5n\log^2n,n\sqrt n/4000\log^3n]$ bits of redundancy that given as input a graph $G$ on $n$ vertices $(2-\varepsilon)$-approximates all pairs shortest paths distances correctly with $\ge 99/100$ probability, its total number of probes is $\Omega(n^4/(r\log^3n))$ at the worst case.
\end{theorem}

\begin{proof}
    We prove the lower bound through a reduction from set disjointness. For every instance of set disjointness, we construct an input to APSP, such that data structures to APSP, through the reduction, can be revised into data structures for set disjointness with the same success probability and complexity.

    Suppose for the sake of contradiction that there exists a data structure $D$ of bounded complexity for $(2-\varepsilon)$-APSP.
    Given an input from $\mcala$. We construct a graph $G=(V,E)$ of $2n$ vertices, where $V=V_1\cup W\cup V_2$ and $2|V_1|=|W|=2|V_2|=n$. For every $u\in V_1$ and $w\in[n]$, $u$ connects to the $w$-th vertex in $W$ if and only if $A_{u}(w)=1$. Analogously, for every $v\in V_2$ and $w\in[n]$, $v$ connects to the $w$-th vertex in $W$ if and only if $A_v(w)=1$.

    For each query $(u,v)$ from $Q$, we ask $D$ about the shortest path length between $u\in V_1$ and $v\in V_2$ in $G$. If the answer $d\in[2,4)$, we answer that $|A_u\cap A_v|\ge 1$; otherwise, we answer $|A_u\cap A_v|=0$.

    To see the correctness, by our construction, $d_G(u,v)=2$ if and only if $|A_u\cap A_v|\ge 1$; otherwise $d_G(u,v)\ge 4$. The constructed graph is a bipartite graph on two parts $V_1\cup V_2$ and $W$, and there cannot be paths of length $3$ between $u$ and $v$. We get a data structure for set disjointness with the same redundancy and the same probe complexity. Therefore, the same lower bound holds for $(2-\varepsilon)$-approximate APSP.
\end{proof}

\section{Acknowledgement}

The authors would like to thank Tianren Liu for helpful discussions. And we also thank anonymous reviewers for their valuable suggestions.

\newpage

\bibliographystyle{alpha}
\bibliography{ref}

@article{cell_yao,
  title = {Should tables be sorted?},
  author = {Yao, Andrew Chi-Chih},
  journal = {Journal of the ACM (JACM)},
  volume = {28},
  number = {3},
  pages = {615--628},
  year = {1981},
  publisher = {ACM New York, NY, USA},
  timestamp = {Tue, 06 Nov 2018 00:00:00 +0100},
  biburl = {https://dblp.org/rec/journals/jacm/Yao81a.bib},
  bibsource = {dblp computer science bibliography, https://dblp.org},
  doi = {10.1145/322261.322274},
  _bib2doi_selected = {dblp:/rec/journals/jacm/Yao81a.bib},
  _bib2doi_confirmed = {true},
}

@book{jacobson1988succinct,
  title = {Succinct static data structures},
  author = {Jacobson, Guy Joseph},
  year = {1988},
  publisher = {Carnegie Mellon University},
  _bib2doi_finished = {true},
}

@inproceedings{miltersen1995data,
  title = {On data structures and asymmetric communication complexity},
  author = {Miltersen, Peter Bro and Nisan, Noam and Safra, Shmuel and Wigderson, Avi},
  booktitle = {Proceedings of the twenty-seventh annual ACM symposium on Theory of computing},
  pages = {103--111},
  year = {1995},
  timestamp = {Sat, 30 Sep 2023 01:00:00 +0200},
  biburl = {https://dblp.org/rec/conf/stoc/MiltersenNSW95.bib},
  bibsource = {dblp computer science bibliography, https://dblp.org},
  doi = {10.1145/225058.225093},
  _bib2doi_selected = {dblp:/rec/conf/stoc/MiltersenNSW95.bib},
  _bib2doi_confirmed = {true},
}

@article{thorup2005approximate,
  title = {Approximate distance oracles},
  author = {Thorup, Mikkel and Zwick, Uri},
  journal = {Journal of the ACM (JACM)},
  volume = {52},
  number = {1},
  pages = {1--24},
  year = {2005},
  publisher = {ACM New York, NY, USA},
  timestamp = {Tue, 06 Nov 2018 00:00:00 +0100},
  biburl = {https://dblp.org/rec/journals/jacm/ThorupZ05.bib},
  bibsource = {dblp computer science bibliography, https://dblp.org},
  doi = {10.1145/1044731.1044732},
  _bib2doi_selected = {dblp:/rec/journals/jacm/ThorupZ05.bib},
  _bib2doi_confirmed = {true},
}

@inproceedings{puatracscu2010cell,
  title = {Cell-Probe Lower Bounds for Succinct Partial Sums},
  author = {Mihai Patrascu and Emanuele Viola},
  booktitle = {Proceedings of the Twenty-First Annual {ACM-SIAM} Symposium on Discrete Algorithms, {SODA} 2010, Austin, Texas, USA, January 17-19, 2010},
  pages = {117--122},
  year = {2010},
  organization = {SIAM},
  timestamp = {Tue, 02 Feb 2021 00:00:00 +0100},
  biburl = {https://dblp.org/rec/conf/soda/PatrascuV10.bib},
  bibsource = {dblp computer science bibliography, https://dblp.org},
  doi = {10.1137/1.9781611973075.11},
  publisher = {{SIAM}},
  url = {https://doi.org/10.1137/1.9781611973075.11},
  editor = {Moses Charikar},
  _bib2doi_selected = {dblp:/rec/conf/soda/PatrascuV10.bib},
  _bib2doi_confirmed = {true},
  _bib2doi_finished = {true},
}

@article{nisan1998products,
  title = {Products and help bits in decision trees},
  author = {Nisan, Noam and Rudich, Steven and Saks, Michael},
  journal = {SIAM Journal on Computing},
  volume = {28},
  number = {3},
  pages = {1035--1050},
  year = {1998},
  publisher = {SIAM},
  timestamp = {Sat, 27 May 2017 01:00:00 +0200},
  biburl = {https://dblp.org/rec/journals/siamcomp/NisanRS99.bib},
  bibsource = {dblp computer science bibliography, https://dblp.org},
  doi = {10.1137/S0097539795282444},
  _bib2doi_selected = {dblp:/rec/journals/siamcomp/NisanRS99.bib},
  _bib2doi_confirmed = {true},
}

@inproceedings{beigel1998one,
  title = {One Help Bit Doesn't Help},
  author = {Richard Beigel and Tirza Hirst},
  booktitle = {Proceedings of the Thirtieth Annual {ACM} Symposium on the Theory of Computing, Dallas, Texas, USA, May 23-26, 1998},
  pages = {124--130},
  year = {1998},
  timestamp = {Tue, 06 Nov 2018 00:00:00 +0100},
  biburl = {https://dblp.org/rec/conf/stoc/BeigelH98.bib},
  bibsource = {dblp computer science bibliography, https://dblp.org},
  doi = {10.1145/276698.276720},
  publisher = {{ACM}},
  url = {https://doi.org/10.1145/276698.276720},
  editor = {Jeffrey Scott Vitter},
  _bib2doi_selected = {dblp:/rec/conf/stoc/BeigelH98.bib},
  _bib2doi_confirmed = {true},
  _bib2doi_finished = {true},
}

@article{gal2007cell,
  title = {The cell probe complexity of succinct data structures},
  author = {G{\'a}l, Anna and Miltersen, Peter Bro},
  journal = {Theoretical computer science},
  volume = {379},
  number = {3},
  pages = {405--417},
  year = {2007},
  timestamp = {Wed, 17 Feb 2021 00:00:00 +0100},
  biburl = {https://dblp.org/rec/journals/tcs/GalM07.bib},
  bibsource = {dblp computer science bibliography, https://dblp.org},
  doi = {10.1016/j.tcs.2007.02.047},
  _bib2doi_selected = {dblp:/rec/journals/tcs/GalM07.bib},
  _bib2doi_confirmed = {true},
}

@inproceedings{chakraborty2018tight,
  title = {Tight cell probe bounds for succinct boolean matrix-vector multiplication},
  author = {Chakraborty, Diptarka and Kamma, Lior and Larsen, Kasper Green},
  booktitle = {Proceedings of the 50th Annual ACM SIGACT Symposium on Theory of Computing},
  pages = {1297--1306},
  year = {2018},
  timestamp = {Tue, 14 Oct 2025 01:00:00 +0200},
  biburl = {https://dblp.org/rec/conf/stoc/ChakrabortyKL18.bib},
  bibsource = {dblp computer science bibliography, https://dblp.org},
  doi = {10.1145/3188745.3188830},
  _bib2doi_selected = {dblp:/rec/conf/stoc/ChakrabortyKL18.bib},
  _bib2doi_confirmed = {true},
}

@inproceedings{bringmann2013succinct,
  title = {Succinct sampling from discrete distributions},
  author = {Bringmann, Karl and Larsen, Kasper Green},
  booktitle = {Proceedings of the forty-fifth annual ACM symposium on Theory of Computing},
  pages = {775--782},
  year = {2013},
  timestamp = {Tue, 14 Oct 2025 01:00:00 +0200},
  biburl = {https://dblp.org/rec/conf/stoc/BringmannL13.bib},
  bibsource = {dblp computer science bibliography, https://dblp.org},
  doi = {10.1145/2488608.2488707},
  _bib2doi_selected = {dblp:/rec/conf/stoc/BringmannL13.bib},
  _bib2doi_confirmed = {true},
}

@inproceedings{larsen2012higher,
  title = {Higher cell probe lower bounds for evaluating polynomials},
  author = {Larsen, Kasper Green},
  booktitle = {2012 IEEE 53rd Annual Symposium on Foundations of Computer Science},
  pages = {293--301},
  year = {2012},
  organization = {IEEE},
  timestamp = {Tue, 14 Oct 2025 01:00:00 +0200},
  biburl = {https://dblp.org/rec/conf/focs/Larsen12.bib},
  bibsource = {dblp computer science bibliography, https://dblp.org},
  doi = {10.1109/FOCS.2012.21},
  _bib2doi_selected = {dblp:/rec/conf/focs/Larsen12.bib},
  _bib2doi_confirmed = {true},
}

@inproceedings{golynski2008redundancy,
  title = {On the redundancy of succinct data structures},
  author = {Golynski, Alexander and Raman, Rajeev and Rao, S Srinivasa},
  booktitle = {Scandinavian Workshop on Algorithm Theory},
  pages = {148--159},
  year = {2008},
  timestamp = {Thu, 15 Jun 2017 01:00:00 +0200},
  biburl = {https://dblp.org/rec/conf/swat/GolynskiRR08.bib},
  bibsource = {dblp computer science bibliography, https://dblp.org},
  doi = {10.1007/978-3-540-69903-3_15},
  _bib2doi_selected = {dblp:/rec/conf/swat/GolynskiRR08.bib},
  _bib2doi_confirmed = {true},
}

@article{golynski2007optimal,
  title = {Optimal lower bounds for rank and select indexes},
  author = {Golynski, Alexander},
  journal = {Theoretical Computer Science},
  volume = {387},
  number = {3},
  pages = {348--359},
  year = {2007},
  timestamp = {Wed, 17 Feb 2021 00:00:00 +0100},
  biburl = {https://dblp.org/rec/journals/tcs/Golynski07.bib},
  bibsource = {dblp computer science bibliography, https://dblp.org},
  doi = {10.1016/j.tcs.2007.07.041},
  _bib2doi_selected = {dblp:/rec/journals/tcs/Golynski07.bib},
  _bib2doi_confirmed = {true},
}

@inproceedings{golynski2007size,
  title = {On the Size of Succinct Indices},
  author = {Alexander Golynski and Roberto Grossi and Ankur Gupta and Rajeev Raman and S. Srinivasa Rao},
  booktitle = {Algorithms - {ESA} 2007, 15th Annual European Symposium, Eilat, Israel, October 8-10, 2007, Proceedings},
  pages = {371--382},
  year = {2007},
  timestamp = {Thu, 29 Dec 2022 00:00:00 +0100},
  biburl = {https://dblp.org/rec/conf/esa/GolynskiGGRR07.bib},
  bibsource = {dblp computer science bibliography, https://dblp.org},
  doi = {10.1007/978-3-540-75520-3_34},
  publisher = {Springer},
  volume = {4698},
  url = {https://doi.org/10.1007/978-3-540-75520-3\_34},
  editor = {Lars Arge and Michael Hoffmann and Emo Welzl},
  series = {Lecture Notes in Computer Science},
  _bib2doi_selected = {dblp:/rec/conf/esa/GolynskiGGRR07.bib},
  _bib2doi_confirmed = {true},
  _bib2doi_finished = {true},
}

@inproceedings{miltersen2005lower,
  title = {Lower bounds on the size of selection and rank indexes},
  author = {Miltersen, Peter Bro},
  booktitle = {SODA},
  volume = {5},
  pages = {11--12},
  year = {2005},
  timestamp = {Fri, 07 Dec 2012 00:00:00 +0100},
  biburl = {https://dblp.org/rec/conf/soda/Miltersen05.bib},
  bibsource = {dblp computer science bibliography, https://dblp.org},
  url = {http://dl.acm.org/citation.cfm?id=1070432.1070435},
  _bib2doi_selected = {dblp:/rec/conf/soda/Miltersen05.bib},
  _bib2doi_confirmed = {true},
}

@book{belsley1993rates,
  title = {Rates of convergence of Markov chains related to association schemes},
  author = {Belsley, Eric David},
  year = {1993},
  publisher = {Harvard University},
  _bib2doi_finished = {true},
}

@article{fulman2015stein,
  title = {Stein’s method and the rank distribution of random matrices over finite fields},
  author = {Fulman, Jason and Goldstein, Larry},
  journal = {The Annals of Probability},
  volume = {43},
  number = {3},
  pages = {1274--1314},
  year = {2015},
  _bib2doi_finished = {true},
}

@article{motwani1996randomized,
  title = {Randomized algorithms},
  author = {Motwani, Rajeev and Raghavan, Prabhakar},
  journal = {ACM Computing Surveys (CSUR)},
  volume = {28},
  number = {1},
  pages = {33--37},
  year = {1996},
  publisher = {ACM New York, NY, USA},
  timestamp = {Thu, 02 Jan 2020 00:00:00 +0100},
  biburl = {https://dblp.org/rec/journals/csur/MotwaniR96.bib},
  bibsource = {dblp computer science bibliography, https://dblp.org},
  doi = {10.1145/234313.234327},
  _bib2doi_selected = {dblp:/rec/journals/csur/MotwaniR96.bib},
  _bib2doi_confirmed = {true},
}

@book{horn2012matrix,
  title = {Matrix Analysis, 2nd Ed},
  author = {Roger A. Horn and Charles R. Johnson},
  year = {2012},
  publisher = {Cambridge University Press},
  timestamp = {Mon, 29 Jul 2019 01:00:00 +0200},
  biburl = {https://dblp.org/rec/books/cu/HJ2012.bib},
  bibsource = {dblp computer science bibliography, https://dblp.org},
  doi = {10.1017/CBO9781139020411},
  url = {https://doi.org/10.1017/CBO9781139020411},
  isbn = {9780521548236},
  _bib2doi_selected = {dblp:/rec/books/cu/HJ2012.bib},
  _bib2doi_confirmed = {true},
  _bib2doi_finished = {true},
}

@article{wei2017upper,
  title = {Upper bound for intermediate singular values of random matrices},
  author = {Wei, Feng},
  journal = {Journal of Mathematical Analysis and Applications},
  volume = {445},
  number = {2},
  pages = {1530--1547},
  year = {2017},
  _bib2doi_finished = {true},
}

@article{szarek1990spaces,
  title = {Spaces with large distance to $\ell^n_\infty$ and random matrices},
  author = {Szarek, Stanislaw J},
  journal = {American Journal of Mathematics},
  volume = {112},
  number = {6},
  pages = {899--942},
  year = {1990},
  publisher = {JSTOR},
  _bib2doi_finished = {true},
}

@article{thompson1976behavior,
  title = {The behavior of eigenvalues and singular values under perturbations of restricted rank},
  author = {Thompson, Robert C},
  journal = {Linear Algebra and its Applications},
  volume = {13},
  number = {1-2},
  pages = {69--78},
  year = {1976},
  _bib2doi_finished = {true},
}

@inproceedings{natarajan2020equivalence,
  title = {Equivalence of Systematic Linear Data Structures and Matrix Rigidity},
  author = {Sivaramakrishnan Natarajan Ramamoorthy and Cyrus Rashtchian},
  booktitle = {11th Innovations in Theoretical Computer Science Conference, {ITCS} 2020, Seattle, Washington, USA, January 12-14, 2020},
  pages = {35:1--35:20},
  year = {2020},
  organization = {Schloss Dagstuhl--Leibniz-Zentrum f{\"u}r Informatik},
  timestamp = {Mon, 06 Jan 2020 00:00:00 +0100},
  biburl = {https://dblp.org/rec/conf/innovations/RamamoorthyR20.bib},
  bibsource = {dblp computer science bibliography, https://dblp.org},
  doi = {10.4230/LIPIcs.ITCS.2020.35},
  publisher = {Schloss Dagstuhl - Leibniz-Zentrum f{\"{u}}r Informatik},
  volume = {151},
  url = {https://doi.org/10.4230/LIPIcs.ITCS.2020.35},
  editor = {Thomas Vidick},
  series = {LIPIcs},
  _bib2doi_selected = {dblp:/rec/conf/innovations/RamamoorthyR20.bib},
  _bib2doi_confirmed = {true},
  _bib2doi_finished = {true},
}

@inproceedings{dvir2019static,
  title = {Static data structure lower bounds imply rigidity},
  author = {Dvir, Zeev and Golovnev, Alexander and Weinstein, Omri},
  booktitle = {Proceedings of the 51st Annual ACM SIGACT Symposium on Theory of Computing},
  pages = {967--978},
  year = {2019},
  timestamp = {Sat, 22 Jun 2019 01:00:00 +0200},
  biburl = {https://dblp.org/rec/conf/stoc/DvirGW19.bib},
  bibsource = {dblp computer science bibliography, https://dblp.org},
  doi = {10.1145/3313276.3316348},
  _bib2doi_selected = {dblp:/rec/conf/stoc/DvirGW19.bib},
  _bib2doi_confirmed = {true},
}

@article{viola2019lower,
  title = {Lower bounds for data structures with space close to maximum imply circuit lower bounds},
  author = {Viola, Emanuele},
  journal = {Theory of Computing},
  volume = {15},
  number = {1},
  pages = {1--9},
  year = {2019},
  publisher = {Theory of Computing Exchange},
  timestamp = {Tue, 09 Feb 2021 00:00:00 +0100},
  biburl = {https://dblp.org/rec/journals/toc/Viola19.bib},
  bibsource = {dblp computer science bibliography, https://dblp.org},
  doi = {10.4086/toc.2019.v015a018},
  _bib2doi_selected = {dblp:/rec/journals/toc/Viola19.bib},
  _bib2doi_confirmed = {true},
}

@inproceedings{golovnev2020data,
  title = {Data structures meet cryptography: 3SUM with preprocessing},
  author = {Golovnev, Alexander and Guo, Siyao and Horel, Thibaut and Park, Sunoo and Vaikuntanathan, Vinod},
  booktitle = {Proceedings of the 52nd annual ACM SIGACT symposium on theory of computing},
  pages = {294--307},
  year = {2020},
  timestamp = {Thu, 09 Apr 2026 01:00:00 +0200},
  biburl = {https://dblp.org/rec/conf/stoc/GolovnevGHPV20.bib},
  bibsource = {dblp computer science bibliography, https://dblp.org},
  doi = {10.1145/3357713.3384342},
  _bib2doi_selected = {dblp:/rec/conf/stoc/GolovnevGHPV20.bib},
  _bib2doi_confirmed = {true},
}

@inproceedings{corrigan2019function,
  title = {The function-inversion problem: Barriers and opportunities},
  author = {Corrigan-Gibbs, Henry and Kogan, Dmitry},
  booktitle = {Theory of Cryptography Conference},
  pages = {393--421},
  year = {2019},
  timestamp = {Mon, 25 Nov 2019 00:00:00 +0100},
  biburl = {https://dblp.org/rec/conf/tcc/Corrigan-GibbsK19.bib},
  bibsource = {dblp computer science bibliography, https://dblp.org},
  doi = {10.1007/978-3-030-36030-6_16},
  _bib2doi_selected = {dblp:/rec/conf/tcc/Corrigan-GibbsK19.bib},
  _bib2doi_confirmed = {true},
}

@inproceedings{dvovrak2021data,
  title = {Data Structures Lower Bounds and Popular Conjectures},
  author = {Pavel Dvor{\'{a}}k and Michal Kouck{\'{y}} and Karel Kr{\'{a}}l and Veronika Sl{\'{\i}}vov{\'{a}}},
  booktitle = {29th Annual European Symposium on Algorithms, {ESA} 2021, Lisbon, Portugal (Virtual Conference), September 6-8, 2021},
  pages = {39:1--39:15},
  year = {2021},
  organization = {Schloss Dagstuhl--Leibniz-Zentrum f{\"u}r Informatik},
  timestamp = {Wed, 07 Dec 2022 00:00:00 +0100},
  biburl = {https://dblp.org/rec/conf/esa/Dvorak00S21.bib},
  bibsource = {dblp computer science bibliography, https://dblp.org},
  doi = {10.4230/LIPIcs.ESA.2021.39},
  publisher = {Schloss Dagstuhl - Leibniz-Zentrum f{\"{u}}r Informatik},
  volume = {204},
  url = {https://doi.org/10.4230/LIPIcs.ESA.2021.39},
  editor = {Petra Mutzel and Rasmus Pagh and Grzegorz Herman},
  series = {LIPIcs},
  _bib2doi_selected = {dblp:/rec/conf/esa/Dvorak00S21.bib},
  _bib2doi_confirmed = {true},
  _bib2doi_finished = {true},
}

@inproceedings{henzinger2015unifying,
  title = {Unifying and strengthening hardness for dynamic problems via the online matrix-vector multiplication conjecture},
  author = {Henzinger, Monika and Krinninger, Sebastian and Nanongkai, Danupon and Saranurak, Thatchaphol},
  booktitle = {Proceedings of the forty-seventh annual ACM symposium on Theory of computing},
  pages = {21--30},
  year = {2015},
  timestamp = {Mon, 03 Jan 2022 00:00:00 +0100},
  biburl = {https://dblp.org/rec/conf/stoc/HenzingerKNS15.bib},
  bibsource = {dblp computer science bibliography, https://dblp.org},
  doi = {10.1145/2746539.2746609},
  _bib2doi_selected = {dblp:/rec/conf/stoc/HenzingerKNS15.bib},
  _bib2doi_confirmed = {true},
}

@inproceedings{valiant1977graph,
  title = {Graph-theoretic arguments in low-level complexity},
  author = {Valiant, Leslie G},
  booktitle = {International Symposium on Mathematical Foundations of Computer Science},
  pages = {162--176},
  year = {1977},
  timestamp = {Fri, 19 May 2017 01:00:00 +0200},
  biburl = {https://dblp.org/rec/conf/mfcs/Valiant77.bib},
  bibsource = {dblp computer science bibliography, https://dblp.org},
  doi = {10.1007/3-540-08353-7_135},
  _bib2doi_selected = {dblp:/rec/conf/mfcs/Valiant77.bib},
  _bib2doi_confirmed = {true},
}

@article{valiant1992boolean,
  title = {Why is Boolean complexity theory difficult},
  author = {Valiant, Leslie G},
  journal = {Boolean Function Complexity},
  volume = {169},
  number = {84-94},
  pages = {4},
  year = {1992},
  publisher = {Cambridge University Press Cambridge},
  _bib2doi_finished = {true},
}

@article{brin1998anatomy,
  title = {The anatomy of a large-scale hypertextual web search engine},
  author = {Brin, Sergey and Page, Lawrence},
  journal = {Computer networks and ISDN systems},
  volume = {30},
  number = {1-7},
  pages = {107--117},
  year = {1998},
  timestamp = {Wed, 19 Feb 2020 00:00:00 +0100},
  biburl = {https://dblp.org/rec/journals/cn/BrinP98.bib},
  bibsource = {dblp computer science bibliography, https://dblp.org},
  doi = {10.1016/S0169-7552(98)00110-X},
  _bib2doi_selected = {dblp:/rec/journals/cn/BrinP98.bib},
  _bib2doi_confirmed = {true},
}

@book{norris1998markov,
  title = {Markov chains},
  author = {James R. Norris},
  series = {Cambridge series in statistical and probabilistic mathematics},
  year = {1998},
  publisher = {Cambridge University Press},
  timestamp = {Thu, 21 Apr 2011 01:00:00 +0200},
  biburl = {https://dblp.org/rec/books/daglib/0095301.bib},
  bibsource = {dblp computer science bibliography, https://dblp.org},
  isbn = {978-0-521-48181-6},
  _bib2doi_selected = {dblp:/rec/books/daglib/0095301.bib},
  _bib2doi_confirmed = {true},
  _bib2doi_finished = {true},
}

@book{golub2013matrix,
  title = {Matrix computations},
  author = {Golub, Gene H and Van Loan, Charles F},
  year = {2013},
  publisher = {JHU press},
  _bib2doi_finished = {true},
}

@article{cohen2010fast,
  title = {Fast set intersection and two-patterns matching},
  author = {Cohen, Hagai and Porat, Ely},
  journal = {Theoretical Computer Science},
  volume = {411},
  number = {40-42},
  pages = {3795--3800},
  year = {2010},
  timestamp = {Sun, 04 Aug 2024 01:00:00 +0200},
  biburl = {https://dblp.org/rec/journals/tcs/CohenP10.bib},
  bibsource = {dblp computer science bibliography, https://dblp.org},
  doi = {10.1016/j.tcs.2010.06.002},
  _bib2doi_selected = {dblp:/rec/journals/tcs/CohenP10.bib},
  _bib2doi_confirmed = {true},
}

@inproceedings{patrascu2010distance,
  title = {Distance oracles beyond the Thorup-Zwick bound},
  author = {Patrascu, Mihai and Roditty, Liam},
  booktitle = {2010 IEEE 51st Annual Symposium on Foundations of Computer Science},
  pages = {815--823},
  year = {2010},
  organization = {IEEE},
  timestamp = {Thu, 23 Mar 2023 00:00:00 +0100},
  biburl = {https://dblp.org/rec/conf/focs/PatrascuR10.bib},
  bibsource = {dblp computer science bibliography, https://dblp.org},
  doi = {10.1109/FOCS.2010.83},
  _bib2doi_selected = {dblp:/rec/conf/focs/PatrascuR10.bib},
  _bib2doi_confirmed = {true},
}

@inproceedings{kopelowitz2020towards,
  title = {Towards Optimal Set-Disjointness and Set-Intersection Data Structures},
  author = {Tsvi Kopelowitz and Virginia Vassilevska Williams},
  booktitle = {47th International Colloquium on Automata, Languages, and Programming, {ICALP} 2020, Saarbr{\"{u}}cken, Germany (Virtual Conference), July 8-11, 2020},
  pages = {74:1--74:16},
  year = {2020},
  organization = {Schloss Dagstuhl--Leibniz-Zentrum f{\"u}r Informatik},
  timestamp = {Thu, 16 Sep 2021 01:00:00 +0200},
  biburl = {https://dblp.org/rec/conf/icalp/KopelowitzW20.bib},
  bibsource = {dblp computer science bibliography, https://dblp.org},
  doi = {10.4230/LIPIcs.ICALP.2020.74},
  publisher = {Schloss Dagstuhl - Leibniz-Zentrum f{\"{u}}r Informatik},
  volume = {168},
  url = {https://doi.org/10.4230/LIPIcs.ICALP.2020.74},
  editor = {Artur Czumaj and Anuj Dawar and Emanuela Merelli},
  series = {LIPIcs},
  _bib2doi_selected = {dblp:/rec/conf/icalp/KopelowitzW20.bib},
  _bib2doi_confirmed = {true},
  _bib2doi_finished = {true},
}

@article{tarjan1979class,
  title = {A class of algorithms which require nonlinear time to maintain disjoint sets},
  author = {Tarjan, Robert Endre},
  journal = {Journal of computer and system sciences},
  volume = {18},
  number = {2},
  pages = {110--127},
  year = {1979},
  timestamp = {Tue, 16 Feb 2021 00:00:00 +0100},
  biburl = {https://dblp.org/rec/journals/jcss/Tarjan79.bib},
  bibsource = {dblp computer science bibliography, https://dblp.org},
  doi = {10.1016/0022-0000(79)90042-4},
  _bib2doi_selected = {dblp:/rec/journals/jcss/Tarjan79.bib},
  _bib2doi_confirmed = {true},
}

@inproceedings{afshani2016data,
  title = {Data Structure Lower Bounds for Document Indexing Problems},
  author = {Peyman Afshani and Jesper Sindahl Nielsen},
  booktitle = {43rd International Colloquium on Automata, Languages, and Programming, {ICALP} 2016, Rome, Italy, July 11-15, 2016},
  pages = {93:1--93:15},
  year = {2016},
  organization = {Schloss Dagstuhl--Leibniz-Zentrum f{\"u}r Informatik},
  timestamp = {Thu, 23 Aug 2018 01:00:00 +0200},
  biburl = {https://dblp.org/rec/conf/icalp/AfshaniN16.bib},
  bibsource = {dblp computer science bibliography, https://dblp.org},
  doi = {10.4230/LIPIcs.ICALP.2016.93},
  publisher = {Schloss Dagstuhl - Leibniz-Zentrum f{\"{u}}r Informatik},
  volume = {55},
  url = {https://doi.org/10.4230/LIPIcs.ICALP.2016.93},
  editor = {Ioannis Chatzigiannakis and Michael Mitzenmacher and Yuval Rabani and Davide Sangiorgi},
  series = {LIPIcs},
  _bib2doi_selected = {dblp:/rec/conf/icalp/AfshaniN16.bib},
  _bib2doi_confirmed = {true},
  _bib2doi_finished = {true},
}

@inproceedings{goldstein2017conditional,
  title = {Conditional lower bounds for space/time tradeoffs},
  author = {Goldstein, Isaac and Kopelowitz, Tsvi and Lewenstein, Moshe and Porat, Ely},
  booktitle = {Workshop on Algorithms and Data Structures},
  pages = {421--436},
  year = {2017},
  timestamp = {Fri, 21 Jul 2017 01:00:00 +0200},
  biburl = {https://dblp.org/rec/conf/wads/GoldsteinKLP17.bib},
  bibsource = {dblp computer science bibliography, https://dblp.org},
  doi = {10.1007/978-3-319-62127-2_36},
  _bib2doi_selected = {dblp:/rec/conf/wads/GoldsteinKLP17.bib},
  _bib2doi_confirmed = {true},
}

@inproceedings{goldstein2019hardness,
  title = {On the Hardness of Set Disjointness and Set Intersection with Bounded Universe},
  author = {Isaac Goldstein and Moshe Lewenstein and Ely Porat},
  booktitle = {30th International Symposium on Algorithms and Computation, {ISAAC} 2019, Shanghai University of Finance and Economics, Shanghai, China, December 8-11, 2019},
  pages = {7:1--7:22},
  year = {2019},
  organization = {Schloss Dagstuhl--Leibniz-Zentrum f{\"u}r Informatik},
  timestamp = {Thu, 28 Nov 2019 00:00:00 +0100},
  biburl = {https://dblp.org/rec/conf/isaac/GoldsteinLP19.bib},
  bibsource = {dblp computer science bibliography, https://dblp.org},
  doi = {10.4230/LIPIcs.ISAAC.2019.7},
  publisher = {Schloss Dagstuhl - Leibniz-Zentrum f{\"{u}}r Informatik},
  volume = {149},
  url = {https://doi.org/10.4230/LIPIcs.ISAAC.2019.7},
  editor = {Pinyan Lu and Guochuan Zhang},
  series = {LIPIcs},
  _bib2doi_selected = {dblp:/rec/conf/isaac/GoldsteinLP19.bib},
  _bib2doi_confirmed = {true},
  _bib2doi_finished = {true},
}

@book{Vershynin_2018,
  place = {Cambridge},
  series = {Cambridge Series in Statistical and Probabilistic Mathematics},
  title = {High-Dimensional Probability: An Introduction with Applications in Data Science},
  publisher = {Cambridge University Press},
  author = {Vershynin, Roman},
  year = {2018},
  collection = {Cambridge Series in Statistical and Probabilistic Mathematics},
  _bib2doi_finished = {true},
}

@incollection{bogachev2007measures,
  title = {Measures on topological spaces},
  author = {Bogachev, Vladimir I},
  booktitle = {Measure theory},
  pages = {476--583},
  year = {2007},
  publisher = {Springer},
  _bib2doi_finished = {true},
}

\newpage

\appendix

\section{Missing proof to Lemma~\ref{lem:round_negl}}
\label{app:lipschitz}

The probability we are going to bound is the probability that any of the dimension of the $\mbold^k[u,v]$ is not close to $a/2n^{\alpha'}$, for integers $a$. It is non-trivial to bound since $\mbold^k$ does not follow a uniform distribution.

Specifically, we will show that each dimension of $\mbold^k$ is close to such a value with probability $\le 1/\textnormal{poly}(n)$. The desired probability bound follows by a union bound. For each dimension $[u,v]$ of $\mbold^k$, we are going to show that, even when all the other elements of $\mboldu$ are fixed by an arbitrarily partial assignment $\sigma$, the function $\phi_{k,\sigma}:\mbbr\rightarrow \mbbr$ that maps from $\mboldu[u,v]$ to $\mbold^k[u,v]$ evaluates to a value that is close to $a/2n^{\alpha'}$ with low probability. This is because $\phi_{k,\sigma}$ is an increasing function with its derivative well-controlled everywhere.

First, we need the following partial derivative bounds that generalizes Lemma~\ref{lem:partmboldu} and Lemma~\ref{lem:partmkm}, and works for every pair of indices $[u,v], [i,j]$. Recall that we assume that each $\mboldu^k[u,v]$ is a function on $n^2$ variables $\mboldu[i,j]$. For simplicity, we study for now their partial derivatives without the constraint of being a symmetric matrix. Given the symmetric constraint, their partial derivative is exactly $\frac{\partial \mboldu^k[u,v]}{\partial \mboldu[i,j]}+\frac{\partial \mboldu^k[u,v]}{\partial \mboldu[j,i]}$ when $i\ne j$, and is $\frac{\partial \mboldu^k[u,v]}{\partial \mboldu[i,j]}$ otherwise.

\begin{lemma}
    \label{lem:mboldu_bound_precise}
    Let $\gamma> 2$, and the values of elements of $\mboldu$ range from $[\frac 1n-\frac{1}{n^\gamma},\frac 1n]$. For every integer $2\le k\le n$, for every element $[i,j]$ of $\mboldu$ and every element $[u,v]$ of $\mboldu^k$, we have the following bounds of $\frac{\partial \mboldu^k[u,v]}{\partial \mboldu[i,j]}$:
        $$\frac{\partial \mboldu^k[u,v]}{\partial \mboldu[i,j]}= 
        \begin{cases}
            \Theta(\frac{n}{n^2}), &u=i\textnormal{ or }v=j\\
            \Theta(\frac{k-2}{n^2}), &\textnormal{otherwise}
        \end{cases}$$
    If $k=1$, we have:
        $$\frac{\partial \mboldu^k[u,v]}{\partial \mboldu[i,j]}=\left\{
        \begin{aligned}
            &1, &u=i\textnormal{ and }v=j&\\
            &0, &\textnormal{otherwise}&
        \end{aligned}
        \right.$$
\end{lemma}

\begin{proof}
    We generalize Lemma~\ref{lem:partmboldu} to every pair of vertices $[u,v]$ and $[i,j]$. Recall that
    \begin{align*}
    \frac{\partial \mboldu^k[u,v]}{\partial\mboldu[i,j]}=\sum_{g=0}^{k-1}\left(\sum_{u=h_0,h_1,...,h_g=i,h_{g+1}=j,...,h_{k-1},h_k=v}\left(\prod_{t=0,t\ne g}^{k-1}\mboldu[h_t,h_{t+1}]\right)\right)
    \end{align*}
    
Again, $\prod_{t=0,t\ne g}^{k-1}\mboldu[h_t,h_{t+1}]$ can be lower-bounded by setting all elements in $\mboldu$ equal to $\frac 1n-\frac{1}{n^{\gamma}}$ and upper bounded by setting all elements in $\mboldu$ equal to $\frac1n$. We only need to count the number of sequences $u=h_0,h_1,...,h_{k-1},h_k=v$ in which two adjacent indices are $i$ and $j$, while in both the lower and upper bounds all elements are equivalent. When $g$ is picked from $\{1,2,...,k-2\}$, there are $(k-2)n^{k-3}$ distinct sequences $h_0,\dots,h_k$ for each summation. $g$ can be $0$ (resp., $k-1$) only if $u$ (resp., $v$) coincides with $i$ (resp., $j$). More precisely, the number of extra and distinct sequences for both cases $u=i$ and $j=v$ are exactly $n^{k-2}$.

Therefore, when $u\ne i$ and $v\ne j$,
$$\frac{\partial \mboldu^k[u,v]}{\partial\mboldu[i,j]}\le (k-2)n^{k-3}\frac{1}{n^{k-1}}=\frac{k-2}{n^2}$$
and
$$\frac{\partial \mboldu^k[u,v]}{\partial\mboldu[i,j]}\ge (k-2)n^{k-3}\left(\frac1n-\frac1{n^\gamma}\right)^{k-1}\ge e^{-\frac{k-1}{n^{\gamma-1}-1}}\cdot \frac{k-2}{n^2}\ge \left(1-o(1)\right)\cdot \frac{k-2}{n^2}.$$
by Fact~\ref{fact:e_inequality} and $\gamma>2$.

Similarly, for $u= i$ or $v=j$ but not both,

$$(1-o(1))\frac{n+k-2}{n^2}\le \frac{\partial \mboldu^k[u,v]}{\partial\mboldu[i,j]}\le \frac{n+k-2}{n^2}$$

When $u=i$ and $v=j$,

$$(1-o(1))\frac{2n+k-2}{n^2}\le \frac{\partial \mboldu^k[u,v]}{\partial\mboldu[i,j]}\le \frac{2n+k-2}{n^2}$$

\end{proof}

The following lemma generalizes Lemma~\ref{lem:partmkm}.

\begin{lemma}
    \label{lem:mkm_bound_precise}
    Let $\gamma> 2$ and $\beta\ge \gamma+2$, and the values of elements of $\mboldu$ range from $[\frac 1n-\frac{1}{n^\gamma},\frac 1n]$. For every integer $3\le k\le n$, for every element $[i,j]$ of $\mbold$ and every element $[u,v]$ of $\mbold^k$, we have the following bound of $\frac{\partial \mbold^k[u,v]}{\partial \mbold[i,j]}$:
    $$\frac{\partial \mbold^k[u,v]}{\partial \mbold[i,j]}=\begin{cases}
        \Theta(k), &u=i\textnormal{ and }v=j\\
        \Theta(\frac{k^2}{n^{1+\beta}}), &u=i\textnormal{ or }v=j,\textnormal{ but not both}\\
        \Theta(\frac{k^3}{n^{1+2\beta}}), &\textnormal{otherwise}
    \end{cases}.$$
\end{lemma}

\begin{proof}
    We generalize Lemma~\ref{lem:partmkm} to every pair of indices $[u,v]$ and $[i,j]$. Recall that
    \begin{align*}
        \frac{\partial\mbold^k[u,v]}{\partial\mbold[i,j]}=\sum_{t=0}^k\binom kt\left(\frac{n^2-1}{n^2}\right)^{k-t}\left(\frac{1}{n^\beta}\right)^{t-1}\frac{\partial\mboldu^t[u,v]}{\partial\mboldu[i,j]}
    \end{align*}
    which is dominated by $t=1$ when $u=i$, $v=j$; dominated by $t=2$ when $u=i$ or $v=j$, but not both; and dominated by the term of $t=3$ otherwise.
    
    Combined with Lemma~\ref{lem:mboldu_bound_precise}, we have the following bounds. When $u=i,v=j$,
    \begin{align*}
        \Omega(k)\le \frac{\partial\mbold^k[u,v]}{\partial\mbold[i,j]}\le k
    \end{align*}

    When $u=i$ or $v=j$, but not both,
    $$\Omega\left(\frac{k^2}{n^{1+\beta}}\right)\le \frac{\partial\mbold^k[u,v]}{\partial\mbold[i,j]}\le \frac{k^2(n+k-2)}{n^{2+\beta}}$$
    
    When $u\ne i$ and $v\ne j$,
    $$\Omega\left(\frac{k^3}{n^{1+2\beta}}\right)\le \frac{\partial\mbold^k[u,v]}{\partial\mbold[i,j]}\le \frac{k^3(2n+k-2)}{n^{2+2\beta}}$$
\end{proof}

Now, we are ready to prove Lemma~\ref{lem:round_negl}.

\begin{proof}[Proof to Lemma~\ref{lem:round_negl}]
    Recall that we are going to show that the following happens with high probability: for every $a\in\mathbb N$ and every $u,v\in[n]$, $|\mbold^k[u,v]-a/2n^{\alpha'}|>1/n^\alpha$. We bound its fail probability for every fixed $u,v$. The desired bound follows by a union bound.

    Let $\sigma$ denote a partial assignment to all the other elements of $\mboldu$ except $\mboldu[u,v],\mboldu[v,u]$. And we let $\phi_{k,\sigma}:\mbbr\rightarrow\mbbr$ a function that maps $\mboldu[u,v]=\mboldu[v,u]$ to $\mbold^k[u,v]$. We have
    $$\Pr_{\mbold\leftarrow\mcalm}[\exists a\in\mathbb N,|\mbold^k[u,v]-a/2n^{\alpha'}|\le 1/n^\alpha]\le \max_{\sigma}\Pr_{x\in [\frac1n-\frac1{n^\gamma},\frac1n]}[\exists a\in\mathbb N,|\phi_{k,\sigma}(x)-a/2n^{\alpha'}|\le 1/n^\alpha].$$

    There are two cases. When $u=v$,
    $$\frac{\partial \phi_{k,\sigma}}{\partial x}=\frac{\partial \mbold^k[u,v]}{\partial \mboldu[u,v]}=\frac{1}{n^\beta}\cdot \frac{\partial \mbold^k[u,v]}{\partial \mbold[u,v]}=\Theta(k/n^\beta).$$
    When $u\ne v$,
    $$\frac{\partial \phi_{k,\sigma}}{\partial x}=\frac{\partial \mbold^k[u,v]}{\partial \mboldu[u,v]}+\frac{\partial \mbold^k[u,v]}{\partial \mboldu[v,u]}=\frac{1}{n^\beta}\cdot \left(\frac{\partial \mbold^k[u,v]}{\partial \mbold[u,v]}+\frac{\partial \mbold^k[u,v]}{\partial \mbold[v,u]}\right)=\Theta(k/n^\beta).$$

    Let us call $D(a/2n^{\alpha'}, 1/n^{\alpha})$ a disc, which is the set of real numbers that is within $1/n^{\alpha}$ distance from $a/2n^{\alpha'}$.
    We also need to lower bound the range length of $\mbold^k[u,v]$, which gives the lower bound of how many discs centered at $a/2n^{\alpha'}$ that $f$ will meet. Since the range length of $\mboldu[u,v]$ is $1/n^{\gamma}$, the range length of $\mbold^k[u,v]$ is $\Theta(k/n^{\gamma+\beta})$. Therefore, there will be $\Theta(kn^{\alpha'-\gamma-\beta})$ discs.

    For each disc, the range length of $\mboldu[u,v]$ such that $\phi_{k,\sigma}$ evaluates to this disc is at most $\Theta(n^\beta/kn^\alpha)$. Therefore, the total volume of $\mboldu[u,v]$ such that $\phi_{k,\sigma}$ falls in at least one discs is at most $\Theta(n^{\alpha'-\gamma-\alpha})$. Therefore, by union bound,
    \begin{align*}
        &\Pr_{\mbold\leftarrow\mcalm}[\exists u,v\in[n], \exists a\in\mathbb N,|\mbold^k[u,v]-a/2n^{\alpha'}|\le 1/n^\alpha]\\
        \le&n^2\max_{\sigma}\Pr_{x\in [\frac1n-\frac1{n^\gamma},\frac1n]}[\exists a\in\mathbb N,|\phi_{k,\sigma}(x)-a/2n^{\alpha'}|\le 1/n^\alpha]\\
        \le&n^2\cdot \Theta(n^{\alpha'-\gamma-\alpha})/(1/n^\gamma)\\
        =&\Theta(n^{\alpha'-\alpha+2})
    \end{align*}
    which is $o(1)$ when $\alpha>2+\alpha'$.
\end{proof}

\section{Missing proof to Lemma~\ref{lem:singular_value_concentration}}
\label{app:singular_value}

We prove Lemma~\ref{lem:singular_value_concentration} in this section, which is a direct application of the concentration bound on intermediate singular values given in \cite{szarek1990spaces, wei2017upper}. Below are definitions used in presenting the concentration bound.

\begin{definition}
    Let $Z$ be a random variable. Then the $\psi_2$-norm (subgaussian norm) of $Z$ is defined as
    $$\|Z\|_{\psi_2}:=\inf\left\{\lambda>0:\mbbe\left[\exp\left(\frac{|Z|^2}{\lambda^2}\right)\right]\le 2\right\}$$
\end{definition}

\begin{lemma}[Corollary 1.11. of \cite{wei2017upper}]
    \label{lem:wei_concentration}
    Let $\abold$ an $m\times m$ random matrix with i.i.d. entries that have mean $0$, variance $1$ and $\psi_2$-norm $K$. Assume also that there exist a constant $c(K)>0$ that depends only on $K$, and positive numbers $p>0$, $s\le c(K)\min(1/p,1)$ such that for every $i,j\in[m]$,
    \begin{equation}
    \label{eq:wei_singular_bound}
    \sup_{u\in\mbbc}\Pr[|\abold[i,j]-u|\le s]\le ps.
    \end{equation}
    Then there exist $0<C_1<C_2$ and $C_3>0$ that depend only on $K$ and $p$ such that for all $\ell\in[m]$,
    $$\Pr\left[\frac{C_1\ell}{\sqrt m}\le \sigma_{m+1-\ell}(\abold)\le \frac{C_2\ell}{\sqrt m}\right]\ge 1-\exp(-C_3\ell).$$
\end{lemma}

Lemma~\ref{lem:singular_value_concentration} follows by showing that both $p$ and $K$ are constants in our distribution. We also use the following inequality.

\begin{proposition}[Young's convolution inequality \cite{bogachev2007measures}]
    Suppose $f$ is in the Lebesgue space $L^p(\mbbr)$ and $g$ is in $L^q(\mbbr)$ and $1/p+1/q=1/r+1$ with $1\le p,q,r\le \infty$. Then
    $$\|f*g\|_r\le \|f\|_p\|g\|_q$$
    where $(f*g)(t):=\int_{-\infty}^{\infty}f(\tau)g(t-\tau)d\tau$ is the convolution of $f$ and $g$.
\end{proposition}

\begin{proof}[Proof to Lemma~\ref{lem:singular_value_concentration}]
    Recall that $\mu$ is either $\mu_1$ or $\mu_2$, the distributions of $X_1+Y_1+Z_1+W_1$ or $X_2Y_2+Z_2W_2$, where $X_i,Y_i,Z_i,W_i$ are i.i.d. uniform on some symmetric intervals chosen so that the resulting distribution has variance $1$. Concretely,
    $$
    X_1,Y_1,Z_1,W_1\leftarrow\textnormal{Unif}\left[-\frac{\sqrt3}{2},\frac{\sqrt3}{2}\right]
    \qquad
    X_2,Y_2,Z_2,W_2\leftarrow\textnormal{Unif}\left[-\left(\frac{9}{2}\right)^{1/4},\left(\frac{9}{2}\right)^{1/4}\right]
    $$

    We will apply Lemma~\ref{lem:wei_concentration} to $\abold$ and show that $C_1,C_2,C_3$ are constants, i.e., in both cases $\mu=\mu_1\textnormal{ or }\mu=\mu_2$, $p,K$ are constants.

    To show that the $\psi_2$-norm $K$ of $\abold[i,j]$ is a constant, note that $\mu$ is supported on a constant-bounded interval in both cases. $K=\Theta(1)$ follows directly from the fact that bounded random variables are subgaussian \cite{Vershynin_2018}.

    What remains is to show that there exists $p$ such that for a possibly small constant $c(K)$, inequality (\ref{eq:wei_singular_bound}) holds. To that end, we give upper bounds to the probability density functions (PDF) $f_{\mu_1}$ and $f_{\mu_2}$. Specifically, $\|f_\mu\|_\infty\le M$ for an absolute constant $M$ implies that
    \begin{equation}\label{eq:smallball_from_density}
        \sup_{u\in\mathbb R}\Pr[|\abold[i,j]-u|\le s]
        = \sup_{u\in\mathbb R}\int_{u-s}^{u+s} f_\mu(t)dt
        \le 2s\|f_\mu\|_\infty
        \le (2M)s
    \end{equation}
    where it suffices to take $u\in \mbbr$ since $\abold[i,j]$ is real-valued.
    Thus \eqref{eq:wei_singular_bound} holds with $p := 2M$.

    \smallskip\noindent\textbf{Case $\mu=\mu_1$.}
    Let $a:=\sqrt 3/2$. The uniform density on $[-a,a]$ is $f_0(t)=\frac{1}{2a}\mbbone_{\{|t|\le a\}}$,
    hence $\|f_0\|_\infty = \frac{1}{2a}=\frac{1}{\sqrt 3}$.
    The density of $X_1+Y_1+Z_1+W_1$ is the 4-fold convolution $f_{\mu_1}=f_0*f_0*f_0*f_0$. Note that $\|f_0\|_1=1$ since $f_0$ is a PDF.
    By repeatedly applying Young's convolution inequality $\|f*g\|_\infty\le \|f\|_\infty\cdot \|g\|_1$,
    $$
    \|f_{\mu_1}\|_\infty =\| (f_0*f_0*f_0)*f_0 \|_\infty
    \le \|f_0*f_0*f_0\|_\infty\cdot \|f_0\|_1
    \le\|f_0\|_\infty\cdot 1
    =\frac{1}{\sqrt 3}.
    $$

    \noindent\textbf{Case $\mu=\mu_2$.} Let $b:=\left(\frac{9}{2}\right)^{1/4}$ and define $U:=X_2Y_2$. Let $f_{X,Y}(x,y)=\frac{1}{4b^2}$ be the joint PDF of $X_2,Y_2$, for $|x|,|y|\le b$.

    To obtain a clean formula of the PDF $f_U(u)$, we use a change of variables. Let $(U,V)=(X_2Y_2,X_2)$. Then $x=v,y=u/v$. Notice the following Jacobian
    $$\left|\frac{\partial (x,y)}{\partial(u,v)}\right|=\begin{vmatrix}
\partial x/\partial u & \partial x/\partial v\\
\partial y/\partial u & \partial y/\partial v
\end{vmatrix}=\begin{vmatrix}
0 & 1\\
1/v & -u/v^2
\end{vmatrix}=\frac1{|v|},$$
which implies that $f_{U,V}(u,v)=f_{X,Y}(x,y)\cdot \left|\frac{\partial (x,y)}{\partial(u,v)}\right|=f_{X,Y}(v,\frac{u}{v})\cdot\frac1{|v|}$.
Hence,
$$f_U(u)=\int_{|v|\le b\textnormal{, }|u/v|\le b}f_{U,V}(u,v)dv=\int_{|v|\le b\textnormal{, }|u/v|\le b}f_{X,Y}\left(v,\frac uv\right)\frac1{|v|}dv=\int_{|v|\in[|u|/b,b]}\frac1{4b^2}\frac1{|v|}dv=\frac{1}{2b^2}\ln \frac{b^2}{|u|}$$
for $0<|u|\le b^2$,
and
$$\|f_U\|_2^2=\int_{-b^2}^{b^2}|f_U(u)|^2du=\frac{1}{2b^4}\int_{0}^{b^2}\ln^2\left(\frac{b^2}{u}\right)du=\frac{1}{2b^2}\int_0^{\infty}t^2e^{-t}dt=\frac1{b^2}$$
where we substitute $u=b^2e^{-t}$ and use the fact that $\Gamma(3)=\int_{0}^{\infty} t^2e^{-t}dt=2$.

Let $U':=Z_2W_2$. Note that $U$ and $U'$ are i.i.d. By applying Young's convolution inequality, we have
$$\|f_{\mu_2}\|_\infty=\|f_U*f_{U'}\|_\infty\le \|f_U\|_2\cdot \|f_{U'}\|_2=\frac1{b^2}=\frac{\sqrt 2}{3}$$

In both cases, $M:=\frac{\sqrt 3}{3}\ge \|f_\mu\|_\infty$. By (\ref{eq:smallball_from_density}), we have
$$\sup_{u\in\mathbb R}\Pr[|\abold[i,j]-u|\le s]\le ps$$
for $p=\frac{2\sqrt 3}{3}$ and $s\le c(K)/p$.
By Lemma~\ref{lem:wei_concentration}, $C_1,C_2,C_3$ are constants.
\end{proof}

\section{Lower bound to set disjointness for a full range of redundancy}
\label{app:sd_full_range}

This appendix presents the LLM-generated proof obtained by prompting GPT-5.5 pro after the CCC version of this paper. The proof follows the forcing argument of Section~\ref{sec:set disjointness} with a different nemesis distribution, and establishes an $r\cdot t=\Omega(n^2)$ trade-off for every $r\in[n,n^2/1024]$, where $t$ denotes the amortized probe complexity. Since counting set intersection determines set disjointness, this result asymptotically subsumes both Corollary~\ref{thm:csi_lb} and Theorem~\ref{thm:sd_lb}, up to constant factors in the parameter ranges. We use the notation from Section~\ref{sec:set disjointness}. Our proof also follows a similar but simpler forcing argument as in Section~\ref{sec:set disjointness}

\paragraph{Nemesis input distribution.} Let $n$ be sufficiently large and even, and fix $r\in[n,n^2/1024]$. Set
$$b:=\left\lceil\frac{64r}{n}\right\rceil,\qquad d:=\left\lfloor\frac nb\right\rfloor,$$
and, for every $j\in[b]$, let $W_j$ be the fixed interval
$$W_j:=\{(j-1)d+1,\dots,jd\}\subseteq[n].$$
Let $\mcala$ denote the following input distribution on $A=(A_1,\dots,A_n)$:
\begin{itemize}
    \item For every $u\in[n/2]$ and $w\in[n]$, independently set $A_u(w)=1$ with probability $1-2^{-1/d}$.
    \item For every $v\in[n/2]$, fix $A_{n/2+v}=W_{1+((v-1)\bmod b)}$.
\end{itemize}
Let $V_1:=\{1,\dots,n/2\}$ and $V_2:=\{n/2+1,\dots,n\}$. By our construction, different indices in $V_2$ may represent the same set. We use the same sequence of queries $Q$ as in Section~\ref{sec:set disjointness}, namely, the concatenation of the $n/2$ shifted bijections between $V_1$ and $V_2$.

\begin{theorem}[Lower bound to set disjointness for a full range of redundancy]
    \label{thm:sd_lb_full_range}
    Given a (possibly randomized) succinct and systematic data structure with $r\in[n,n^2/1024]$ bits of redundancy that, given as input $n$ subsets of $[n]$, answers the sequence of queries $Q$ correctly with $\ge 99/100$ probability, its total number of probes is $\Omega(n^4/r)$ in the worst case. Equivalently, it requires $\Omega(n^2/r)$ probes per query on average.
\end{theorem}

As in the proof of Theorem~\ref{thm:sd_lb}, we define a set of good inputs on which a large proportion of the answers to the first $b$ sets in $V_2$ are $0$.

\begin{definition}
    \label{def:sd_full_good_input}
    Let $\msfa$ be the set of good inputs $A=(A_1,\dots,A_n)$ satisfying
    $$\bigl|\{(u,j)\in[n/2]\times[b]:Ans_{u,n/2+j}=0\}\bigr|\ge nb/8,$$
    where $Ans_{u,n/2+j}$ is a binary bit that is $0$ iff $A_u\cap A_{n/2+j}= \emptyset$.
\end{definition}

\begin{lemma}
    \label{lem:sd_full_good_input}
    For every sufficiently large $n$, a random input from $\mcala$ is good with probability at least $99/100$.
\end{lemma}
\begin{proof}
    For every $u\in[n/2]$ and $j\in[b]$, since $A_{n/2+j}=W_j$,
    $$\Pr_{A\leftarrow\mcala}[Ans_{u,n/2+j}=0]=(2^{-1/d})^d=\frac12.$$
    These $nb/2$ answers depend on disjoint groups of independent input bits, so they are independent uniform bits. The expected number of zero answers is $nb/4$. By the Chernoff bound,
    $$\Pr_{A\leftarrow\mcala}[A\notin\msfa]\le\exp(-nb/32)<1/100,$$
    where the last inequality follows from $nb\ge64r\ge64n$.
\end{proof}

\begin{proof}[Proof to Theorem~\ref{thm:sd_lb_full_range}]
    Again, by Yao's minimax principle, we may focus on deterministic data structures that answer $Q$ correctly with probability at least $99/100$ given a random input from $\mcala$.

    Following Lemma~\ref{lem:reduction_ds_qws}, partition $Q$ into blocks of $m:=nb=\Theta(r)$ consecutive queries, starting from the beginning of $Q$, and omit the final incomplete block, if any. Each block consists of $2b$ complete shifted bijections.

    Since the sets in $V_2$ repeat as $W_1,\dots,W_b$, each block queries every $A_u$, $u\in V_1$, against all these $b$ sets. Indeed, among the $2b$ partners of $A_u$, at least $b$ are consecutive without crossing from $A_n$ back to $A_{n/2+1}$, and these include every $W_j$. Thus every block contains the same answers $Ans_{u,n/2+j}$, $u\in[n/2]$, $j\in[b]$, up to ordering and repetitions. It suffices to prove the conditional joint min-entropy bound for the first block.

    The query-with-sketch algorithm for each block receives the updated redundancy at its beginning as an arbitrary $r$-bit sketch of the input. Its success probability is at least $99/100$, since correctness on $Q$ implies correctness on each block. By Lemma~\ref{lem:min-entr}, it suffices to show that, for the first block $(u_1,v_1),\dots,(u_m,v_m)$ and every partial assignment $\sigma$ of positive probability and length at most $q:=\lfloor n^2/32\rfloor$, the conditional joint min-entropy of the answers on good inputs is $>2r$. We set $\msfa_\sigma=\{0,1\}^{n\times n}$ for every $\sigma$, so $\Pr[A\in\msfa\backslash\msfa_\sigma\mid\sigma]=0$, and Lemma~\ref{lem:sd_full_good_input} gives the required probability of good inputs.

    Fix such a partial assignment $\sigma$ and any possible answer vector $(y_1,\dots,y_m)$ for the block on a good input from $\mcala$ consistent with $\sigma$. By Definition~\ref{def:sd_full_good_input}, this vector determines at least $nb/8$ distinct answers $Ans_{u,n/2+j}$ that are $0$. Each such answer forces all the $d$ input bits $A_u(w)$, $w\in W_j$, to be $0$. Since the groups of input coordinates $\{(u,w):w\in W_j\}$ are pairwise disjoint, at least $nbd/8$ distinct input bits are forced to be $0$. At most $q$ of these bits are revealed by $\sigma$. The unrevealed bits remain independent conditioned on $\sigma$, each equal to $0$ with probability $2^{-1/d}$. Hence,
    \begin{align*}
        &\Pr[(Ans_{u_i,v_i})_{i\in[m]}=(y_i)_{i\in[m]},\ A\in\msfa\mid\sigma]\\
        &\qquad\le (2^{-1/d})^{nbd/8-q}
        =2^{-nb/8+q/d}\le2^{-nb/16},
    \end{align*}
    where the last inequality uses $d=\lfloor n/b\rfloor\ge n/(2b)$ and $q\le n^2/32$. We conclude that
    $$\hmin(Ans_{u_1,v_1},\dots,Ans_{u_m,v_m},A\in\msfa\cap\msfa_\sigma\mid\sigma)\ge\frac{nb}{16}\ge4r>2r.$$

    Therefore, by Lemma~\ref{lem:min-entr}, each complete block requires $\Omega(q)=\Omega(n^2)$ probes on average under $\mcala$. There are $\lfloor n/(4b)\rfloor=\Theta(n^2/r)$ complete blocks, so the total average-case probe complexity is
    $$\Omega(n^2)\cdot\left\lfloor\frac{n}{4b}\right\rfloor=\Omega(n^4/r).$$
    This implies the claimed worst-case lower bound. Since $|Q|=n^2/4$, the amortized probe complexity is $\Omega(n^2/r)$.
\end{proof}

The upper endpoint of the redundancy range must be a sufficiently small constant multiple of $n^2$. Indeed, $Q$ contains $n^2/4$ Boolean answers, so a data structure with $r\ge n^2/4$ can store all of them in its redundancy and answer $Q$ without any input probes.

The above lower bound also applies to $(2-\varepsilon)$-approximate APSP through the same reduction as in Section~\ref{sec:set disjointness}.

\begin{corollary}[Lower bound to $(2-\varepsilon)$-approximate APSP]
    \label{cor:apsp_lb_full_range}
    Fix $\varepsilon>0$. Given a (possibly randomized) succinct and systematic data structure with $r\in[n/2,n^2/8192]$ bits of redundancy that, given as input a graph $G$ on $n$ vertices, $(2-\varepsilon)$-approximates all pairs shortest paths distances correctly with $\ge99/100$ probability, its total number of probes is $\Omega(n^4/r)$ in the worst case. Equivalently, it requires $\Omega(n^2/r)$ probes per query on average.
\end{corollary}

\begin{proof}
    We prove the lower bound through a reduction from set disjointness, as in the proof of Theorem~\ref{thm:apsp_lb}. Suppose there exists a data structure $D$ of bounded complexity for $(2-\varepsilon)$-approximate APSP.

    Given an input $A=(A_1,\dots,A_n)$ from $\mcala$, construct a graph $G=(V,E)$ on $2n$ vertices, where $V=V_1\cup W\cup V_2$ and $W$ is a disjoint copy of $[n]$. Thus $2|V_1|=|W|=2|V_2|=n$. For every $u\in V_1$ and $w\in[n]$, connect $u$ to the $w$-th vertex in $W$ if and only if $A_u(w)=1$. Analogously, for every $v\in V_2$ and $w\in[n]$, connect $v$ to the $w$-th vertex in $W$ if and only if $A_v(w)=1$. These are all the edges of $G$.

    For each query $(u,v)$ from $Q$, ask $D$ about the shortest path distance between $u\in V_1$ and $v\in V_2$. If the answer lies in $[2,4)$, answer that $|A_u\cap A_v|\ge1$; otherwise, answer that $|A_u\cap A_v|=0$.

    To see correctness, by construction, $d_G(u,v)=2$ if and only if $|A_u\cap A_v|\ge1$. Otherwise, $d_G(u,v)\ge4$, since $G$ is bipartite with parts $V_1\cup V_2$ and $W$. A $(2-\varepsilon)$-approximation returns a value less than $4$ in the first case and at least $4$ in the second case, so it determines the set disjointness answer.

    Every probe to the adjacency matrix of $G$ can be answered either directly from the construction or by one probe to an input bit $A_u(w)$. We therefore obtain a data structure for set disjointness with the same redundancy, at most the same probe complexity, and the same success guarantee. By Theorem~\ref{thm:sd_lb_full_range}, answering $Q$ requires $\Omega(n^4/r)$ probes for $r\in[n,n^2/1024]$. The constructed graph has $2n$ vertices, so rescaling $n$ gives the claimed lower bound. For other numbers of vertices, padding with at most three isolated vertices gives the stated range $r\in[n/2,n^2/8192]$ for sufficiently large $n$.
\end{proof}

\end{document}